\documentclass[12pt,psfig]{article}
\usepackage{multirow}
\usepackage{longtable}
\usepackage{bibentry}
\usepackage{amsfonts}
\usepackage{amssymb}
\usepackage{amsthm}
\usepackage[namelimits]{amsmath}
\usepackage{graphicx,epsfig}
\usepackage{amscd}
\usepackage{epsfig,color}
\usepackage{tikz}
\usepackage{pgflibraryarrows}
\usepackage{pgflibrarysnakes}

\date{}

\newcommand{\de}{\delta}

\newcommand{\ga}{\gamma}

\newcommand{\al}{\alpha}
\newcommand{\be}{\beta}

\newcommand{\ld}{\lambda}
\newcommand{\Ld}{\Lambda}
\newcommand{\ba}{\begin{array}}
\newcommand{\ea}{\end{array}}
\newcommand{\beqq}{\begin{equation*}}
\newcommand{\eeqq}{\end{equation*}}
\newcommand{\beq}{\begin{equation}}
\newcommand{\eeq}{\end{equation}}
\newcommand{\bdm}{\begin{displaymath}}
\newcommand{\edm}{\end{displaymath}}

\theoremstyle{definition}
\newtheorem{theorem}{Theorem}[section]
\newtheorem{definition}{Definition}[section]
\newtheorem{lemma}{Lemma}[section]

\newtheorem{remark}{Remark}[section]
\newtheorem{example}{Example}[section]

\numberwithin{equation}{section}
\begin{document}
\pagestyle{plain}

\begin{center}
{\large \bf  Chaotic polynomial maps}
\\ [0.3in]
XU ZHANG \footnote{ Email address:\ \ xuzhang08@gmail.com (X. Zhang).}

 \vspace{0.15in}
{\it  Department of Mathematics, Michigan State University,
East Lansing, MI 48824, USA}
\end{center}

\vspace{0.18in}

\baselineskip=20pt

{\bf\large{ Abstract.}}\  This paper introduces a class of polynomial maps in Euclidean spaces, investigates the conditions under which there exist Smale horseshoes and uniformly hyperbolic invariant sets, studies the chaotic dynamical behavior and strange attractors, and shows that some maps are chaotic in the sense of Li-Yorke or Devaney. This type of maps includes both the Logistic map and the H\'{e}non map. For some maps in three-dimensional spaces under certain conditions, if the expansion dimension is equal to one or two, it is shown that there exist a Smale horseshoe and a uniformly hyperbolic invariant set on which the system is topologically conjugate to the two-sided fullshift on finite alphabet; if the system is expanding, then it is verified that there is an forward invariant set on which the system is topologically semi-conjugate to the one-sided fullshift on eight symbols. For three types of high-dimensional polynomial maps with degree two, the existence of Smale horseshoe and the uniformly hyperbolic invariant sets
are studied, and it is proved that the map is topologically conjugate to the two-sided fullshift on finite alphabet on the invariant set under certain conditions. Some interesting maps with chaotic attractors and positive Lyapunov exponents in three-dimensional spaces are found  by using computer simulations. In the end, two examples are provided to illustrate the theoretical results.


\bigskip

\baselineskip=20pt
\section{Introduction}

The development of chaos theory dates back to Poincar\'{e}'s work on the three-body problem \cite{Poincare90, Poincare93}. Thanks to the discovery of the Lorenz attractor and Li-Yorke chaos \cite{LiYorke, Lorenz, Sparrow}, the chaos theory is gradually becoming popular. Chaos theory is being successfully applied in many fields,  from natural science to engineering. For example, it is an important branch of dynamical
systems in mathematics \cite{Wiggins1990},  planetary
orbits \cite{MurrayHolman1999} and fluid motion \cite{EckmannRuelle1985} in physics,  the analysis of chemical
compounds \cite{Strizhak2003} in chemistry, secure communication \cite{FengTse2007} and
chaos control \cite{ChenDong1993} in
engineering, new types of econometric
models \cite{Kelsey1998} and job selection \cite{PryorBright2007} in economics.

The well-known chaotic polynomial maps are the Logistic map and the H\'{e}non map, which were brought forward by May \cite{May} and H\'{e}non \cite{henon}, respectively. These examples are important models in the research of nonlinear dynamics because of their simple expressions and complicated dynamical behavior. In engineering, chaotic polynomial maps
were used as random number generators and have many applications \cite{Tsuneda2000}, such as spread spectrum
communications \cite{KohdaOokuboIshii1998},
and cryptosystems containing image encoding \cite{BelkhoucheQidwai2003} and data encryption \cite{Jessa2000}. Since the encoding methods are dependent on the properties of the chaotic maps, and the expressions for the polynomial maps are relatively easy, the polynomial maps in our present work might be very useful in the future applications.

The H\'{e}non map or generalized H\'{e}non map is an important model in the research of two-dimensional polynomial diffeomorphic maps. Devaney and Nitecki investigated the
conditions under which the real quadratic H\'{e}non map has a hyperbolic invariant set on which it is topologically conjugate to the two-sided fullshift on two symbols \cite{devaneynitecki}. Friedland and Milnor showed that the polynomial diffeomorphic map from the real or complex plane to itself is either conjugate to a composition of generalized H\'{e}non maps or dynamically trivial \cite{friedlandmilnor}. Dullin and Meiss studied the conditions which imply that the real cubic H\'{e}non map has hyperbolic invariant set \cite{dullinmeiss}.
The deeply relationship between the dynamical behavior and the change of parameters for the real H\'{e}non map was described clearly based on the work contributed by Benedicks, Carleson, Viana, Young and et al. \cite{BenedicksCarleson, BenedicksViana, BenedicksYoung}.
Bedford and Smillie applied the techniques in complex dynamics to show the existence of a hyperbolic horseshoe and a quadratic tangency between stable and unstable manifolds of fixed points for the real H\'{e}non map under certain conditions, where the parameter is on the boundary of the set, in which the map is uniformly hyperbolic on its non-wandering set \cite{bedfsmil, bedfsmil2}.
Cao et al. investigated the existence of
 the uniformly hyperbolic set and an orbit of tangency for the H\'{e}non-like families of diffeomorphisms of real plane by using the methods in real dynamics \cite{CaoLuzzatoRios2008}.
The hyperbolicity, ergodicity, Lyapunov exponents, topological entropy and so on, were also discussed for the complex H\'{e}non maps (see \cite{BedfordSmillieVI, BedfordSmillieVII, BedfordSmillieVIII} and references therein).
In \cite{Zhang2010}, we studied the existence of Smale horseshoe and uniformly hyperbolic invariant set on which the map is topologically conjugate to the two-sided fullshift on finite alphabet for the generalized H\'{e}non maps.

There are many interesting chaotic attractors have been found in the study of the three-dimensional polynomial maps. Many authors investigated the following three-dimensional map $F:(x,y,z)\in\mathbb{R}^3\to(f_1,f_2,f_3)\in\mathbb{R}^3:$
\begin{equation}\label{sys5}
\left\{
  \begin{array}{ll}
    f_1(x,y,z)=y \\
    f_2(x,y,z)=z \\
    f_3(x,y,z)=M_1+Bx+M_2y-z^2
  \end{array}
\right.
\end{equation}
where $M_1$, $M_2$, $B$ are parameters. In \cite{GonchenkoOvsyannikov}, Gonchenko et.al. studied the wild Lorenz-type strange attractors of system \eqref{sys5}.  In \cite{GonchenkoMeiss}, Gonchenko et.al. studied the existence of wild-hyperbolic strange attractors and homoclinic bifurcations of system \eqref{sys5}. In \cite{GonchenkoShilnikov}, Gonchenko et.al. investigated the bifurcations of system \eqref{sys5} with non-transverse heteroclinic cycles. In \cite{LiYang}, Li and Yang investigated the three-dimensional Smale horseshoe and gave a computer assisted verification of the existence of hyperchaos of a class of three-dimensional polynomial maps. In \cite{ElhadjSprott}, Elhadj and Sprott determined
all the possible forms of the three-dimensional quadratic diffeomorphisms with constant Jacobian. In \cite{FournierLopez}, Fournier-Prunaret et.al.
studied the bifurcation and chaotic dynamics of a kind of three-dimensional maps of logistic type. In \cite{GonchenkoLi2008}, Gonchenko and Li studied the existence of the uniformly hyperbolic invariant sets topologically conjugating to the Smale horseshoe for two cases of three-dimensional quadratic maps with constant Jacobian. In \cite{Sprott1993, Sprott2004}, Sprott gave the idea that the computer simulations play an important role in the search of the chaotic attractors and studied several types of polynomial maps.

The polynomial diffeomorphic maps in high-dimensional spaces are different from those in two-dimensional spaces, since the degree of the inverse maps might be a large number and it is difficult to find the explicit expressions for the inverse maps. For more examples, please refer to the maps in Section 3 or related references \cite{Moser1994}. There exists a bound on the degree of the inverse, that is, $\mbox{degree}(F^{-1})\leq(\mbox{degree}(F))^{n-1}$ \cite{ChengWangYu}. To investigate the inverse of the polynomial maps, an interesting problem is the well-known Jacobian Conjecture on $\mathbb{R}^n$, that is, is every polynomial map $F:\mathbb{R}^n\to\mathbb{R}^n$ with non-zero constant $\mbox{det}F'(x)$ a bijective map with a polynomial inverse? This is called the real Jacobian Conjecture, if the map is defined on $\mathbb{C}^n$, then this is the complex Jacobian Conjecture. The important progress on this problem includes Moh's proof for $n=2$ and $\mbox{degree}(F)\leq100$ \cite{Moh1983}, Wang's study for all $n$ and $\mbox{degree}(F)\leq2$ \cite{Wang1980}, Bass, Connell, and Wright reduced the problem to a special case, and obtained that the maps of homogeneous type of degree exactly three can imply the Jacobian Conjecture \cite{BassConnellWright}, and so on.

We introduce the following type of maps: 
\begin{equation}\label{newmap}
\left\{
\begin{array}{ll}
f_1(x_1,...,x_n)=a_{11}p_1(x_1)+\sum_{2\leq j\leq n} a_{1j}x_j\\
\vdots \qquad  \qquad \vdots \qquad  \qquad \vdots \qquad  \qquad \vdots\\
f_k(x_1,...,x_n)=a_{kk}p_k(x_k)+\sum_{1\leq j\leq n,\ j\neq k}a_{kj}x_{j}\\
\vdots \qquad  \qquad \vdots \qquad  \qquad \vdots \qquad  \qquad \vdots\\
f_n(x_1,...,x_n)=a_{nn}p_n(x_n)+\sum_{1\leq j\leq n-1}a_{nj}x_j
\end{array} \right.,
\end{equation}
where $p_i(x_i)$, $1\leq i\leq n$, are real polynomials, $a_{ij}$, $1\leq i,j\leq n$, are real parameters. We study the case that $n\geq3$. For $n=1,2$, some results can be found in \cite{Zhang2010, zhang2012}.

First, we study the following type of three-dimensional polynomial maps $F:(x,y,z)\in\mathbb{R}^3\to(f_1,f_2,f_3)\in\mathbb{R}^3:$
\begin{equation}\label{sys1}
\left\{
  \begin{array}{ll}
    f_1(x,y,z)=a_1p(x)+a_2y+a_3z \\
    f_2(x,y,z)=b_1x+b_2q(y)+b_3z \\
    f_3(x,y,z)=c_1x+c_2y+c_3r(z),
  \end{array}
\right.
\end{equation}
where $a_i,b_i,c_i$ are real parameters, $1\leq i\leq3$, $p(x)$, $q(y)$, and $r(z)$ are real polynomials.

(1). We investigate a type of diffeomorphism with polynomial inverse by assuming that $a_1\neq0$, $b_2=c_3=0$, and $p(x)$ having two different non-negative real zeros, and verify the existence of a Smale horseshoe and a uniformly hyperbolic invariant set on which the system is topologically conjugate to the two-sided fullshift on two symbols under certain conditions (See Theorems \ref{onepositive-zero}--\ref{onepositive}). If $p(x)$ has only $m$ simple real roots, then there is a uniformly hyperbolic invariant set on which the system is topologically conjugate to the two-sided fullshift on $m$ symbols under certain conditions (see Theorem \ref{onepositive-n-zero}).

(2). We study a class of diffeomorphism with polynomial inverse,  where $a_1b_2\neq0$, $c_3=0$, $p(x)$ and $q(y)$ are assumed to have two distinct non-negative real roots, and show that there exist a Smale horseshoe and a uniformly hyperbolic invariant set on which the system is topologically conjugate to the two-sided fullshift on four symbols under certain conditions (See Theorems \ref{twopositive-zero}-- \ref{twopositive-1zero-1positive}).
If $p(x)$ has only $m$ simple real roots, and $q(y)$ has only $n$ simple real zeros, then there is a uniformly hyperbolic invariant set on which the system is topologically conjugate to the two-sided fullshift on $mn$ symbols under certain conditions (see Theorem \ref{twopositive-nzero}).

(3). We study a kind of maps, where $a_1b_2c_3\neq0$, $p(x)$, $q(y)$, and $r(z)$ have two different non-negative real roots, and prove that there is an forward invariant set on which the system is topologically semi-conjugate to the one-sided fullshift on eight symbols (See Theorems \ref{threepositive-zero}--\ref{threepositive-2zero-1positive}).

Second, we investigate the following type of maps : $F:(x_1,...,x_n)\in\mathbb{R}^n\to(f_1,...,f_n)\in\mathbb{R}^n$:
\begin{equation}\label{generalmap}
\left\{
\begin{array}{ll}
f_1(x_1,...,x_n)=a_{11}-x_1^2+\sum_{2\leq j\leq n} a_{1j}x_j\\
\vdots \qquad  \qquad \vdots \qquad  \qquad \vdots \qquad  \qquad \vdots\\
f_d(x_1,...,x_n)=a_{dd}-x_d^2+\sum_{1\leq j\leq n,\ j\neq d}a_{dj}x_{j}\\
f_{d+1}(x_1,...,x_n)=\sum_{1\leq j\leq n,\ j\neq d+1}a_{d+1,j}x_{j}\\
\vdots \qquad  \qquad \vdots \qquad  \qquad \vdots \qquad  \qquad \vdots\\
f_n(x_1,...,x_n)=\sum_{1\leq j\leq n-1}a_{nj}x_j
\end{array} \right.,
\end{equation}
where $a_{ij}$, $1\leq i,j\leq n$, are real parameters. We study three kinds of polynomial maps by giving the inverse expressions, and show that there exist Smale horseshoe and uniformly hyperbolic invariant sets on which the maps are topologically conjugate with fullshift on finite alphabet (see Theorem \ref{hyp-highd-dim}).

We show that some maps are chaotic in the sense of Li-Yorke or Devaney. We also apply the computer simulation methods to study the existence of chaotic attractors and calculate the maximal Lyapunov exponents (see Section 5).

The rest of the paper is organized as follows. In Section 2, we introduce some basic concepts and results. In Section 3, we investigate the existence of uniformly hyperbolic invariant sets and complicated dynamics of maps \eqref{sys1} in three-dimensional spaces. We split this section into three parts. In Section 4, we study the existence of uniformly hyperbolic invariant sets and the chaotic dynamics of the maps \eqref{generalmap} in high-dimensional Euclidean spaces. In Section 5, some interesting maps with chaotic attractors and positive Lyapunov exponents are obtained by applying computer simulations. In Section 6, several examples are given to illustrate the theoretical results.

\bigskip

\section{Preliminaries}

In this section, the concepts of uniformly hyperbolic set, symbolic dynamics, and chaos, and related results are introduced.

\begin{definition} \cite{Newhouse, Robinson}  Consider a map $F:\mathbb{R}^n\to\mathbb{R}^n$.
An invariant set $\Lambda\subset\mathbb{R}^n$ is called a uniformly hyperbolic set for $F$, if there are constants $C>0$, $\lambda>1$, and a
continuous splitting $T_w\mathbb{R}^n=E^u_w\oplus E^s_w$, $w\in\Lambda$, such that
\begin{itemize}
\item [1]. $DF_w(E^s_w)=E^s_{F(w)}$ and $DF_w(E^u_w)=E^u_{F(w)}$;
\item [2].
$|DF^m_w(v)|\geq C\lambda^m|v|$, $v\in E^u_w$, $m\geq0$;
\item [3]. $|DF^{-m}_w(v)|\geq C\lambda^m|v|$, $v\in E^s_w$, $m\geq0$.
\end{itemize}
\end{definition}

Given a subspace $E\subset\mathbb{R}^n$ with $0 < \mbox{dim}E<n$ and $\mathbb{R}^n=E\oplus E^{c}$, where $E^c$ is the algebraic complement of $E$, the
standard unit cone with respect to $E$ is given by:
\beqq
K_1(E, E^c)=\bigg\{
\mathbf{v} =
\left(
\begin{array}{c}
v_1\\
v_2
\end{array}
\right)
:\ v_1\in E,\ v_2\in E^{c},\ \mbox{and}\ |v_2|\leq|v_1|\bigg\}.
\eeqq

Given a linear automorphism
$T : \mathbb{R}^n\to\mathbb{R}^n$ with $T(E)=E$, the image $T(K_1(E,E^c))$ is called a cone in $\mathbb{R}^n$ with core $E$,  denoted by $\mathcal{C}(E)$.
For convenience, $\mathcal{C}$ is used to represent a cone in  $\mathbb{R}^n$, where the proper subspace $E$ of $\mathbb{R}^n$ is omitted.

Suppose that $F$ is a diffeomorphism on $\mathbb{R}^n$ and $U$ is a compact subset of $\mathbb{R}^n$, then $\Lambda=\cap^{\infty}_{-\infty}F^i(U)$ is a compact invariant set. Let the cone field $\mathcal{C}=\{\mathcal{C}_w\}$ on $\Lambda$ be the set of cones $\mathcal{C}_w\subset  T\mathbb{R}^n=\mathbb{R}^n$, $w\in\Lambda$. It is said that the cone field $\mathcal{C}_w$ has constant orbit core dimension on $\Lambda$ if
$\mbox{dim}E_w=\mbox{dim}E_{F(w)}$, $w\in\Lambda$, where $E_w$ and $E_{F(w)}$ are the cores of $\mathcal{C}_w$ and
$\mathcal{C}_{F(w)}$, respectively.

Set
\beqq
m_{\mathcal{C},w}=m_{\mathcal{C},w}(F):=\inf_{\mathbf{v}\in\mathcal{C}_w\setminus\{0\}}\frac{|DF_w(\mathbf{v})|}{|\mathbf{v}|}
\eeqq
and
\beqq
m'_{\mathcal{C},w}=m'_{\mathcal{C},w}(F):=\inf_{\mathbf{v}\not\in\mathcal{C}_{F(w)}}\frac{|DF^{-1}_{F(w)}(\mathbf{v})|}{|\mathbf{v}|},
\eeqq
where $m_{\mathcal{C},w}$ is said to be the minimal expansion of $F$ on $\mathcal{C}_w$, and $m'_{\mathcal{C},w}$
is called the minimal co-expansion of $F$ on $\mathcal{C}_w$.

\begin{lemma} (\cite{Newhouse}, Theorem 1.4)\label{unifhyp}
A necessary and sufficient condition for  $\Lambda$ to be a uniformly hyperbolic set for $F$ is that there are an integer $N>0$ and a cone field $\mathcal{C}$
with constant orbit core dimension over  $\Lambda$ such that $F^N$ is both expanding
and co-expanding on $\mathcal{C}$.
\end{lemma}

Next, we introduce the symbolic dynamics \cite{Robinson}. Set $S_0:=\{1,2,...,m\},\ m\geq2$. Let
$$\textstyle\sum_m:=\{\al=(...,a_{-2},a_{-1},a_0,a_1,a_2,...):\ a_i\in S_0,\ \ i\in\mathbb{Z}\}$$
be the two-sided sequence space, where the distance on $\sum_m$ is defined by
$ d(\al,\beta)=\sum^{\infty}_{i=-\infty}\frac{d(a_i,b_i)}{2^{|i|}},$
 for $\al=(...,a_{-2},a_{-1},a_0,a_1,a_2,...)$, $\beta=(...,b_{-2},b_{-1},b_0,b_1,b_2,...)\in\sum_m$, and
$d(a_i,b_i)=1$ if $a_i\neq b_i$, and $d(a_i,b_i)=0$ if $a_i=b_i$,
$i\in\mathbb{Z}$. The shift map $\sigma:\sum_m\to\sum_m$ is defined by
$\sigma(\al)=(...,b_{-2},b_{-1},b_0,b_1,b_2...)$ for any $\al=(...,a_{-2},a_{-1},a_0,a_1,a_2...)\in\sum_m$, where $b_{i}=a_{i+1}$ for any $i\in\mathbb{Z}$. We call $(\sum_m,\sigma)$ is the two-sided symbolic dynamical
system on $m$ symbols, or simply two-sided fullshift on $m$ symbols.

Two definitions of chaos are introduced, which will be useful
in the sequel.

\begin{definition}\cite{LiYorke} Let $(X,d)$ be a metric space,
$F:X\rightarrow X$ a map, and $S$ a subset of $X$ with at least two
points. Then $S$ is called a scrambled set of $F$ if for any two
distinct points $x,y\in S,$
$$\liminf_{n\to\infty}d(F^{n}(x),F^{n}(y))=0,\ \limsup_{n\to\infty}d(F^{n}(x),F^{n}(y))>0.$$
The map $F$ is said to be chaotic in the sense of Li-Yorke if there
exists an uncountable scrambled set $S$ of $F$.
\end{definition}

Let $(X,d)$ be a metric space. A continuous map $F:X\to X$ is said to be topologically transitive
if for any two open subsets $U$ and $V$ of $X$, there exists a positive integer $m$ such
that $F^m(U)\cap V\neq\emptyset$; $F$ is said
to have sensitive dependence on initial conditions in $X$ if there exists
a positive constant $\de$ such that for any point $x\in X$ and any neighborhood $U$ of $x$,
there exist $y\in U$ and a positive integer $m$ such that $d(F^m(x),F^m(y))>\de$.

\begin{definition} \cite{devaney} Let $(X,d)$ be a
metric space. A map $F: V \subset X \to V$ is said to be chaotic on
$V$ in the sense of Devaney if
\begin{itemize}
\item [\rm {(i)}] the set of the periodic points of $F$ is dense in $V$;
\item [\rm {(ii)}] $F$ is topologically transitive in $V$;
\item [\rm {(iii)}] $F$ has sensitive dependence on initial conditions in $V$.
\end{itemize}
\end{definition}

In the above definition, condition $\rm{(iii)}$ is redundant if $F$
is continuous in $V$ by the result of \cite{banksbrooks}, where $V$ is an infinite set.

\begin{lemma} \cite[Theorem 2.2]{ShiJuChen} \label{chaoticshift}
The fullshift on finite alphabet is chaotic in the sense of both of Li-Yorke and Devaney.
\end{lemma}

\bigskip


\section{Smale horseshoe in three-dimensional space}

In this section, we study the parameter regions of the maps \eqref{sys1} such that there exist invariant sets on which the maps are uniformly hyperbolic and topologically conjugate with fullshift on finite alphabet.

For map \eqref{sys1}, consider the situations that the polynomials $p(x)$, $q(y)$, and $r(z)$ have at least two distinct real zeros.

The polynomial $p$ can be written as
$$p(x)=(x-\al_1)^{m_1}\cdots(x-\al_{r_1})^{m_{r_1}}(y^2+\be_{r_1+1}y+\ga_{r_1+1})^{m_{r_1+1}}\cdots(y^2+\be_{s_1}y+\ga_{s_1})^{m_{s_1}},$$
where $$m_i\geq1,1\leq i\leq s_1;\
\sum^{r_1}_{i=1}m_i+\sum^{s_1}_{i=r_1+1}2m_i=d_1;\ \be^2_j-4\ga_j<0,\ r_1+1\leq
j\leq s_1,$$ and $\al_1,...,\al_{r_1}$ are real zeros of $p(x)$ with
$\al_1<\cdots<\al_{r_1}$.

The polynomial $q(y)$ can be represented as follows:
$$q(y)=(y-\xi_1)^{n_1}\cdots(y-\xi_{r_2})^{n_{r_2}}(y^2+\eta_{r_2+1}y+\zeta_{r_2+1})^{n_{r_2+1}}\cdots(y^2+\eta_{s_2}y+\zeta_{s_2})^{n_{s_2}},$$
where $$n_i\geq1,1\leq i\leq s_2;\
\sum^{r_2}_{i=1}n_i+\sum^{s_2}_{i=r_2+1}2n_i=d_2;\ \eta^2_j-4\zeta_j<0,\ r_2+1\leq
j\leq s_2,$$ and $\xi_1,...,\xi_{r_2}$ are real zeros of $q(y)$ with
$\xi_1<\cdots<\xi_{r_2}$.

The polynomial $r(z)$ can be expressed as follows:
$$r(z)=(z-\kappa_1)^{l_1}\cdots(z-\kappa_{r_3})^{l_{r_3}}(z^2+\mu_{r_3+1}z+\nu_{r_3+1})^{l_{r_3+1}}\cdots(z^2+\mu_{s_3}z+\nu_{s_3})^{l_{s_3}},$$
where $$l_i\geq1,1\leq i\leq s_3;\
\sum^{r_3}_{i=1}l_i+\sum^{s_3}_{i=r_3+1}2l_i=d_3;\ \mu^2_k-4\nu_k<0,\ r_3+1\leq
k\leq s_3,$$ and $\kappa_1,...,\kappa_{r_3}$ are real zeros of $r(z)$ with
$\kappa_1<\cdots<\kappa_{r_3}$.

It is easy to obtain that
$$\frac{p'(x)}{p(x)}=\sum^{r_1}_{i=1}\frac{m_i}{x-\al_{i}}+
\sum^{s_1}_{i=r_1+1}\frac{m_i(2x+\be_i)}{x^2+\be_ix+\ga_i},$$
$$\frac{q'(y)}{q(y)}=\sum^{r_2}_{i=1}\frac{n_i}{y-\xi_{i}}+
\sum^{s_2}_{i=r_2+1}\frac{n_i(2y+\eta_i)}{y^2+\eta_iy+\zeta_i}.$$
$$\frac{r'(z)}{r(z)}=\sum^{r_3}_{i=1}\frac{l_i}{z-\kappa_{i}}+
\sum^{s_3}_{i=r_3+1}\frac{l_i(2z+\mu_i)}{z^2+\mu_iz+\nu_i}.$$
So,
\beq\label{polyineq1}
\lim_{x\to\al^{+}_{i_0}}\frac{p'(x)}{p(x)}=\infty,\ \lim_{x\to\al^{-}_{i_0+1}}\frac{p'(x)}{p(x)}=-\infty;
\eeq
\beq\label{polyineq2}
\lim_{y\to\xi^{+}_{j_0}}\frac{q'(y)}{q(y)}=\infty,\ \lim_{y\to\xi^{-}_{j_0+1}}\frac{q'(y)}{q(y)}=-\infty;
\eeq
\beq\label{polyineq3}
\lim_{z\to\kappa^{+}_{k_0}}\frac{r'(z)}{r(z)}=\infty,\ \lim_{z\to\kappa^{-}_{k_0+1}}\frac{r'(z)}{r(z)}=-\infty.
\eeq

\begin{lemma}\label{monotone} \cite[Lemma 3.1]{Zhang2010}
\begin{itemize}
\item [(i)] Suppose that there exists $i_0$, $1\leq i_0<r_1$, such that
$\al_{i_0}\geq0$. Then, there exist $\delta_1,\delta'_1\in(\al_{i_0},\al_{i_0+1})$
such that $p^{(m_{i_0})}(\al_{i_0})p'(x)>0$ for all $x\in(\al_{i_0},\delta_1]$,
$p^{(m_{i_0})}(\al_{i_0})p'(x)<0$ for all $x\in[\delta'_1,\al_{i_0+1})$.
\item [\rm {(ii)}] Suppose that there exists $j_0$, $1\leq j_0<r_2$, such that
$\xi_{j_0}\geq0$. Then, there exist $\delta_2,\delta'_2\in(\xi_{j_0},\xi_{j_0+1})$
such that $q^{(n_{j_0})}(\xi_{j_0})q'(y)>0$ for all $y\in(\xi_{j_0},\delta_2]$,
$q^{(n_{j_0})}(\xi_{j_0})q'(y)<0$ for all $y\in[\delta'_2,\xi_{j_0+1})$.
\item [\rm {(iii)}] Suppose that there exists $k_0$, $1\leq k_0<r_3$, such that
$\kappa_{k_0}\geq0$. Then, there exist $\delta_3,\delta'_3\in(\kappa_{k_0},\kappa_{k_0+1})$
such that $r^{(l_{k_0})}(\kappa_{k_0})r'(z)>0$ for all $z\in(\kappa_{k_0},\delta_3]$,
$r^{(l_{k_0})}(\kappa_{k_0})r'(z)<0$ for all $z\in[\delta'_3,\kappa_{k_0+1})$.
\end{itemize}
\end{lemma}
\medskip

\subsection{The maps with the dimension of the unstable subspace equals to one}

In this subsection, we investigate the polynomial maps \eqref{sys1} in the following form, $F:(x,y,z)\in\mathbb{R}^3\to(f_1,f_2,f_3)\in\mathbb{R}^3:$
\begin{equation}\label{sys3}
\left\{
  \begin{array}{ll}
    f_1(x,y,z)=a_1p(x)+a_2y+a_3z \\
    f_2(x,y,z)=b_1x \\
    f_3(x,y,z)=c_2y
  \end{array}
\right.,
\end{equation}
where $a_1,a_2,a_3,b_1,c_2$ are real parameters and $a_3b_1c_2\neq0$.

The derivative of \eqref{sys3} is
 \begin{equation}\label{derivativesys3}
DF=\left(
  \begin{array}{ccc}
    a_1p'(x) & a_2 & a_3 \\
    b_1 & 0 & 0 \\
    0 & c_2 &0\\
  \end{array}
\right).
\end{equation}
The inverse of $F$ is $H:(x,y,z)\in\mathbb{R}^3\to(h_1,h_2,h_3)\in\mathbb{R}^3:$
\begin{equation}\label{inversesys3}
\left\{
  \begin{array}{ll}
    h_1(x,y,z)=\frac{y}{b_1}\\
    h_2(x,y,z)=\frac{z}{c_2} \\
    h_3(x,y,z)=\frac{1}{a_3}\bigg(x-a_1p(y/b_1)-\frac{a_2}{c_2}z\bigg)
  \end{array}
\right.
\end{equation}
and the derivative of $H$ is
\begin{equation}\label{inversederivativesys3}
DF^{-1}=\left(
  \begin{array}{ccc}
0& \frac{1}{b_1} &0 \\
0 & 0 & \frac{1}{c_2} \\
    \frac{1}{a_3} & -\frac{a_1}{b_1a_3}p'(y/b_1) & -\frac{a_2}{a_3c_2} \\
\end{array}
\right).
\end{equation}
So, the determinant of the Jacobian of this type of maps is constant.

To start our work, an example is introduced as follows:
\begin{equation*}
\left\{
  \begin{array}{ll}
    f_1(x,y,z)=a_1x(1-x)+a_2y+a_3z \\
    f_2(x,y,z)=b_1x \\
    f_3(x,y,z)=c_2y
  \end{array}
\right..
\end{equation*}
This example can be thought of as the generalization of the H\'{e}non map in three-dimensional spaces.

\begin{theorem}\label{onepositive-zero}
Suppose that there are $1\leq i_0<r_1$ such that $0=\al_{i_0}<\al_{i_0+1}$ and $m_{i_0}=m_{i_0+1}=1$, that is, there are two distinct non-negative zeros of $p(x)$ with the multiplicity equals to one. If
\beq\label{assumption-1-zero}
\max\{a_3c_2b_1\al_{i_0+1},
a_2b_1\al_{i_0+1},a_2b_1\al_{i_0+1}+a_3c_2b_1\al_{i_0+1}\}\leq0\ \mbox{and}\ a_1p'(0)>0,
\eeq
then, for fixed $a_2$, $a_3$, $b_1$, and $c_2$ with $|c_2|<1$, and sufficiently large $|a_1|$, there exist a Smale horseshoe and a uniformly hyperbolic invariant set $\Lambda$ for the map \eqref{sys3}, on which $F$ is topologically conjugate to two-sided fullshift on two  symbols. Consequently, they are chaotic in the sense of both Li-Yorke and Devaney.
\end{theorem}

For convenient, assume that $[a,b]$ and $[b,a]$ represent the same interval. Set $$U:=[0,\al_{i_0+1}]\times
[0,b_1\al_{i_0+1}]\times
[0,c_2b_1\al_{i_0+1}],\ \lambda:=\frac{1}{|c_2|}.$$

\begin{lemma}\label{connectedtwo-zero}
Under the assumptions of Theorem \ref{onepositive-zero}, $F(U)\cap U$ and $F^{-1}(U)\cap U$ have two non-empty connected components, respectively.
\end{lemma}

\begin{proof}
First, we need to determine the image of the vertices of the cubic $U$ under $F$. The vertices of $U$ are
\beqq
A=(\al_{i_0+1},0,0),
B=(\al_{i_0+1},b_1\al_{i_0+1},0),
C=(0,b_1\al_{i_0+1},0),
D=(0,0,0),
\eeqq
\beqq
A'=(\al_{i_0+1},0,c_2b_1\al_{i_0+1}),
B'=(\al_{i_0+1},b_1\al_{i_0+1},c_2b_1\al_{i_0+1}),
\eeqq
\beqq
C'=(0,b_1\al_{i_0+1},c_2b_1\al_{i_0+1}),
D'=(0,0,c_2b_1\al_{i_0+1}).
\eeqq
The images of these vertices under $F$ are
\beqq
F(A)=(0,b_1\al_{i_0+1},0),
F(B)=(a_2b_1\al_{i_0+1},b_1\al_{i_0+1},c_2b_1\al_{i_0+1}),
\eeqq
\beqq
F(C)=(a_2b_1\al_{i_0+1},0,c_2b_1\al_{i_0+1}),
F(D)=(0,0,0),
\eeqq
\beqq
F(A')=(a_3c_2b_1\al_{i_0+1},b_1\al_{i_0+1},0),
F(B')=(a_2b_1\al_{i_0+1}+a_3c_2b_1\al_{i_0+1},b_1\al_{i_0+1},c_2b_1\al_{i_0+1}),
\eeqq
\beqq
F(C')=(a_2b_1\al_{i_0+1}+a_3c_2b_1\al_{i_0+1},0,c_2b_1\al_{i_0+1}),
F(D')=(a_3c_2b_1\al_{i_0+1},0,0).
\eeqq

From the map \eqref{sys3}, it follows that the transformation is linear with respect to the $y$ and $z$ variables. For the $x$ variable, it is a parabolic-like function for $x\in[0,\al_{i_0+1}]$. By \eqref{assumption-1-zero} and the calculation of the vertices above, one has that if $|a_1|$ is sufficiently large, then there are two connected components of $U\cap F(U)$.

On the other hand, the image of the vertices under $F^{-1}=H$ are
\beqq
H(A)=(0,0,\al_{i_0+1}/a_3),
H(B)=(\al_{i_0+1},0,\al_{i_0+1}/a_3),
H(C)=(\al_{i_0+1},0,0),
H(D)=(0,0,0),
\eeqq
\beqq
H(A')=(0,b_1\al_{i_0+1},\al_{i_0+1}(1-a_2b_1)/a_3),
H(B')=(\al_{i_0+1},b_1\al_{i_0+1},\al_{i_0+1}(1-a_2b_1)/a_3),
\eeqq
\beqq
H(C')=(\al_{i_0+1},b_1\al_{i_0+1},-a_2b_1\al_{i_0+1}/a_3),
H(D')=(0,b_1\al_{i_0+1},-a_2b_1\al_{i_0+1}/a_3).
\eeqq
Since $F$ is diffeomorphic, $U\cap F^{-1}(U)=F^{-1}(U\cap F(U))$ has two connected components. This completes the proof.
\end{proof}

\begin{lemma}\label{hyponeexp-singel}
Under the assumptions of Theorem \ref{onepositive-zero}, the invariant set $\Lambda$ is uniformly hyperbolic.
\end{lemma}

\begin{proof}
By (i) of Lemma \ref{monotone}, assume that $|a_1|$ is large enough such that for $w=(x,y,z)\in U\cap F^{-1}(U)$, one has that $x\in[0,\delta_1]\cup[\delta'_1,\al_{i_0+1}]$.

Since $m_{i_0}=m_{i_0+1}=1$, set $$M_0:=\min_{x\in[0,\delta_1]\cup[\delta'_1,\al_{i_0+1}]}|p'(x)|>0.$$ Suppose that
\beq\label{bound-deriv-1-single}
|a_1|\geq\max\bigg\{\frac{\lambda+|a_2|+|a_3|}{M_0},\ \frac{\lambda|a_3b_1c_2|+|b_1c_2|+|a_2b_1|}{M_0|c_2|}\bigg\}.
\eeq

Now, it is to show that for any $w_0\in U\cap F^{-1}(U)$, one has that $m_{\mathcal{C},w_0}\geq\lambda$.

We introduce the unit cone:
\beqq
K_1(\mathbb{R}, \mathbb{R}^2)=\bigg\{\mathbf{v}=
\left(
\begin{array}{c}
v_0\\
u_0
\end{array}
\right)
:\ v_0\in \mathbb{R},\ u_0=(u_{01},u_{02})^T\in\mathbb{R}^2,\ \mbox{and}\ |u_0|\leq|v_0|\bigg\},
\eeqq
where $|u_0|=\max\{|u_{01}|,|u_{02}|\}$ and $|\mathbf{v}|=\max\{|v_0|,|u_0|\}$.

Denote $
\left(
\begin{array}{c}
v_1\\
u_1
\end{array}
\right)
=DF_{w_0}
\left(
\begin{array}{c}
v_0\\
u_0
\end{array}
\right)$, where $
\left(
\begin{array}{c}
v_0\\
u_0
\end{array}
\right)\in K_1(\mathbb{R},\mathbb{R}^2)$.
By \eqref{derivativesys3},
\beqq
\left\{
  \begin{array}{ll}
v_{1}=a_1p'(x_0)v_0+a_2u_{01}+a_3u_{02}\\
u_{11}=b_1v_{0}\\
u_{12}=c_2u_{01}
  \end{array},
\right.
\eeqq
this, together with \eqref{bound-deriv-1-single}, implies that $$|v_{1}|\geq |v_{0}||a_1p'(x_0)|-|a_2||u_{01}|-|a_3||u_{02}|\geq
 |v_{0}||a_1p'(x_0)|-|a_2||v_{0}|-|a_3||v_{0}|\geq
 \lambda|v_{0}|.$$ Hence, $m_{\mathcal{C},w_0}\geq\lambda>1$.

Next, it is to show that for any point $w_0\in U\cap F(U)$,  $m'_{\mathcal{C},w_0}\geq\lambda$.
For any $
\left(
\begin{array}{c}
v_0\\
u_0
\end{array}
\right)\not\in K_1(\mathbb{R},\mathbb{R}^2)$, denote $
\left(
\begin{array}{c}
v_{-1}\\
u_{-1}
\end{array}
\right)
=DF^{-1}_{w_0}
\left(
\begin{array}{c}
v_0\\
u_0
\end{array}
\right)$, it follows from \eqref{inversederivativesys3} that
\beqq
\left\{
  \begin{array}{ll}
 v_{-1}=\frac{1}{b_1}u_{01} \\
  u_{-11} =\frac{1}{c_2}u_{02}\\
  u_{-12}=  \frac{1}{a_3}v_0-\frac{a_1}{b_1a_3}p'(y_0/b_1)u_{01} -\frac{a_2}{a_3c_2}u_{02}\\
\end{array}.
\right.
\eeqq
If $|u_0|=|u_{01}|$, then by \eqref{bound-deriv-1-single},
\beqq
\begin{split}
|u_{-1}|\geq&|u_{-12}|\geq -\frac{1}{|a_3|}|v_0|+\frac{1}{|b_1a_3|}|a_1p'(y_0/b_1)||u_{01}| -\frac{|a_2|}{|a_3c_2|}|u_{02}|\\
\geq&\frac{1}{|b_1a_3|}|a_1p'(y_0/b_1)||u_{0}|-\frac{1}{|a_3|}|u_0|-\frac{|a_2|}{|a_3c_2|}|u_{0}|\\
\geq&\lambda|u_0|;
\end{split}
\eeqq
if $|u_0|=|u_{02}|$, then $$|u_{-1}|\geq|u_{02}/c_2|=\lambda|u_{0,2}|=\lambda|u_0|.$$
Hence, $m'_{\mathcal{C},w_0}\geq\lambda$.

Thus, by Lemma \ref{unifhyp}, one has that the invariant set $\Lambda$ is uniformly hyperbolic. The proof is completed.
\end{proof}

Therefore, it follows from Lemmas \ref{chaoticshift}, \ref{connectedtwo-zero}, and \ref{hyponeexp-singel} that Theorem \ref{onepositive-zero} holds.

\begin{theorem}\label{onepositive}
Suppose that there are $1\leq i_0<r_1$ such that $0<\al_{i_0}<\al_{i_0+1}$, that is, there are two distinct positive zeros of $p(x)$. Set
\beqq
M_1:=\max\{a_2b_1\al_{i_0}+a_3c_2b_1\al_{i_0},
a_2b_1\al_{i_0}+a_3c_2b_1\al_{i_0+1},
a_2b_1\al_{i_0+1}+a_3c_2b_1\al_{i_0},
a_2b_1\al_{i_0+1}+a_3c_2b_1\al_{i_0+1}\},
\eeqq
\beqq
M_2:=\min\{a_2b_1\al_{i_0}+a_3c_2b_1\al_{i_0},
a_2b_1\al_{i_0}+a_3c_2b_1\al_{i_0+1},
a_2b_1\al_{i_0+1}+a_3c_2b_1\al_{i_0},
a_2b_1\al_{i_0+1}+a_3c_2b_1\al_{i_0+1}\}.
\eeqq
Consider the following two situations:
\beq\label{assumption-1-positive}
\mbox{(i).}\ M_1<\al_{i_0}\ \mbox{and}\ a_1p^{(m_{i_0})}(\al_{i_0})>0;\
\mbox{(ii).}\ M_2>\al_{i_0+1}\ \mbox{and}\ a_1p^{(m_{i_0})}(\al_{i_0})<0.
\eeq
Then, for fixed $a_2$, $a_3$, $b_1$, and $c_2$ with $|c_2|<1$, and sufficiently large $|a_1|$, there exist a Smale horseshoe and a uniformly hyperbolic invariant set $\Lambda$ for the map \eqref{sys3}, on which $F$ is topologically conjugate to two-sided fullshift on two  symbols. Consequently, they are chaotic in the sense of both Li-Yorke and Devaney.
\end{theorem}

For convenient, suppose that $[a,b]$ and $[b,a]$ represent the same interval. Denote $$U:=[\al_{i_0},\al_{i_0+1}]\times
[b_1\al_{i_0},b_1\al_{i_0+1}]\times
[c_2b_1\al_{i_0},c_2b_1\al_{i_0+1}].$$

Set $\lambda:=\frac{1}{|c_2|}$. It follows from the assumption $0<|c_2|<1$ that $\lambda>1$. The invariant set is given by
$$\Lambda:=\bigcap^{+\infty}_{-\infty}F^i(U).$$

\begin{lemma} \label{connectedtwo}
In Cases (i)-(ii) of Theorem \ref{onepositive},  for fixed $a_2$, $a_3$, $b_1$, and $c_2$, and sufficiently large $|a_1|$, one has that
$F(U)\cap U$ and $F^{-1}(U)\cap U$ have two non-empty connected components, respectively.
\end{lemma}

\begin{proof}
By the definition of map \eqref{sys3}, one has that the transformation is linear with respect to the $y$ and $z$ variables.

Next, it is to study the transformation with respect to the $x$ variable. First, it is to obtain the image of the vertices of $U$ under the map $F$. The vertices of $U$ are
\beqq
A=(\al_{i_0+1},b_1\al_{i_0},c_2b_1\al_{i_0}),
B=(\al_{i_0+1},b_1\al_{i_0+1},c_2b_1\al_{i_0}),
\eeqq
\beqq
C=(\al_{i_0},b_1\al_{i_0+1},c_2b_1\al_{i_0}),
D=(\al_{i_0},b_1\al_{i_0},c_2b_1\al_{i_0}),
\eeqq
\beqq
A'=(\al_{i_0+1},b_1\al_{i_0},c_2b_1\al_{i_0+1}),
B'=(\al_{i_0+1},b_1\al_{i_0+1},c_2b_1\al_{i_0+1}),
\eeqq
\beqq
C'=(\al_{i_0},b_1\al_{i_0+1},c_2b_1\al_{i_0+1}),
D'=(\al_{i_0},b_1\al_{i_0},c_2b_1\al_{i_0+1}).
\eeqq
The images under $F$ are
\beqq
F(A)=(a_2b_1\al_{i_0}+a_3c_2b_1\al_{i_0},b_1\al_{i_0+1},c_2b_1\al_{i_0}),
F(B)=(a_2b_1\al_{i_0+1}+a_3c_2b_1\al_{i_0},b_1\al_{i_0+1},c_2b_1\al_{i_0+1}),
\eeqq
\beqq
F(C)=(a_2b_1\al_{i_0+1}+a_3c_2b_1\al_{i_0},b_1\al_{i_0},c_2b_1\al_{i_0+1}),
F(D)=(a_2b_1\al_{i_0}+a_3c_2b_1\al_{i_0},b_1\al_{i_0},c_2b_1\al_{i_0}),
\eeqq
\beqq
F(A')=(a_2b_1\al_{i_0}+a_3c_2b_1\al_{i_0+1},b_1\al_{i_0+1},c_2b_1\al_{i_0}),
F(B')=(a_2b_1\al_{i_0+1}+a_3c_2b_1\al_{i_0+1},b_1\al_{i_0+1},c_2b_1\al_{i_0+1}),
\eeqq
\beqq
F(C')=(a_2b_1\al_{i_0+1}+a_3c_2b_1\al_{i_0+1},b_1\al_{i_0},c_2b_1\al_{i_0+1}),
F(D')=(a_2b_1\al_{i_0}+a_3c_2b_1\al_{i_0+1},b_1\al_{i_0},c_2b_1\al_{i_0}).
\eeqq
This, together with \eqref{assumption-1-positive}, implies that there are two connected components of $U\cap F(U)$ for sufficiently large $|a_1|$.

On the other hand, the image of the vertices of $U$ under $F^{-1}=H$ are
\beqq
H(A)=\bigg(\al_{i_0},b_1\al_{i_0},\frac{1}{a_3}(\al_{i_0+1}-a_2b_1\al_{i_0})\bigg),
H(B)=\bigg(\al_{i_0+1},b_1\al_{i_0},\frac{1}{a_3}(\al_{i_0+1}-a_2b_1\al_{i_0})\bigg),
\eeqq
\beqq
H(C)=\bigg(\al_{i_0+1},b_1\al_{i_0},\frac{1}{a_3}(\al_{i_0}-a_2b_1\al_{i_0})\bigg),
H(D)=\bigg(\al_{i_0},b_1\al_{i_0},\frac{1}{a_3}(\al_{i_0}-a_2b_1\al_{i_0})\bigg),
\eeqq
\beqq
H(A')=\bigg(\al_{i_0},b_1\al_{i_0+1},\frac{1}{a_3}(\al_{i_0+1}-a_2b_1\al_{i_0+1})\bigg),
H(B')=\bigg(\al_{i_0+1},b_1\al_{i_0+1},\frac{1}{a_3}(\al_{i_0+1}-a_2b_1\al_{i_0+1})\bigg),
\eeqq
\beqq
H(C')=\bigg(\al_{i_0+1},b_1\al_{i_0+1},\frac{1}{a_3}(\al_{i_0}-a_2b_1\al_{i_0+1})\bigg),
H(D')=\bigg(\al_{i_0},b_1\al_{i_0+1},\frac{1}{a_3}(\al_{i_0}-a_2b_1\al_{i_0+1})\bigg).
\eeqq
It follows from the fact $F$ is diffeomorphic that $U\cap F^{-1}(U)=F^{-1}(U\cap F(U))$ has two connected components. The proof is completed.
\end{proof}

\begin{lemma} \label{hyponeexp}
In Cases (i)-(ii) of Theorem \ref{onepositive},  for fixed $a_2$, $a_3$, $b_1$, and $c_2$, and sufficiently large $|a_1|$,  one has that the invariant set $\Lambda$ is uniformly hyperbolic.
\end{lemma}

\begin{proof}

In the following discussions, we always assume that $a_2$, $a_3$, $b_1$, and $c_2$ are fixed, and $|a_1|$ are sufficiently large.
By (i) of Lemma \ref{monotone}, assume that $|a_1|$ is sufficiently large such that for $w=(x,y,z)\in U\cap F^{-1}(U)$, one has that $x\in(\al_{i_0},\delta_1]\cup[\delta'_1,\al_{i_0+1})$.
For any point $w_0=(x_0,y_0,z_0)\in U$, set $w_1=F(w_0)=(x_1,y_1,z_1)$ and $w_{-1}=F^{-1}(w_0)=(x_{-1},y_{-1},z_{-1})$.

First, by \eqref{polyineq1} and Lemma \ref{monotone},  one has that for $w=(x,y,z)\in U\cap F^{-1}(U)$,
\beq\label{mono-ineq-1-i}
\mbox{Case (i)}:\  \frac{|a_1p'(x)|}{|a_1p(x)|}\geq
\frac{\max\bigg\{\lambda+|a_2|+|a_3|,\ \frac{\lambda|a_3b_1c_2|+|b_1c_2|+|a_2b_1|}{|c_2|}\bigg\}}{\al_{i_0}-M_1};
\eeq
\beq\label{mono-ineq-1-ii}
\mbox{Case (ii)}:\  \frac{|a_1p'(x)|}{|a_1p(x)|}\geq \frac{\max\bigg\{\lambda+|a_2|+|a_3|,
\frac{\lambda|a_3b_1c_2|+|b_1c_2|+|a_2b_1|}{|c_2|}\bigg\}}{M_2-\al_{i_0+1}}.
\eeq

For any point $w_0\in U\cap F^{-1}(U)$, one has
\beqq
\left\{
  \begin{array}{ll}
   x_1=a_1p(x_0)+a_2y_0+a_3z_0 \\
    y_1=b_1x_0 \\
    z_1=c_2x_0
  \end{array}
\right..
\eeqq
that is, $a_1p(x_0)=x_1-a_2y_0-a_3z_0$, implying that
\beq\label{bound-1-i}
\mbox{Case (i)}:\ \ a_1p(x_0)\geq\al_{i_0}-M_1;\ \mbox{Case (ii)}:\ \ -a_1p(x_0)\geq M_2-\al_{i_0+1}.
\eeq

For any point $w_0\in U\cap F(U)$, by \eqref{inversesys3}, one has,
\beqq
\left\{
  \begin{array}{ll}
    x_{-1}=\frac{y_0}{b_1} \\
    y_{-1}=\frac{z_0}{c_2} \\
    z_{-1}=\frac{1}{a_3}\bigg(x_0-a_1p(y_0/b_1)-\frac{a_2}{c_2}z_0\bigg)
  \end{array}
\right.,
\eeqq
that is,
$a_1p(x_{-1})=a_1p(y_0/b_1)=x_0-a_3z_{-1}-\frac{a_2}{c_2}z_0$,
yielding that
\beq\label{bound-1-ii}
\mbox{Case (i)}:\ \ a_1p(x_{-1})\geq\al_{i_0}-M_1;\  \mbox{Case (ii)}:\ \ -a_1p(x_{-1})\geq M_2-\al_{i_0+1}.
\eeq

Hence, it follows from \eqref{mono-ineq-1-i} and \eqref{bound-1-i} that for any point $w_0\in U\cap F^{-1}(U)$,
\beq\label{bound-deriv-1-i}
\mbox{Cases (i) and (ii)}:\ \ |a_1p'(x_0)|\geq \lambda+|a_2|+|a_3|.
\eeq

Now, it is to show that for any $w_0\in U\cap F(U)$, one has that $m_{\mathcal{C},w_0}\geq\lambda$.

Consider the unit cone
\beqq
K_1(\mathbb{R}, \mathbb{R}^2)=\bigg\{\mathbf{v}=
\left(
\begin{array}{c}
v_0\\
u_0
\end{array}
\right)
:\ v_0\in \mathbb{R},\ u_0=(u_{01},u_{02})^T\in\mathbb{R}^2,\ \mbox{and}\ |u_0|\leq|v_0|\bigg\},
\eeqq
where $|u_0|=\max\{|u_{01}|,|u_{02}|\}$ and $|\mathbf{v}|=\max\{|v_0|,|u_0|\}$.

Suppose that $
\left(
\begin{array}{c}
v_1\\
u_1
\end{array}
\right)
=DF_{w_0}
\left(
\begin{array}{c}
v_0\\
u_0
\end{array}
\right)$, where $
\left(
\begin{array}{c}
v_0\\
u_0
\end{array}
\right)\in K_1(\mathbb{R},\mathbb{R}^2)$.
By \eqref{derivativesys3},
\beqq
\left\{
  \begin{array}{ll}
v_{1}=a_1p'(x_0)v_0+a_2u_{01}+a_3u_{02}\\
u_{11}=b_1v_{0}\\
u_{12}=c_2u_{01}
  \end{array},
\right.
\eeqq
this, together with \eqref{bound-deriv-1-i}, implies that $$|v_{1}|\geq |v_{0}||a_1p'(x_0)|-|a_2||u_{01}|-|a_3||u_{02}|\geq
 |v_{0}||a_1p'(x_0)|-|a_2||v_{0}|-|a_3||v_{0}|\geq
 \lambda|v_{0}|.$$ Hence, $m_{\mathcal{C},w_0}\geq\lambda$.

By \eqref{mono-ineq-1-ii} and \eqref{bound-1-ii}, one has that for any point $w_0\in U\cap F(U)$,
\beq\label{bound-deriv-1-ii}
|a_1p'(x_{-1})|=|a_1p'(y_0/b_1)|\geq\frac{\lambda|a_3b_1c_2|+|b_1c_2|+|a_2b_1|}{|c_2|},
\eeq

Next, it is to show that for any point $w_0\in U\cap F(U)$,  $m'_{\mathcal{C},w_0}\geq\lambda$.
For any $
\left(
\begin{array}{c}
v_0\\
u_0
\end{array}
\right)\not\in K_1(\mathbb{R},\mathbb{R}^2)$, that is ,  $|u_0|>|v_0|$, suppose $
\left(
\begin{array}{c}
v_{-1}\\
u_{-1}
\end{array}
\right)
=DF^{-1}_{w_0}
\left(
\begin{array}{c}
v_0\\
u_0
\end{array}
\right)$, by \eqref{inversederivativesys3}, one has
\beqq
\left\{
  \begin{array}{ll}
 v_{-1}=\frac{1}{b_1}u_{01} \\
  u_{-11} =\frac{1}{c_2}u_{02}\\
  u_{-12}=  \frac{1}{a_3}v_0-\frac{a_1}{b_1a_3}p'(y_0/b_1)u_{01} -\frac{a_2}{a_3c_2}u_{02}\\
\end{array}.
\right.
\eeqq
If $|u_0|=|u_{02}|$, then $|u_{-1}|\geq|u_{02}/c_2|=\lambda|u_{0,2}|=\lambda|u_0|$. If $|u_0|=|u_{01}|$, then by \eqref{bound-deriv-1-ii},
\beqq
\begin{split}
|u_{-1}|\geq&|u_{-12}|\geq -\frac{1}{|a_3|}|v_0|+\frac{1}{|b_1a_3|}|a_1p'(y_0/b_1)||u_{01}| -\frac{|a_2|}{|a_3c_2|}|u_{02}|\\
\geq&\frac{1}{|b_1a_3|}|a_1p'(y_0/b_1)||u_{0}|-\frac{1}{|a_3|}|u_0|-\frac{|a_2|}{|a_3c_2|}|u_{0}|\\
\geq&\lambda|u_0|.
\end{split}
\eeqq
So, $m'_{\mathcal{C},w_0}\geq\lambda$.

Thus, it follows from Lemma \ref{unifhyp} that the invariant set $\Lambda$ is uniformly hyperbolic. This completes the proof.
\end{proof}

Therefore, by Lemmas \ref{chaoticshift}, \ref{connectedtwo}, and \ref{hyponeexp}, one has that Theorem \ref{onepositive} holds.

\begin{theorem}\label{onepositive-n-zero}
For system \eqref{sys3}, suppose that $p(x)=\prod^m_{i=1}(x-\al_i)$, where $\al_1<\al_{2}<\cdots<\al_m$ are real numbers, and $m\geq2$. For any fixed $a_3$, $b_1$, and $c_2$ with $|c_2|<1$, and sufficiently large $|a_1|$, there exist a Smale horseshoe and a hyperbolic invariant set on which $F$ is topologically conjugate to the two-sided fullshift on $m$ symbols. Consequently, $F$ is chaotic in the sense of both Li-Yorke and Devaney.
\end{theorem}

\begin{proof}
Denote $\ld:=\frac{1}{|c_2|}$. By the properties of $p(x)$, there are positive constants $\eta_i$, $1\leq i\leq m$, such that the sets $V_i=[\al_i-\eta_i, \al_i+\eta_i]$, $1\leq i\leq m$, are pairwise disjoint, and $|p'(x)|>0$ for all $x\in V_i$, $1\leq i\leq m$.

Set
$$U:=[\al_1-\eta_1,\al_m+\eta_m]\times[b_1(\al_1-\eta_1),b_1(\al_m+\eta_m)]\times[c_2b_1(\al_1-\eta_1),c_2b_1(\al_m+\eta_m)],$$ and
$$\Ld:=\bigcap^{+\infty}_{-\infty}F^{j}(U).$$
Set
$$M_0:=\inf_{x\in\cup^m_{i=1}V_i}|p'(x)|,\ \ M_1:=\min_{1\leq i\leq m}|p(\al_i\pm\eta_i)|;$$
\beqq
N_0:=\max\bigg\{\frac{\lambda+|a_2|+|a_3|}{M_0},\ \frac{\lambda|a_3b_1c_2|+|b_1c_2|+|a_2b_1|}{M_0|c_2|}\bigg\}.
\eeqq
\beqq
N_1:=\frac{(1+|a_2b_1|+|a_3c_2b_1|)\max\{|\al_1-\eta_1|,|\al_n+\eta_n|\}}{M_1}.
\eeqq

If $|a_1|>N_0$, by applying similar discussions in the proof of Lemma \ref{hyponeexp-singel}, one has that
the invariant set $\Ld$ is uniformly hyperbolic, where the constant $N_0$ is similar with the constant specified in \eqref{bound-deriv-1-single}.

If $|a_1|>N_1$, by applying simple calculation and  the fact that $p(x)$ has simple real roots, one has that $F(U)\cap U$ has $m$ connected components.

Hence, if $|a_1|>\max\{N_0,N_1\}$, then there exist a Smale horseshoe and a hyperbolic invariant set on which $F$ is topologically conjugate to the two-sided fullshift on $m$ symbols. This, together with Lemma \ref{chaoticshift}, yields that $F$ is chaotic in the sense of both Li-Yorke and Devaney. This completes the proof.
\end{proof}

\medskip

\subsection{The maps with the dimension of the unstable subspace equals to two}

In this subsection, we consider a subclass of the maps \eqref{sys1} satisfying that $a_3=c_2=c_3=0$, $a_2b_3\neq0$, $c_1=1$, and $a_2b_3\neq0$, which are represented as follows: $F:(x,y,z)\in\mathbb{R}^3\to(f_1,f_2,f_3)\in\mathbb{R}^3:$
\begin{equation} \label{sys6}
\left\{
  \begin{array}{ll}
    f_1(x,y,z)=a_1p(x)+a_2y \\
    f_2(x,y,z)=b_1x+b_2q(y)+b_3z \\
    f_3(x,y,z)=x
  \end{array}
\right..
\end{equation}
The derivative of $F$ and the determinant are
\begin{equation}\label{positivederivative}
DF=\left(
    \begin{array}{ccc}
      a_1p'(x) & a_2 & 0 \\
      b_1 & b_2q'(y) & b_3 \\
      1 & 0 & 0 \\
    \end{array}
  \right)
\end{equation}
and $a_2b_3\neq0$, respectively.
The inverse of $F$ is $H:(x,y,z)\in\mathbb{R}^3\to(h_1,h_2,h_3)\in\mathbb{R}^3:$
\begin{equation}\label{inversesys6}
\left\{
  \begin{array}{ll}
    h_1(x,y,z)=z \\
    h_2(x,y,z)=\frac{1}{a_2}(x-a_1p(z)) \\
    h_3(x,y,z)=\frac{1}{b_3}(y-b_1z-b_2q(\frac{1}{a_2}(x-a_1p(z))))
  \end{array}
\right.
\end{equation}
and the derivative of $H$ is
\begin{equation}\label{inversederivative}
DF^{-1}=\left(
  \begin{array}{ccc}
    0 & 0 & 1 \\
    \frac{1}{a_2} & 0 & -\frac{a_1}{a_2}p'(z) \\
    -\frac{b_2}{a_2b_3}q'(\frac{1}{a_2}(x-a_1p(z))) & \frac{1}{b_3} &-\frac{b_1}{b_3}+\frac{a_1b_2}{a_2b_3}q'(\frac{1}{a_2}(x-a_1p(z)))p'(z) \\
  \end{array}
\right).
\end{equation}
The determinant is $1/a_2b_3$.
Hence, polynomial maps \eqref{sys6} are diffeomorphisms with constant Jacobian and polynomial inverse.

To begin our study, a toy model is introduced as follows:
\begin{equation*}
\left\{
  \begin{array}{ll}
    f_1(x,y,z)=a_1x(1-x)+a_2y \\
    f_2(x,y,z)=b_1x+b_2y(1-y)+b_3z \\
    f_3(x,y,z)=x
  \end{array}
\right..
\end{equation*}
This model can be thought of as another generalization of the H\'{e}non map in three-dimensional spaces.

Now, we show that for certain parameters, there exist uniformly hyperbolic invariant sets and chaotic dynamics.

\begin{theorem}\label{twopositive-zero}
Suppose that there are $1\leq i_0<r_1$ and $1\leq j_0<r_2$ such that $0=\al_{i_0}<\al_{i_0+1}$, $0=\xi_{j_0}<\xi_{j_0+1}$, and $m_{i_0}=m_{i_0+1}=n_{j_0}=n_{j_0+1}=1$, that is, there are two distinct non-negative zeros of $p(x)$ and $q(y)$. If
\beq\label{assumption-2-zero}
a_2<0\ \mbox{and}\ a_1p'(0)>0,\ \mbox{and}\ b_1\leq0,\ b_3<0,\ \mbox{and}\ b_2q'(0)>0,
\eeq
then, for fixed $a_2$, $b_1$, and $b_3$, and sufficiently large $|a_1|$ and $|b_2|$, there are a Smale horseshoe for the map \eqref{sys6}, and a uniformly hyperbolic invariant set $\Lambda$ on which $F$ is topologically conjugate to two-sided fullshift on four  symbols. Therefore, they are Li-Yorke chaotic as well as Devaney chaotic.
\end{theorem}

Consider the compact set $$U:=[0,\al_{i_0+1}]\times[0,\xi_{j_0+1}]\times[0,\al_{i_0+1}].$$

Take a constant $\lambda>1$. The invariant set is denoted by
$$\Lambda:=\bigcap^{+\infty}_{-\infty}F^i(U).$$

\begin{lemma}\label{connectedfour-zero}
Under the assumptions of Theorem \ref{twopositive-zero}, one has that
$F(U)\cap U$ and $F^{-1}(U)\cap U$ have four non-empty connected components, respectively.
\end{lemma}

\begin{proof}
It is to show that $F(U)\cap U$ has four connected components. We will describe how to think of the image of $F(U)$.

First, it is to find the position of the image of the vertices of the cubic $U$, that is, the image of $A,B,C,D,A',B',C',D'$ under $F$. We only need to know the relative position of $F(A)$, $F(B)$, $F(C)$, $F(D)$, $F(A')$, $F(B')$, $F(C')$, $F(D')$.
The vertices of the cubic $U$ are:
\beqq
A=(\al_{i_0+1},0,0),
B=(\al_{i_0+1},\xi_{j_0+1},0),
C=(0,\xi_{j_0+1},0),
D=(0,0,0),
\eeqq
\beqq
A'=(\al_{i_0+1},0,\al_{i_0+1}),
B'=(\al_{i_0+1},\xi_{j_0+1},\al_{i_0+1}),
C'=(0,\xi_{j_0+1},\al_{i_0+1}),
D'=(0,0,\al_{i_0+1});
\eeqq
the image of the vertices under $F$ are
\beqq
F(A)=(0,b_1\al_{i_0+1},\al_{i_0+1}),\ F(B)=(a_2\xi_{j_0+1},b_1\al_{i_0+1},\al_{i_0+1}),
F(C)=(a_2\xi_{j_0+1},0,0),
\eeqq
\beqq
F(D)=(0,0,0),F(A')=(0,(b_1+b_3)\al_{i_0+1},\al_{i_0+1}),\ F(B')=(a_2\xi_{j_0+1},(b_1+b_3)\al_{i_0+1},\al_{i_0+1}),
\eeqq
\beqq
F(C')=(a_2\xi_{j_0+1},b_3\al_{i_0+1},0),\
F(D')=(0,b_3\al_{i_0+1},0).
\eeqq

Set
$$AD:=\{(x,y,z):\ x\in[0,\al_{i_0+1}],y=0,z=0\},$$
$$AB:=\{(x,y,z):\ x=\al_{i_0+1},y\in[0,\xi_{j_0+1}],z=0\},$$
$$AA':=\{(x,y,z):\ x=\al_{i_0+1},y=0,z\in[0,\al_{i_0+1}]\},$$
$$ABCD:=\{(x,y,z):\ (x,y)\in[0,\al_{i_0+1}]\times[0,\xi_{j_0+1}],z=0\}.$$

Second, it is to determine the image of the plane $ABCD$ by using the expression \eqref{sys6}.
By \eqref{sys6}, $F(AD)$ is a parabola along the $x$-axis, and $F(AB)$ is a parabola along $y$-axis. The image of $ABCD$ under $F$ can be thought of as a movement of the parabola $F(AD)$ along another parabola $F(AB)$, but the direction of $F(AD)$ should not change too much.

Finally, it is to study the graph of $F(U)$. By \eqref{sys6}, $F(AA')$ is a line segment. We push forward the surface $F(ABCD)$ obtained in the previous step along $F(AA')$. The graph of $F(U)$ comes out!

Hence, there are four non-empty connected components of $F(U)\cap U$ for sufficiently large $|a_1|$ and $|b_2|$.

Now, it is to study $F^{-1}(U)\cap U$. The image of the vertices of $U$ under $F^{-1}=H$ are
\beqq
H(A)=\bigg(0,\frac{\al_{i_0+1}}{a_2},\frac{-b_2q(\al_{i_0+1}/a_2)}{b_3}\bigg),
H(B)=\bigg(0,\frac{\al_{i_0+1}}{a_2},\frac{(\xi_{j_0+1}-b_2q(\al_{i_0+1}/a_2))}{b_3}\bigg),
\eeqq
\beqq
H(C)=(0,0,\xi_{j_0+1}/b_3),
H(D)=(0,0,0),
\eeqq
\beqq
H(A')=\bigg(\al_{i_0+1},\frac{\al_{i_0+1}}{a_2},\frac{(-b_1\al_{i_0+1}-b_2q(\al_{i_0+1}/a_2))}{b_3}\bigg),
\eeqq
\beqq
H(B')=\bigg(\al_{i_0+1},\frac{\al_{i_0+1}}{a_2},\frac{(\xi_{j_0+1}-b_1\al_{i_0+1}-b_2q(\al_{i_0+1}/a_2))}{b_3}\bigg),
\eeqq
\beqq
H(C')=\bigg(\al_{i_0+1},0,\frac{(\xi_{j_0+1}-b_1\al_{i_0+1})}{b_3}\bigg),
H(D')=\bigg(\al_{i_0+1},0,\frac{-b_1\al_{i_0+1}}{b_3}\bigg).
\eeqq
By \eqref{inversesys6}, one has that for fixed $a_2$ and $b_3$, and sufficiently large $|a_1|$ and $|b_2|$, given any $(x_0,y_0,z_0)\in U$,
\beqq
\left\{
  \begin{array}{ll}
    x_0=z \\
    y_0=\frac{1}{a_2}(x-a_1p(z)) \\
    z_0=\frac{1}{b_3}(y-b_1z-b_2q(\frac{1}{a_2}(x-a_1p(z))))
  \end{array}
\right.,
\eeqq
if there is a solution of the above equations, then there should exist four different solutions by Lemma \ref{monotone}. Further, it follows from $F^{-1}(U)\cap U=F^{-1}(U\cap F(U))$, the fact that $F$ is a diffeomorphism and the above geometric description of $F(U)$ that $F^{-1}(U)\cap U$ has four non-empty connected components. The proof is completed.
\end{proof}

\begin{lemma} \label{hyptwoexp-zero}
Under the assumptions of Theorem \ref{twopositive-zero}, the invariant set $\Lambda$ is uniformly hyperbolic.
\end{lemma}

\begin{proof}
Fix $a_2$, $b_1$, and $b_3$, it follows from (i) and (ii) of Lemma \ref{monotone} that we could assume $|a_1|$ and $|b_2|$ are sufficiently large such that for $w=(x,y,z)\in U\cap F^{-1}(U)$, one has that $x\in[\al_{i_0},\delta_1]\cup[\delta'_1,\al_{i_0+1}]$
and $y\in[\xi_{j_0},\delta_2]\cup[\delta'_2,\xi_{j_0+1}]$.
For any point $w_0=(x_0,y_0,z_0)\in U$, set $w_1=F(w_0)=(x_1,y_1,z_1)$ and $w_{-1}=F^{-1}(w_0)=(x_{-1},y_{-1},z_{-1})$.

Set $$M_0:=\min\bigg\{\min_{x\in[0,\delta_1]\cup[\delta'_1,\al_{i_0+1}]}|p'(x)|,\min_{y\in[0,\delta_2]\cup[\delta'_2,\xi_{j_0+1}]}|q'(y)|\bigg\}.$$
Suppose that
\beq\label{bound-deriv-2-single-zero}
|a_1|\geq\frac{\max\{\lambda+|a_2|,1+|a_2b_3|\lambda+|a_2|(1+|b_1|)\}}{M_0},\ |b_2|\geq\frac{\lambda+|b_1|+|b_3|}{M_0}.
\eeq

Now, it is to show that for any point $w_0\in U\cap F^{-1}(U)$, $m_{\mathcal{C},w_0}\geq\lambda$.

Introduce the unit cone
\beqq
K_1(\mathbb{R}^2, \mathbb{R})=\bigg\{\mathbf{v}=
\left(
\begin{array}{c}
v_0\\
u_0
\end{array}
\right)
:\ v_0=(v_{01},v_{02})^T\in \mathbb{R}^2,\ u_0\in\mathbb{R},\ \mbox{and}\ |u_0|\leq|v_0|\bigg\},
\eeqq
where $|v_0|=\max\{|v_{01}|,|v_{02}|\}$ and $|\mathbf{v}|=\max\{|v_0|,|u_0|\}$.

For any $
\left(
\begin{array}{c}
v_0\\
u_0
\end{array}
\right)\in K_1(\mathbb{R}^2,\mathbb{R})$, $
\left(
\begin{array}{c}
v_1\\
u_1
\end{array}
\right)
=DF_{w_0}
\left(
\begin{array}{c}
v_0\\
u_0
\end{array}
\right)$.
By \eqref{positivederivative},
\beqq
\left\{
  \begin{array}{ll}
v_{11}=a_1p'(x_0)v_{01}+a_2v_{02}\\
v_{12}=b_1v_{01}+b_2q'(y_0)v_{02}+b_3u_0\\
u_1=v_{01}
  \end{array}.
\right.
\eeqq
By \eqref{bound-deriv-2-single-zero},
if $|v_0|=|v_{01}|$, then $$|v_{11}|\geq |v_{01}||a_1p'(x_0)|-|a_2||v_{01}|\geq\lambda|v_{01}|;$$ if $|v_0|=|v_{02}|$, then $$|v_{12}|\geq|v_{02}||b_2q'(y_0)|-|b_1||v_{02}|-|b_3||u_0|\geq\lambda|v_{02}|,$$ which implies that $|(v_1,u_1)|\geq\lambda|(v_0,u_0)|$.  Hence, $m_{\mathcal{C},w_0}\geq\lambda$.

Next, it is to show that $m'_{\mathcal{C},w_0}\geq\lambda$. for any point $w_0\in U\cap F(U)$, take any
$\left(
\begin{array}{c}
v_0\\
u_0
\end{array}
\right)
\not\in K_1(\mathbb{R}^2,\mathbb{R})$, that is ,  $|u_0|>|v_0|$, suppose $
\left(
\begin{array}{c}
v_{-1}\\
u_{-1}
\end{array}
\right)=DF^{-1}_{w_0}
\left(
\begin{array}{c}
v_0\\
u_0
\end{array}
\right)$, by \eqref{inversederivative}, one has
\beqq
\left\{
  \begin{array}{ll}
 v_{-11}=u_0 \\
  v_{-12} =\frac{1}{a_2} v_{01}-\frac{a_1}{a_2}p'(z_0) u_0\\
  u_{-1}=  -\frac{b_2}{a_2b_3}q'(\frac{1}{a_2}(x_0-a_1p(z_0)))v_{01}+\frac{1}{b_3}v_{02} -\frac{b_1}{b_3}u_0+\frac{a_1b_2}{a_2b_3}q'(\frac{1}{a_2}(x_0-a_1p(z_0)))p'(z_0) u_0\\
\end{array}.
\right.
\eeqq
So, by \eqref{bound-deriv-2-single-zero},
\beqq
\begin{split}
|u_{-1}|&\geq \bigg|\frac{a_1b_2}{a_2b_3}q'\bigg(\frac{1}{a_2}(x_0-a_1p(z_0))\bigg)p'(z_0)\bigg| |u_0|-\bigg|\frac{b_2}{a_2b_3}q'\bigg(\frac{1}{a_2}(x_0-a_1p(z_0))\bigg)\bigg||v_{01}|-\frac{|v_{02}|}{|b_3|}-\frac{|b_1|}{|b_3|}|u_0|\\
\geq& \bigg|\frac{a_1b_2}{a_2b_3}q'\bigg(\frac{1}{a_2}(x_0-a_1p(z_1))\bigg)p'(z_0)\bigg| |u_0|-\bigg|\frac{b_2}{a_2b_3}q'\bigg(\frac{1}{a_2}(x_0-a_1p(z_0))\bigg)\bigg||v_{0}|-\frac{|v_{02}|}{|b_3|}-\frac{|b_1|}{|b_3|}|u_0|\\
\geq&\bigg|\frac{a_1b_2}{a_2b_3}q'\bigg(\frac{1}{a_2}(x_0-a_1p(z_0))\bigg)p'(z_0)\bigg| |u_0|-\bigg|\frac{b_2}{a_2b_3}q'\bigg(\frac{1}{a_2}(x_0-a_1p(z_0))\bigg)\bigg||u_{0}|-\frac{|u_0|}{|b_3|}-\frac{|b_1|}{|b_3|}|u_0|\\
\geq&\lambda|u_0|,
\end{split}
\eeqq
which yields that $m'_{\mathcal{C},w_0}\geq\lambda$.

It follows from Lemma \ref{unifhyp} that the invariant set $\Lambda$ is uniformly hyperbolic. This completes the proof.
\end{proof}

Therefore, it follows from Lemmas \ref{chaoticshift}, \ref{connectedfour-zero}, and \ref{hyptwoexp-zero} that Theorem \ref{twopositive-zero} holds.

\begin{theorem}\label{twopositive}
Suppose that there are $1\leq i_0<r_1$ and $1\leq j_0<r_2$ such that $0<\al_{i_0}<\al_{i_0+1}$ and $0<\xi_{j_0}<\xi_{j_0+1}$, that is, there are two distinct positive zeros of $p(x)$ and $q(y)$. Set
\beqq
M_1:=\max\{a_2\xi_{j_0},a_2\xi_{j_0+1}\},\ M_2:=\min\{a_2\xi_{j_0},a_2\xi_{j_0+1}\},
\eeqq
\beqq
N_1:=\max\{b_1\al_{i_0}+b_3\al_{i_0},
b_1\al_{i_0}+b_3\al_{i_0+1},
b_1\al_{i_0+1}+b_3\al_{i_0},
b_1\al_{i_0+1}+b_3\al_{i_0+1}\},
\eeqq
\beqq
N_2:=\min\{b_1\al_{i_0}+b_3\al_{i_0},
b_1\al_{i_0}+b_3\al_{i_0+1},
b_1\al_{i_0+1}+b_3\al_{i_0},
b_1\al_{i_0+1}+b_3\al_{i_0+1}\}.
\eeqq
Consider the following different situations:
\begin{itemize}
\item [(i).] $M_1<\al_{i_0}$ and $a_1p^{(m_{i_0})}(\al_{i_0})>0$, and $N_1<\xi_{j_0}$ and $b_2q^{(n_{j_0})}(\xi_{j_0})>0$;
\item [(ii).]  $M_1<\al_{i_0}$ and $a_1p^{(m_{i_0})}(\al_{i_0})>0$, and $N_2>\xi_{j_0+1}$ and $b_2q^{(n_{j_0})}(\xi_{j_0})<0$;
\item [(iii).]  $M_2>\al_{i_0+1}$ and $a_1p^{(m_{i_0})}(\al_{i_0})<0$, and $N_1<\xi_{j_0}$ and $b_2q^{(n_{j_0})}(\xi_{j_0})>0$;
\item [(iv).]  $M_2>\al_{i_0+1}$ and $a_1p^{(m_{i_0})}(\al_{i_0})<0$, and $N_2>\xi_{j_0+1}$ and $b_2q^{(n_{j_0})}(\xi_{j_0})<0$.
\end{itemize}
Then, for fixed $a_2$, $b_1$, and $b_3$, and sufficiently large $|a_1|$ and $|b_2|$, there exists a Smale horseshoe for the map \eqref{sys6}, especially, there is a uniformly hyperbolic invariant set $\Lambda$ on which $F$ is topologically conjugate to two-sided fullshift on four  symbols. Therefore, they are Li-Yorke chaotic as well as Devaney chaotic.
\end{theorem}

\begin{remark}
The four different situations in the assumptions of Theorem \ref{twopositive} is the generalization of the assumptions \eqref{assumption-1-positive} in Theorem \ref{onepositive}.
\end{remark}

Consider the compact set $$U:=[\al_{i_0},\al_{i_0+1}]\times[\xi_{j_0},\xi_{j_0+1}]\times[\al_{i_0},\al_{i_0+1}].$$

Fix a constant $\lambda>1$. The invariant set is denoted by
$$\Lambda:=\bigcap^{+\infty}_{-\infty}F^i(U).$$

\begin{lemma} \label{connectedfour}
In Cases (i)-(iv) of Theorem \ref{twopositive}, for fixed $a_2$, $b_1$, and $b_3$, and sufficiently large
$|a_1|$ and $|b_2|$, one has that
$F(U)\cap U$ and $F^{-1}(U)\cap U$ have four non-empty connected components, respectively.
\end{lemma}

\begin{proof}
It is to show that $F(U)\cap U$ has four connected components under the assumptions of Theorem \ref{twopositive}. We will give the geometric description of $F(U)$ in the case: $a_2>0$, $b_3>0$, $b_1=0$, $a_2\xi_{j_0+1}<\al_{i_0}$, $b_3\al_{i_0+1}<\xi_{j_0}$, $a_1p^{(m_{i_0})}(\al_{i_0})>0$, and $b_2q^{(n_{j_0})}(\xi_{j_0})>0$. The graph for other parameters could be obtained similarly.

First, Since $U$ is a cubic, the first step is to find the position of the vertices, that is, the image of $A,B,C,D,A',B',C',D'$. We should know the relative position of $F(A)$, $F(B)$, $F(C)$, $F(D)$, $F(A')$, $F(B')$, $F(C')$, $F(D')$. The vertices of the cubic $U$ are listed as follows:
\beqq
A=(\al_{i_0+1},\xi_{j_0},\al_{i_0}),
B=(\al_{i_0+1},\xi_{j_0+1},\al_{i_0}),
C=(\al_{i_0},\xi_{j_0+1},\al_{i_0}),
D=(\al_{i_0},\xi_{j_0},\al_{i_0}),
\eeqq
\beqq
A'=(\al_{i_0+1},\xi_{j_0},\al_{i_0+1}),
B'=(\al_{i_0+1},\xi_{j_0+1},\al_{i_0+1}),
C'=(\al_{i_0},\xi_{j_0+1},\al_{i_0+1}),
D'=(\al_{i_0},\xi_{j_0},\al_{i_0+1});
\eeqq
the image of these points under $F$ are
\beqq
F(A)=(a_2\xi_{j_0},b_1\al_{i_0+1}+b_3\al_{i_0},\al_{i_0+1}),\ F(B)=(a_2\xi_{j_0+1},b_1\al_{i_0+1}+b_3\al_{i_0},\al_{i_0+1}),
\eeqq
\beqq
F(C)=(a_2\xi_{j_0+1},b_1\al_{i_0}+b_3\al_{i_0},\al_{i_0}),\ F(D)=(a_2\xi_{j_0},b_1\al_{i_0}+b_3\al_{i_0},\al_{i_0}),
\eeqq
\beqq
F(A')=(a_2\xi_{j_0},b_1\al_{i_0+1}+b_3\al_{i_0+1},\al_{i_0+1}),\ F(B')=(a_2\xi_{j_0+1},b_1\al_{i_0+1}+b_3\al_{i_0+1},\al_{i_0+1}),
\eeqq
\beqq
F(C')=(a_2\xi_{j_0+1},b_1\al_{i_0}+b_3\al_{i_0+1},\al_{i_0}),\
F(D')=(a_2\xi_{j_0},b_1\al_{i_0}+b_3\al_{i_0+1},\al_{i_0}).
\eeqq
It is evident that the image of the vertices are not contained in $U$.

Set
$$AD:=\{(x,y,z):\ x\in[\al_{i_0},\al_{i_0+1}],y=\xi_{j_0},z=\al_{i_0}\},$$
$$AB:=\{(x,y,z):\ x=\al_{i_0+1},y\in[\xi_{j_0},\xi_{j_0+1}],z=\al_{i_0}\},$$
$$AA':=\{(x,y,z):\ x=\al_{i_0+1},y=\xi_{j_0},z\in[\al_{i_0},\al_{i_0+1}]\},$$
$$ABCD:=\{(x,y,z):\ (x,y)\in[\al_{i_0},\al_{i_0+1}]\times[\xi_{j_0},\xi_{j_0+1}],z=\al_{i_0}\}.$$

Second, it is to determine the image of the plane $ABCD$ under the map $F$. By \eqref{sys6}, $F(AD)$ is a parabola along the $x$-axis, and $F(AB)$ is a parabola along $y$-axis. The image $F(ABCD)$ can be regarded as a movement of the parabola $F(AD)$ along another parabola $F(AB)$, but the direction of $F(AD)$ should not vary too much.

Finally, it is to investigate the graph of $F(U)$. It follows from \eqref{sys6} that $F(AA')$ is a line segment. The graph of $F(U)$ can be obtained by moving the surface $F(ABCD)$ obtained in the previous step along $F(AA')$.

Hence, there are four non-empty connected components of $F(U)\cap U$ for sufficiently large $|a_1|$ and $|b_2|$.

Now, it is to study $F^{-1}(U)\cap U$. The image of the vertices of $U$ under $F^{-1}=H$ are
\beqq
H(A)=\bigg(\al_{i_0},\frac{\al_{i_0+1}}{a_2},\frac{(\xi_{j_0}-b_1\al_{i_0}-b_2q(\al_{i_0+1}/a_2))}{b_3}\bigg),
\eeqq
\beqq
H(B)=\bigg(\al_{i_0},\frac{\al_{i_0+1}}{a_2},\frac{(\xi_{j_0+1}-b_1\al_{i_0}-b_2q(\al_{i_0+1}/a_2))}{b_3}\bigg),
\eeqq
\beqq
H(C)=\bigg(\al_{i_0},\frac{\al_{i_0}}{a_2},\frac{(\xi_{j_0+1}-b_1\al_{i_0}-b_2q(\al_{i_0}/a_2))}{b_3}\bigg),
\eeqq
\beqq
H(D)=\bigg(\al_{i_0},\frac{\al_{i_0}}{a_2},\frac{(\xi_{j_0}-b_1\al_{i_0}-b_2q(\al_{i_0}/a_2))}{b_3}\bigg),
\eeqq
\beqq
H(A')=\bigg(\al_{i_0+1},\frac{\al_{i_0+1}}{a_2},\frac{(\xi_{j_0}-b_1\al_{i_0+1}-b_2q(\al_{i_0+1}/a_2))}{b_3}\bigg),
\eeqq
\beqq
H(B')=\bigg(\al_{i_0+1},\frac{\al_{i_0+1}}{a_2},\frac{(\xi_{j_0+1}-b_1\al_{i_0+1}-b_2q(\al_{i_0+1}/a_2))}{b_3}\bigg),
\eeqq
\beqq
H(C')=\bigg(\al_{i_0+1},\frac{\al_{i_0}}{a_2},\frac{(\xi_{j_0+1}-b_1\al_{i_0+1}-b_2q(\al_{i_0}/a_2))}{b_3}\bigg),
\eeqq
\beqq
H(D')=\bigg(\al_{i_0+1},\frac{\al_{i_0}}{a_2},\frac{(\xi_{j_0}-b_1\al_{i_0+1}-b_2q(\al_{i_0}/a_2))}{b_3}\bigg).
\eeqq

By \eqref{inversesys6}, given any $(x_0,y_0,z_0)\in U$,
\beqq
\left\{
  \begin{array}{ll}
    x_0=z \\
    y_0=\frac{1}{a_2}(x-a_1p(z)) \\
    z_0=\frac{1}{b_3}(y-b_1z-b_2q(\frac{1}{a_2}(x-a_1p(z))))
  \end{array}
\right.,
\eeqq
there should exist four different solutions by Lemma \ref{monotone} if the solution set is non-empty. Further, since $F$ is a diffeomorphism, $F^{-1}(U)\cap U=F^{-1}(U\cap F(U))$. This, together with the above geometric description of $F(U)$, yields that $F^{-1}(U)\cap U$ has four non-empty connected components.

This completes the proof.
\end{proof}

\begin{remark}
In the last section, we provide an example with the help of Mathematica software to draw the graph of $F(U)\cap U$ and $F^{-1}(U)\cap U$ to illustrate the results of Theorem \ref{twopositive}.
\end{remark}

\begin{lemma} \label{hyptwoexp}
In Cases (i)-(iv) of Theorem \ref{twopositive}, for fixed $a_2$, $b_1$, and $b_3$, and sufficiently large
$|a_1|$ and $|b_2|$, one has that the invariant set is uniformly hyperbolic.
\end{lemma}

\begin{proof}

Without loss of generality, assume that $a_2$, $b_1$, and $b_3$ are fixed, $|a_1|$ and $|b_2|$ are large enough. By (i) and (ii) of Lemma \ref{monotone}, we could assume that $|a_1|$ and $|b_2|$ are sufficiently large such that for $w=(x,y,z)\in U\cap F^{-1}(U)$, one has that $x\in(\al_{i_0},\delta_1]\cup[\delta'_1,\al_{i_0+1})$
and $y\in(\xi_{j_0},\delta_2]\cup[\delta'_2,\xi_{j_0+1})$.
For any point $w_0=(x_0,y_0,z_0)\in U$, denote $w_1=F(w_0)=(x_1,y_1,z_1)$ and $w_{-1}=F^{-1}(w_0)=(x_{-1},y_{-1},z_{-1})$.

From \eqref{polyineq1}, \eqref{polyineq2}, and Lemma \ref{monotone}, it follows that for any $w=(x,y,z)\in U\cap F^{-1}(U)$,
\beq\label{mono-ineq-2-i}
\mbox{Case (i)}:  \frac{|a_1p'(x)|}{|a_1p(x)|}\geq\frac{\max\{\lambda+|a_2|,1+|a_2b_3|\lambda+|a_2|(1+|b_1|)\}}{\al_{i_0}-M_1},
\frac{|b_2q'(y)|}{|b_2q(y)|}\geq\frac{\lambda+|b_1|+|b_3|}{\xi_{j_0}-N_1};
\eeq
\beq\label{mono-ineq-2-ii}
\mbox{Case (ii)}:  \frac{|a_1p'(x)|}{|a_1p(x)|}\geq\frac{\max\{\lambda+|a_2|,1+|a_2b_3|\lambda+|a_2|(1+|b_1|)\}}{\al_{i_0}-M_1},
\frac{|b_2q'(y)|}{|b_2q(y)|}\geq\frac{\lambda+|b_1|+|b_3|}{N_2-\xi_{j_0+1}};
\eeq
\beq\label{mono-ineq-2-iii}
\mbox{Case (iii)}: \frac{|a_1p'(x)|}{|a_1p(x)|}\geq\frac{\max\{\lambda+|a_2|,1+|a_2b_3|\lambda+|a_2|(1+|b_1|)\}}{M_2-\al_{i_0+1}},
\frac{|b_2q'(y)|}{|b_2q(y)|}\geq\frac{\lambda+|b_1|+|b_3|}{\xi_{j_0}-N_1};
\eeq
\beq\label{mono-ineq-2-iv}
\mbox{Case (iv)}: \frac{|a_1p'(x)|}{|a_1p(x)|}\geq\frac{\max\{\lambda+|a_2|,1+|a_2b_3|\lambda+|a_2|(1+|b_1|)\}}{M_2-\al_{i_0+1}},
\frac{|b_2q'(y)|}{|b_2q(y)|}\geq\frac{\lambda+|b_1|+|b_3|}{N_2-\xi_{j_0+1}}.
\eeq

Consider the point $w_0\in U\cap F^{-1}(U)$, by \eqref{sys6},
\beqq
\left\{
  \begin{array}{ll}
   x_1=a_1p(x_0)+a_2y_0 \\
    y_1=b_1x_0+b_2q(y_0)+b_3z_0 \\
    z_1=x_0
  \end{array}
\right..
\eeqq
So, $a_1p(x_0)=x_1-a_2y_0$ and $b_2q(y_0)=y_1-b_1x_0-b_3z_0$.
Hence, one has
\beq\label{bound-2-i}
\mbox{Case (i)}:\ \ a_1p(x_0)\geq\al_{i_0}-M_1,\ b_2q(y_0)\geq\xi_{j_0}-N_1;
\eeq
\beq\label{bound-2-ii}
\mbox{Case (ii)}:\ \ a_1p(x_0)\geq \al_{i_0}-M_1,\ -b_2q(y_0)\geq N_2-\xi_{j_0+1};
\eeq
\beq\label{bound-2-iii}
\mbox{Case (iii)}:\ \ -a_1p(x_0)\geq M_2-\al_{i_0+1},\ b_2q(y_0)\geq \xi_{j_0}-N_1;
\eeq
\beq\label{bound-2-iv}
\mbox{Case (iv)}:\ \ -a_1p(x_0)\geq M_2-\al_{i_0+1},\ -b_2q(y_0)\geq N_2-\xi_{j_0+1}.
\eeq

For any point $w_0\in U\cap F(U)$, it follows from \eqref{inversesys6} that
\beqq
\left\{
  \begin{array}{ll}
    x_{-1}=z_0 \\
    y_{-1}=\frac{1}{a_2}(x_0-a_1p(z_0)) \\
    z_{-1}=\frac{1}{b_3}(y_0-b_1z_0-b_2q(\frac{1}{a_2}(x_0-a_1p(z_0))))
  \end{array}
\right..
\eeqq
So,
$a_1p(x_{-1})=a_1p(z_0)=x_0-a_2y_{-1}$ and
$b_2q(y_{-1})=b_2q(\frac{1}{a_2}(x_0-a_1p(z_0)))=y_0-b_1z_0-b_3z_{-1}$.
Thus, one has
\beq\label{inv-bound-2-i}
\mbox{Case (i)}:\ \ a_1p(x_{-1})\geq\al_{i_0}-M_1,\ b_2q(y_{-1})\geq\xi_{j_0}-N_1;
\eeq
\beq\label{inv-bound-2-ii}
\mbox{Case (ii)}:\ \ a_1p(x_{-1})\geq \al_{i_0}-M_1,\ -b_2q(y_{-1})\geq N_2-\xi_{j_0+1};
\eeq
\beq\label{inv-bound-2-iii}
\mbox{Case (iii)}:\ \ -a_1p(x_{-1})\geq M_2-\al_{i_0+1},\ b_2q(y_{-1})\geq \xi_{j_0}-N_1;
\eeq
\beq\label{inv-bound-2-iv}
\mbox{Case (iv)}:\ \ -a_1p(x_{-1})\geq M_2-\al_{i_0+1},\ -b_2q(y_{-1})\geq N_2-\xi_{j_0+1}.
\eeq

By \eqref{mono-ineq-2-i}--\eqref{mono-ineq-2-iv} and \eqref{bound-2-i}--\eqref{bound-2-iv}, one has that for any point $w_0\in U\cap F^{-1}(U)$,
\beq\label{bound-deriv-2-i-iv}
\mbox{Cases (i)--(iv)}:\ |a_1p'(x_0)|\geq\lambda+|a_2|,\ |b_2q'(y_0)|\geq\lambda+|b_1|+|b_3|.
\eeq

Now, it is to show that for any point $w_0\in U\cap F^{-1}(U)$,
 $m_{\mathcal{C},w_0}\geq\lambda>1$.

Consider the unit cone
\beqq
K_1(\mathbb{R}^2, \mathbb{R})=\bigg\{\mathbf{v}=
\left(
\begin{array}{c}
v_0\\
u_0
\end{array}
\right)
:\ v_0=(v_{01},v_{02})^T\in \mathbb{R}^2,\ u_0\in\mathbb{R},\ \mbox{and}\ |u_0|\leq|v_0|\bigg\},
\eeqq
where $|v_0|=\max\{|v_{01}|,|v_{02}|\}$ and $|\mathbf{v}|=\max\{|v_0|,|u_0|\}$.

Suppose that $
\left(
\begin{array}{c}
v_1\\
u_1
\end{array}
\right)
=DF_{w_0}
\left(
\begin{array}{c}
v_0\\
u_0
\end{array}
\right)$, where $
\left(
\begin{array}{c}
v_0\\
u_0
\end{array}
\right)\in K_1(\mathbb{R}^2,\mathbb{R})$.
By \eqref{positivederivative},
\beqq
\left\{
  \begin{array}{ll}
v_{11}=a_1p'(x_0)v_{01}+a_2v_{02}\\
v_{12}=b_1v_{01}+b_2q'(y_0)v_{02}+b_3u_0\\
u_1=v_{01}
  \end{array}.
\right.
\eeqq
By \eqref{bound-deriv-2-i-iv},
if $|v_0|=|v_{01}|$, then $$|v_{11}|\geq |v_{01}||a_1p'(x_0)|-|a_2||v_{01}|\geq\lambda|v_{01}|;$$ if $|v_0|=|v_{02}|$, then $$|v_{12}|\geq|v_{02}||b_2q'(y_0)|-|b_1||v_{02}|-|b_3||u_0|\geq\lambda|v_{02}|,$$ which implies that $|(v_1,u_1)|\geq\lambda|(v_0,u_0)|$.  Hence, $m_{\mathcal{C},w_0}\geq\lambda>1$.

By \eqref{mono-ineq-2-i}--\eqref{mono-ineq-2-iv} and \eqref{inv-bound-2-i}--\eqref{inv-bound-2-iv}, one has that for any point $w_0\in U\cap F(U)$,
\beq\label{bound-deriv-inverse-2}
|a_1p'(x_{-1})|=|a_1p'(z_0)|\geq1+|a_2b_3|\lambda+|a_2|(1+|b_1|),\ |b_2q'(y_{-1})|=\bigg|b_2q'\bigg(\frac{1}{a_2}(x_0-a_1p(z_0))\bigg)\bigg|\geq1.
\eeq

Next, it is to show that
then $m'_{\mathcal{C},w_0}\geq\lambda>1$. For any
$\left(
\begin{array}{c}
v_0\\
u_0
\end{array}
\right)
\not\in K_1(\mathbb{R}^2,\mathbb{R})$, that is ,  $|u_0|>|v_0|$, suppose $
\left(
\begin{array}{c}
v_{-1}\\
u_{-1}
\end{array}
\right)=DF^{-1}_{w_0}
\left(
\begin{array}{c}
v_0\\
u_0
\end{array}
\right)$, by \eqref{inversederivative}, one has
\beqq
\left\{
  \begin{array}{ll}
 v_{-11}=u_0 \\
  v_{-12} =\frac{1}{a_2} v_{01}-\frac{a_1}{a_2}p'(z_0) u_0\\
  u_{-1}=  -\frac{b_2}{a_2b_3}q'(\frac{1}{a_2}(x_0-a_1p(z_0)))v_{01}+\frac{1}{b_3}v_{02} -\frac{b_1}{b_3}u_0+\frac{a_1b_2}{a_2b_3}q'(\frac{1}{a_2}(x_0-a_1p(z_0)))p'(z_0) u_0\\
\end{array}.
\right.
\eeqq
So, by \eqref{bound-deriv-inverse-2},
\beqq
\begin{split}
|u_{-1}|&\geq \bigg|\frac{a_1b_2}{a_2b_3}q'\bigg(\frac{1}{a_2}(x_0-a_1p(z_0))\bigg)p'(z_0)\bigg| |u_0|-\bigg|\frac{b_2}{a_2b_3}q'\bigg(\frac{1}{a_2}(x_0-a_1p(z_0))\bigg)\bigg||v_{01}|-\frac{|v_{02}|}{|b_3|}-\frac{|b_1|}{|b_3|}|u_0|\\
\geq& \bigg|\frac{a_1b_2}{a_2b_3}q'\bigg(\frac{1}{a_2}(x_0-a_1p(z_1))\bigg)p'(z_0)\bigg| |u_0|-\bigg|\frac{b_2}{a_2b_3}q'\bigg(\frac{1}{a_2}(x_0-a_1p(z_0))\bigg)\bigg||v_{0}|-\frac{|v_{02}|}{|b_3|}-\frac{|b_1|}{|b_3|}|u_0|\\
\geq&\bigg|\frac{a_1b_2}{a_2b_3}q'\bigg(\frac{1}{a_2}(x_0-a_1p(z_0))\bigg)p'(z_0)\bigg| |u_0|-\bigg|\frac{b_2}{a_2b_3}q'\bigg(\frac{1}{a_2}(x_0-a_1p(z_0))\bigg)\bigg||u_{0}|-\frac{|u_0|}{|b_3|}-\frac{|b_1|}{|b_3|}|u_0|\\
\geq&\lambda|u_0|,
\end{split}
\eeqq
which yields that $m'_{\mathcal{C},w_0}\geq\lambda>1$.

By Lemma \ref{unifhyp}, one has that the invariant set $\Lambda$ is uniformly hyperbolic. The proof is completed.
\end{proof}

Therefore, it follows from Lemmas \ref{chaoticshift}, \ref{connectedfour}, and \ref{hyptwoexp} that Theorem \ref{twopositive} holds.
\medskip

Based on the proof of Theorems \ref{twopositive-zero} and \ref{twopositive}, one has the following result:
\begin{theorem}\label{twopositive-1zero-1positive}
Suppose that there are $1\leq i_0<r_1$ and $1\leq j_0<r_2$ such that $0=\al_{i_0}<\al_{i_0+1}$ and $0<\xi_{j_0}<\xi_{j_0+1}$, and $m_{i_0}=m_{i_0+1}=1$.
Suppose that
\begin{itemize}
\item [(1).] $a_2<0$ and $a_1p'(0)>0$;
\item [(2).] $\max\{0,
b_3\al_{i_0+1},
b_1\al_{i_0+1},
b_1\al_{i_0+1}+b_3\al_{i_0+1}\}<\xi_{j_0}$ and $b_2q^{(n_{j_0})}(\xi_{j_0})>0$.
\end{itemize}
Then, for fixed $a_2$, $b_1$, and $b_3$, and sufficiently large $|a_1|$ and $|b_2|$, there exists a Smale horseshoe for the map \eqref{sys6}, especially, there is a uniformly hyperbolic invariant set $\Lambda$ on which $F$ is topologically conjugate to the two-sided fullshift on four  symbols. Therefore, they are Li-Yorke chaotic as well as Devaney chaotic.
\end{theorem}

\begin{theorem} \label{twopositive-nzero}
For system \eqref{sys6}, suppose that $p(x)=\prod^m_{i=1}(x-\al_i)$ and
$q(y)=\prod^n_{j=1}(y-\xi_j)$,
where $\al_1<\al_{2}<\cdots<\al_m$ and $\xi_1<\xi_2<\cdots<\xi_n$ are real numbers, $m\geq2$, and $n\geq2$.
Then, for fixed $a_2$, $b_1$, and $b_3$, and sufficiently large $|a_1|$ and $|b_2|$, there exist a Smale horseshoe for the map \eqref{sys6} and a uniformly hyperbolic invariant set $\Lambda$ on which $F$ is topologically conjugate to the two-sided fullshift on $mn$  symbols. Therefore, they are chaotic in the sense of both Li-Yorke and Devaney.
\end{theorem}

\begin{proof}
Fix a constant $\ld>1$. It follows from the properties of $p(x)$ and $q(y)$ that there exist positive constants $\eta_i$, $1\leq i\leq m$, and $\tau_j$, $1\leq j\leq n$, such that, the sets $V_i=[\al_i-\eta_i, \al_i+\eta_i]$, $1\leq i\leq m$, are pairwise disjoint, and $|p'(x)|>0$ for all $x\in V_i$, $1\leq i\leq m$; the sets $W_j=[\xi_j-\tau_j,\xi_j+\tau_j]$, $1\leq j\leq n$, are pairwise disjoint, and $|q'(y)|>0$ for all $y\in W_j$, $1\leq j\leq n$.

Denote
$$U:=[\al_1-\eta_1,\al_m+\eta_m]\times[\xi_1-\tau_1,\xi_n+\tau_n]\times[\al_1-\eta_1,\al_m+\eta_m],$$ and
$$\Ld:=\bigcap^{+\infty}_{-\infty}F^{j}(U).$$
Set
\beqq
M_0:=\inf_{x\in\cup^m_{i=1}V_i}|p'(x)|,\ M'_0:=\inf_{y\in\cup^n_{j=1}W_j}|q'(y)|;
\eeqq
\beqq
M_1:=\min_{1\leq i\leq m}|p(\al_i\pm\eta_i)|,\ M'_1:=\min_{1\leq j\leq n}|q(\xi_j\pm\tau_j)|;
\eeqq
\beqq
N_0:=\frac{\max\{\lambda+|a_2|,1+|a_2b_3|\lambda+|a_2|(1+|b_1|)\}}{M_0},\
N'_0:=\frac{\lambda+|b_1|+|b_3|}{M'_0};
\eeqq
\beqq
N_1:=\frac{|a_2|\max\{|\xi_1-\tau_1|,|\xi_n-\tau_n|\}+\max\{|\al_1-\eta_1|,|\al_m-\eta_m|\}}{M_1},
\eeqq
\beqq
N'_1:=\frac{\max\{|\xi_1-\tau_1|,|\xi_n-\tau_n|\}+(|b_1|+|b_3|)\max\{|\al_1-\eta_1|,|\al_m-\eta_m|\}}{M'_1}.
\eeqq
If $|a_1|>N_0$ and $|b_2|>N'_0$, by applying similar approaches in the proof of Lemma \ref{hyptwoexp-zero}, one has that
the invariant set $\Ld$ is uniformly hyperbolic, where the constants $N_0$ and $N'_0$ have similar effects with the constants specified in \eqref{bound-deriv-2-single-zero}.

If $|a_1|>N_1$ and $|b_2|>N'_1$, it follows from simple computation and  the fact that $p(x)$ and $q(y)$ have simple real roots, one has that $F(U)\cap U$ has $mn$ connected components.

Hence, if $|a_1|>\max\{N_0,N_1\}$ and $|b_2|>\max\{N'_0,N'_1\}$, then there exist a Smale horseshoe and a hyperbolic invariant set on which $F$ is topologically conjugate to the two-sided fullshift on $mn$ symbols. This, together with Lemma \ref{chaoticshift}, yields that $F$ is Li-Yorke chaotic as well as Devaney chaotic. This completes the proof.

\end{proof}

\begin{remark}
Similar results could be obtained for different parameters. For example, we could assume that $b_3=c_1=c_3=0$ and $a_3b_1\neq0$.
\end{remark}
\medskip

\subsection{The expanding maps}

In this subsection, we consider the following type of maps: $F:(x,y,z)\in\mathbb{R}^3\to(f_1,f_2,f_3)\in\mathbb{R}^3:$
\begin{equation}\label{sysexp1}
\left\{
  \begin{array}{ll}
    f_1(x,y,z)=a_1p(x)+a_2y+a_3z \\
    f_2(x,y,z)=b_1x+b_2q(y)+b_3z \\
    f_3(x,y,z)=c_1x+c_2y+c_3r(z)
  \end{array}
\right.,
\end{equation}
where $a_1b_2c_3\neq0$.

The derivative of \eqref{sysexp1} is
\begin{equation}\label{derivativeexp1}
DF=\left(
  \begin{array}{ccc}
    a_1p'(x) & a_2 & a_3 \\
    b_1 & b_2q'(y) & b_3 \\
    c_1 & c_2 & c_3r'(z)\\
  \end{array}
\right).
\end{equation}
Generally, the map \eqref{sysexp1} is not invertible.

An example is given as follows:
\begin{equation*}
\left\{
  \begin{array}{ll}
    f_1(x,y,z)=a_1x(1-x)+a_2y+a_3z \\
    f_2(x,y,z)=b_1x+b_2y(1-y)+b_3z \\
    f_3(x,y,z)=c_1x+c_2y+c_3z(1-z)
  \end{array}
\right..
\end{equation*}
This example can also be thought of as the generalization of the H\'{e}non map in three-dimensional spaces.

\begin{theorem} \label{threepositive-zero}
Suppose that there are $1\leq i_0<r_1$, $1\leq j_0<r_2$, and
$1\leq k_0<r_3$ such that $0=\al_{i_0}<\al_{i_0+1}$, $0=\xi_{j_0}<\xi_{j_0+1}$,
and $0=\kappa_{k_0}<\kappa_{k_0+1}$, and $m_{i_0}=m_{i_0+1}=n_{j_0}=n_{j_0+1}=l_{k_0}=l_{k_0+1}=1$, that is, there are two distinct non-negative zeros of $p(x)$, $q(y)$, and $r(z)$. Suppose that
\begin{itemize}
\item[(1).]
$a_2\leq0,\ a_3\leq0,\ \mbox{and}\ a_1p'(0)>0;$
\item[(2).] $b_1\leq0,\ b_3\leq0,\ \mbox{and}\ b_2q'(0)>0;$
\item[(3).] $c_1\leq0,\ c_2\leq0,\ \mbox{and}\ c_3r'(0)>0.$
\end{itemize}
For fixed $a_2$, $a_3$, $b_1$, $b_3$, $c_1$, and $c_2$, if $|a_1|$, $|b_2|$, and $|c_3|$ are sufficiently large, then there exists a forward invariant set of the map \eqref{sysexp1} on which the map is topologically semi-conjugate to the one-sided fullshift on eight symbols and it is chaotic in the sense of Li-Yorke.
\end{theorem}

\begin{proof}
Denote $$U:=[0,\al_{i_0+1}]\times[0,\xi_{j_0+1}]\times[0,\kappa_{k_0+1}].$$

Fix $\lambda>1$, $a_2$, $a_3$, $b_1$, $b_3$, $c_1$, and $c_2$. It follows from Lemma \ref{monotone} that one could assume that $|a_1|$, $|b_2|$, and $|c_3|$ are sufficiently large such that for $w=(x,y,z)\in U\cap F^{-1}(U)$, one has that $x\in[0,\delta_1]\cup[\delta'_1,\al_{i_0+1}]$, $y\in[0,\delta_2]\cup[\delta'_2,\xi_{j_0+1}]$,
and $z\in[0,\delta_3]\cup[\delta'_3,\kappa_{k_0+1}]$.

Set $$M_0:=\min\bigg\{\min_{x\in[0,\delta_1]\cup[\delta'_1,\al_{i_0+1}]}|p'(x)|,\min_{y\in[0,\delta_2]\cup[\delta'_2,\xi_{j_0+1}]}|q'(y)|,
\min_{z\in[0,\delta_3]\cup[\delta'_3,\kappa_{k_0+1}]}|r'(z)|\bigg\}.$$
Suppose that
\beq\label{bound-deriv-3-single-zero}
|a_1|\geq\frac{\lambda+|a_2|+|a_3|}{M_0},\ |b_2|\geq\frac{\lambda+|b_1|+|b_3|}{M_0},\ |c_3|\geq\frac{\lambda+|c_1|+|c_2|}{M_0}.
\eeq

For $v=(v_{01},v_{02},v_{03})^T\in\mathbb{R}^3$, set $|v|=\max_{1\leq i\leq3}|v_{0i}|$.
By \eqref{derivativeexp1},
\beqq
\left\{
  \begin{array}{ll}
v_{11}=a_1p'(x_0)v_{01}+a_2v_{02}+a_3v_{03}\\
v_{12}=b_1v_{01}+b_2q'(y_0)v_{02}+b_3v_{03}\\
v_{13}=c_1v_{01}+c_2v_{02}+c_3r'(z_0)v_{03}
\end{array}
\right..
\eeqq
Without loss of generality, suppose that $|v|=|v_{01}|$. Hence, by \eqref{bound-deriv-3-single-zero}, one has
\beqq
|v_1|=\max_{1\leq i\leq 3}|v_{1i}|\geq|v_{11}|\geq|a_1p'(x_0)||v|-|a_2||v|-|a_3||v|\geq\lambda|v|,
\eeqq
which implies that $\|DF\|\geq\lambda$. So, one has that
the map is expansion on $U\cap F^{-1}(U)$.

For fixed $a_2,a_3,b_1,b_3,c_1,c_2$, if $|a_1|$, $|b_2|$, and $|c_3|$ are sufficiently large, then there exist eight different closed subsets $V_1,...,V_8$ of $U$ such that $V_i\cap V_j=\emptyset$, $1\leq i\neq j\leq 8$, and $F(V_i)=U$.  This is a coupled-expanding map \cite{ZhangShi,ZhangShiChen2013}. By Theorem 3.1 in \cite{ZhangShiChen2013}, there exists a forward invariant set on which the map is topologically semi-conjugate to the one-sided fullshift on eight symbols, and it is chaotic in the sense of Li-Yorke. This completes the proof.
\end{proof}

\begin{theorem} \label{threepositive}
Suppose that there are $1\leq i_0<r_1$, $1\leq j_0<r_2$, and
$1\leq k_0<r_3$ such that $0<\al_{i_0}<\al_{i_0+1}$, $0<\xi_{j_0}<\xi_{j_0+1}$,
and $0<\kappa_{k_0}<\kappa_{k_0+1}$, that is, there are two distinct positive zeros of $p(x)$, $q(y)$, and $r(z)$. Set
\beqq
M_1:=\max\{a_2\xi_{j_0}+a_3\kappa_{k_0},a_2\xi_{j_0+1}+a_3\kappa_{k_0},a_2\xi_{j_0}+a_3\kappa_{k_0+1},a_2\xi_{j_0+1}+a_3\kappa_{k_0+1}\},\eeqq
\beqq
M_2:=\min\{a_2\xi_{j_0}+a_3\kappa_{k_0},a_2\xi_{j_0+1}+a_3\kappa_{k_0},a_2\xi_{j_0}+a_3\kappa_{k_0+1},a_2\xi_{j_0+1}+a_3\kappa_{k_0+1}\},\eeqq
\beqq
N_1:=\max\{b_1\al_{i_0}+b_3\kappa_{k_0},b_1\al_{i_0+1}+b_3\kappa_{k_0},b_1\al_{i_0}+b_3\kappa_{k_0+1},b_1\al_{i_0+1}+b_3\kappa_{k_0+1}\},\eeqq
\beqq N_2:=\min\{b_1\al_{i_0}+b_3\kappa_{k_0},b_1\al_{i_0+1}+b_3\kappa_{k_0},b_1\al_{i_0}+b_3\kappa_{k_0+1},b_1\al_{i_0+1}+b_3\kappa_{k_0+1}\},\eeqq
\beqq
Q_1:=\max\{c_1\al_{i_0}+c_2\xi_{j_0},c_1\al_{i_0+1}+c_2\xi_{j_0},c_1\al_{i_0}+c_2\xi_{j_0+1},c_1\al_{i_0+1}+c_2\xi_{j_0+1}\},\eeqq
\beqq Q_2:=\min\{c_1\al_{i_0}+c_2\xi_{j_0},c_1\al_{i_0+1}+c_2\xi_{j_0},c_1\al_{i_0}+c_2\xi_{j_0+1},c_1\al_{i_0+1}+c_2\xi_{j_0+1}\}.\eeqq
Consider the following different situations:
\begin{itemize}
\item [(i).] $M_1<\al_{i_0}$ and $a_1p^{(m_{i_0})}(\al_{i_0})>0$, $N_1<\xi_{j_0}$ and $b_2q^{(n_{j_0})}(\xi_{j_0})>0$, and $Q_1<\kappa_{k_0}$ and $c_3r^{(l_{k_0})}(\kappa_{k_0})>0$;
\item [(ii).] $M_1<\al_{i_0}$ and $a_1p^{(m_{i_0})}(\al_{i_0})>0$, $N_1<\xi_{j_0}$ and $b_2q^{(n_{j_0})}(\xi_{j_0})>0$, and $Q_2>\kappa_{k_0+1}$ and $c_3r^{(l_{k_0})}(\kappa_{k_0})<0$;
\item [(iii).] $M_1<\al_{i_0}$ and $a_1p^{(m_{i_0})}(\al_{i_0})>0$, $N_2>\xi_{j_0+1}$ and $b_2q^{(n_{j_0})}(\xi_{j_0})<0$, and $Q_1<\kappa_{k_0}$ and $c_3r^{(l_{k_0})}(\kappa_{k_0})>0$;
\item [(iv).]$M_1<\al_{i_0}$ and $a_1p^{(m_{i_0})}(\al_{i_0})>0$, $N_2>\xi_{j_0+1}$ and $b_2q^{(n_{j_0})}(\xi_{j_0})<0$, and $Q_2>\kappa_{k_0+1}$ and $c_3r^{(l_{k_0})}(\kappa_{k_0})<0$;
\item [(v).]$M_2>\al_{i_0+1}$ and $a_1p^{(m_{i_0})}(\al_{i_0})<0$, $N_1<\xi_{j_0}$ and $b_2q^{(n_{j_0})}(\xi_{j_0})>0$, and $Q_1<\kappa_{k_0}$ and $c_3r^{(l_{k_0})}(\kappa_{k_0})>0$;
\item [(vi).]$M_2>\al_{i_0+1}$ and $a_1p^{(m_{i_0})}(\al_{i_0})<0$, $N_1<\xi_{j_0}$ and $b_2q^{(n_{j_0})}(\xi_{j_0})>0$, and $Q_2>\kappa_{k_0+1}$ and $c_3r^{(l_{k_0})}(\kappa_{k_0})<0$;
\item [(vii).]$M_2>\al_{i_0+1}$ and $a_1p^{(m_{i_0})}(\al_{i_0})<0$,  $N_2>\xi_{j_0+1}$ and $b_2q^{(n_{j_0})}(\xi_{j_0})<0$, and $Q_1<\kappa_{k_0}$ and $c_3r^{(l_{k_0})}(\kappa_{k_0})>0$;
\item [(viii).]$M_2>\al_{i_0+1}$ and $a_1p^{(m_{i_0})}(\al_{i_0})<0$, $N_2>\xi_{j_0+1}$ and $b_2q^{(n_{j_0})}(\xi_{j_0})<0$, and $Q_2>\kappa_{k_0+1}$ and $c_3r^{(l_{k_0})}(\kappa_{k_0})<0$.
\end{itemize}
For fixed $a_2$, $a_3$, $b_1$, $b_3$, $c_1$, $c_2$, if $|a_1|$, $|b_2|$, and $|c_3|$ are sufficiently large, then there exists a forward invariant set of the map \eqref{sysexp1} on which the map is topologically semi-conjugate to the one-sided fullshift on eight symbols and it is chaotic in the sense of Li-Yorke.
\end{theorem}

\begin{remark}
The eight different cases in Theorem \ref{threepositive} are the generalization of the assumptions \eqref{assumption-1-positive} in Theorem \ref{onepositive} and the four situations in the assumptions of Theorem \ref{twopositive}.
\end{remark}

Set $$U:=[\al_{i_0},\al_{i_0+1}]\times[\xi_{j_0},\xi_{j_0+1}]\times[\kappa_{k_0},\kappa_{k_0+1}].$$

Fix a positive constant $\lambda>1$.

\begin{proof}
Fix $a_2$, $a_3$, $b_1$, $b_3$, $c_1$, and $c_2$. By Lemma \ref{monotone}, one could assume that $|a_1|$, $|b_2|$, and $|c_3|$ are sufficiently large such that for $w=(x,y,z)\in U\cap F^{-1}(U)$, one has that $x\in(\al_{i_0},\delta_1]\cup[\delta'_1,\al_{i_0+1})$, $y\in(\xi_{j_0},\delta_2]\cup[\delta'_2,\xi_{j_0+1})$,
and $z\in(\kappa_{k_0},\delta_3]\cup[\delta'_3,\kappa_{k_0+1})$.

In the following discussions, we only study Case (i), other cases could be treated by applying similar methods.

By \eqref{polyineq1}--\eqref{polyineq3}, if $|a_1|$, $|b_2|$, and $|c_3|$ are sufficiently large, then for any $w=(x,y,z)\in U\cap F^{-1}(U)$,
\beq\label{bound-deriv-3-positive}
\frac{|a_1p'(x)|}{|a_1p(x)|}\geq \frac{\lambda+|a_2|+|a_3|}{a_{i_0}-M_1},
\frac{|b_2q'(y)|}{|b_2q(y)|}\geq\frac{\lambda+|b_1|+|b_3|}{\xi_{j_0}-N_1},
\frac{|c_3r'(z)|}{|c_3r(z)|}\geq\frac{\lambda+|c_1|+|c_2|}{\kappa_{k_0}-Q_1}.
\eeq

For any point $w_0=(x_0,y_0,z_0)\in U\cap F^{-1}(U)$, set $w_1=F(w_0)=(x_1,y_1,z_1)$. By \eqref{sysexp1},
\beqq
\left\{
  \begin{array}{ll}
   x_1=a_1p(x_0)+a_2y_0+a_3z_0 \\
    y_1=b_1x_0+b_2q(y_0)+b_3z_0 \\
    z_1=c_1x_0+c_2y_0+c_3r(z_0)
  \end{array}
\right..
\eeqq
So, one has that $a_1p(x_0)=x_1-a_2y_0-a_3z_0$, $b_2q(y_0)=y_1-b_1x_0-b_3z_0$, and $c_3r(z_0)=z_1-c_1x_0-c_2y_0$.
This yields that for any $w=(x,y,z)\in U\cap F^{-1}(U)$,
\beq\label{bound-3-positive}
a_1p(x)\geq\al_{i_0}-M_1,\ b_2q(y)\geq\xi_{j_0}-N_1,\ c_3r(z)\geq\kappa_{k_0}-Q_1.
\eeq
By \eqref{bound-deriv-3-positive} and \eqref{bound-3-positive}, one has that for any $w=(x,y,z)\in U\cap F^{-1}(U)$,
\beq\label{bound-deriv-3-positive-dis}
|a_1p'(x)|\geq\lambda+|a_2|+|a_3|,\ |b_2q'(y)|\geq\lambda+|b_1|+|b_3|,\
|c_3r'(z)|\geq\lambda+|c_1|+|c_2|.
\eeq

Next, it is to show that the map in Case (i) is expansion in distance.

For $v=(v_{01},v_{02},v_{03})^T\in\mathbb{R}^3$, set $|v|=\max_{1\leq i\leq3}|v_{0i}|$.
By \eqref{derivativeexp1},
\beqq
\left\{
  \begin{array}{ll}
v_{11}=a_1p'(x_0)v_{01}+a_2v_{02}+a_3v_{03}\\
v_{12}=b_1v_{01}+b_2q'(y_0)v_{02}+b_3v_{03}\\
v_{13}=c_1v_{01}+c_2v_{02}+c_3r'(z_0)v_{03}
\end{array}
\right..
\eeqq
Without loss of generality, suppose that $|v|=|v_{01}|$. Hence, From \eqref{bound-deriv-3-positive-dis}, it follows that
\beqq
|v_1|=\max_{1\leq i\leq 3}|v_{1i}|\geq|v_{11}|\geq|a_1p'(x_0)||v|-|a_2||v|-|a_3||v|\geq\lambda|v|,
\eeqq
which implies that $\|DF\|\geq\lambda$. So, one has that
the map is expansion on $U\cap F^{-1}(U)$.

Hence, there exist eight disjoint closed subsets $V_1,...,V_8$ of $U$ such that $F(V_i)=U$, $1\leq i\leq8$. The map $F$ on $\cup_{1\leq i\leq8}V_i$ is a coupled-expanding map \cite{ZhangShi,ZhangShiChen2013}. It follows from Theorem 3.1 in \cite{ZhangShiChen2013} that there is a forward invariant set on which the map is topologically semi-conjugate to the one-sided fullshift on eight symbols, and it is chaotic in the sense of Li-Yorke. The proof is completed.
\end{proof}

By the proof of Theorems \ref{threepositive-zero} and \ref{threepositive}, one has the following results:
\begin{theorem} \label{threepositive-1zero-2positive}
Suppose that there are $1\leq i_0<r_1$, $1\leq j_0<r_2$, and
$1\leq k_0<r_3$ such that $0=\al_{i_0}<\al_{i_0+1}$, $0<\xi_{j_0}<\xi_{j_0+1}$,
and $0<\kappa_{k_0}<\kappa_{k_0+1}$, and $m_{i_0}=m_{i_0+1}=1$.
Set
\beqq
N_1:=\max\{b_3\kappa_{k_0},b_1\al_{i_0+1}+b_3\kappa_{k_0},b_3\kappa_{k_0+1},b_1\al_{i_0+1}+b_3\kappa_{k_0+1}\},
\eeqq
\beqq N_2:=\min\{b_3\kappa_{k_0},b_1\al_{i_0+1}+b_3\kappa_{k_0},b_3\kappa_{k_0+1},b_1\al_{i_0+1}+b_3\kappa_{k_0+1}\},
\eeqq
\beqq
Q_1:=\max\{c_2\xi_{j_0},c_1\al_{i_0+1}+c_2\xi_{j_0},c_2\xi_{j_0+1},c_1\al_{i_0+1}+c_2\xi_{j_0+1}\},
\eeqq
\beqq Q_2:=\min\{c_2\xi_{j_0},c_1\al_{i_0+1}+c_2\xi_{j_0},c_2\xi_{j_0+1},c_1\al_{i_0+1}+c_2\xi_{j_0+1}\}.
\eeqq
Consider the following four different cases:

\begin{itemize}
\item [(i).] $a_2\leq0$, $a_3\leq0$, and $a_1p'(0)>0$, $N_1<\xi_{j_0}$ and $b_2q^{(n_{j_0})}(\xi_{j_0})>0$, and $Q_1<\kappa_{k_0}$ and $c_3r^{(l_{k_0})}(\kappa_{k_0})>0$;
\item [(ii).] $a_2\leq0$, $a_3\leq0$, and $a_1p'(0)>0$, $N_1<\xi_{j_0}$ and $b_2q^{(n_{j_0})}(\xi_{j_0})>0$, and $Q_2>\kappa_{k_0+1}$ and $c_3r^{(l_{k_0})}(\kappa_{k_0})<0$;
\item [(iii).] $a_2\leq0$, $a_3\leq0$, and $a_1p'(0)>0$, $N_2>\xi_{j_0+1}$ and $b_2q^{(n_{j_0})}(\xi_{j_0})<0$, and $Q_1<\kappa_{k_0}$ and $c_3r^{(l_{k_0})}(\kappa_{k_0})>0$;
\item [(iv).] $a_2\leq0$, $a_3\leq0$, and $a_1p'(0)>0$, $N_2>\xi_{j_0+1}$ and $b_2q^{(n_{j_0})}(\xi_{j_0})<0$, and $Q_2>\kappa_{k_0+1}$ and $c_3r^{(l_{k_0})}(\kappa_{k_0})<0$.
\end{itemize}
For fixed $a_2$, $a_3$, $b_1$, $b_3$, $c_1$, and $c_2$, if $|a_1|$, $|b_2|$, and $|c_3|$ are sufficiently large, then there exists a forward invariant set of the map \eqref{sysexp1} on which the map is topologically semi-conjugate to the one-sided fullshift on eight symbols and it is chaotic in the sense of Li-Yorke.
\end{theorem}

\begin{theorem} \label{threepositive-2zero-1positive}
Suppose that there are $1\leq i_0<r_1$, $1\leq j_0<r_2$, and
$1\leq k_0<r_3$ such that $0=\al_{i_0}<\al_{i_0+1}$, $0=\xi_{j_0}<\xi_{j_0+1}$,
and $0<\kappa_{k_0}<\kappa_{k_0+1}$, and $m_{i_0}=m_{i_0+1}=n_{j_0}=n_{j_0+1}=1$.  Suppose that
\begin{itemize}
\item [(1).] $a_2\leq0$, $a_3\leq0$, and $a_1p'(0)>0$;
\item [(2).] $b_1\leq0$, $b_3\leq0$, and $b_2q'(0)>0$;
\item [(3).] $\max\{0,c_1\al_{i_0+1},c_2\xi_{j_0+1},c_1\al_{i_0+1}+c_2\xi_{j_0+1}\}<\kappa_{k_0}$ and  $c_3r^{(l_{k_0})}(\kappa_{k_0})>0$.
\end{itemize}
For fixed $a_2$, $a_3$, $b_1$, $b_3$, $c_1$, $c_2$, if $|a_1|$, $|b_2|$, and $|c_3|$ are sufficiently large, then there exists a forward invariant set of the map \eqref{sysexp1} on which the map is topologically semi-conjugate to the one-sided fullshift on eight symbols and it is chaotic in the sense of Li-Yorke.
\end{theorem}

\begin{remark}
Similar results could be obtained if the polynomials $p(x)$, $q(y)$, and $r(z)$ have two different non-positive roots, respectively.
\end{remark}
\bigskip

\section{Smale horseshoe in high-dimensional polynomial maps}

In this section, we study the existence of Smale horseshoe and uniformly hyperbolic invariant sets of the maps \eqref{generalmap}, we generalize the well-known results obtained in \cite{devaneynitecki}.

The derivative of the map \eqref{generalmap} is
\begin{equation}
DF=\mbox{diag}[-2x_{1},-2x_{2},...,-2x_{d},\underbrace{0,...,0}_{n-d}]+\left(
    \begin{array}{cccc}
      0 & a_{12} & \cdots & a_{1n}\\
      a_{21} & 0 & \cdots & a_{2n} \\
      \vdots & \vdots & \vdots & \vdots \\
 a_{n1} & a_{n2} & \cdots & 0
    \end{array}
  \right).
\end{equation}

Now, it is to study the inverse of the map $F$. We consider the following three types of maps with the inverse expressions.

Case {\rm{(1)}}. Suppose that $n\geq2d$, and \eqref{generalmap} can be written as follows:
\begin{align}\label{inversemap1}
&\left(
\begin{array}{c}
f_1(x_1,...,x_n)\\
\vdots\\
f_d(x_1,...,x_n)\\
f_{d+1}(x_1,...,x_n)\\
\vdots\\
f_n(x_1,...,x_n)
\end{array}
\right)
=
\left(
\begin{array}{c}
a_{11}-x_1^2\\
\vdots\\
a_{dd}-x_d^2\\
0\\
\vdots\\
0
\end{array}
\right)+\bordermatrix{
~& d & n-2d & d \cr
d & A&B&C\cr
n-d & D&E&0\cr}
\left(
\begin{array}{c}
x_1\\
\vdots\\
x_d\\
x_{d+1}\\
\vdots\\
x_n
\end{array}
\right),
\end{align}
where $C$ and $G=(D,E)$ are invertible matrices.

By direct calculation, one has that the inverse map $H:\mathbb{R}^n\to\mathbb{R}^n$ is
\beq\label{inversemap11}
\left(
\begin{array}{c}
h_1(x_1,...,x_n)\\
\vdots\\
h_{n-d}(x_1,...,x_n)
\end{array}
\right)=G^{-1}
\left(
\begin{array}{c}
x_{d+1}\\
\vdots\\
x_{n}
\end{array}
\right),
\eeq
\beq\label{inversemap12}
\left(
\begin{array}{c}
h_{n-d+1}(x_1,...,x_n)\\
\vdots\\
h_n(x_1,...,x_n)
\end{array}
\right)=
C^{-1} \left(
\begin{array}{c}
x_1\\
\vdots\\
x_d\\
\end{array}
\right)
-
C^{-1}L
\left(
\begin{array}{c}
h_1\\
\vdots\\
h_d\\
\end{array}
\right)
-
C^{-1}
\left(
\begin{array}{c}
a_{11}-h_1^2\\
a_{22}-h_2^2\\
\vdots\\
a_{dd}-h_d^2
\end{array}
\right),
\eeq
where $L=(A,B)$ is a $d\times (n-d)$ matrix. In the above expressions, we should use \eqref{inversemap11} to substitute $h_1,...,h_d$ in the right hand side of \eqref{inversemap12}.

Suppose that
\beq\label{matrix1}
G^{-1}=(D,E)^{-1}=
\bordermatrix{~& n-d\cr
              d & P \cr
              n-2d & Q\cr}.
\eeq
The Jacobian of the inverse map is
\beq\label{inverse-high-i}
DH=\bordermatrix{~ & d & n-d\cr
n-d & 0 & G^{-1}\cr
d &C^{-1} & -C^{-1}LP+S\cr},
\eeq
where
\beqq
S=
C^{-1}
(\mbox{diag}[2h_1,...,2h_d]) P.
\eeqq
\medskip

Case {\rm{(2)}}. Suppose that $\frac{3}{2}d\leq n<2d$, and \eqref{generalmap} can be written as follows:
\begin{align}\label{inversemap2}
&\left(
\begin{array}{c}
f_1(x_1,...,x_n)\\
\vdots\\
f_d(x_1,...,x_n)\\
f_{d+1}(x_1,...,x_n)\\
\vdots\\
f_n(x_1,...,x_n)
\end{array}
\right)
=
\left(
\begin{array}{c}
a_{11}-x_1^2\\
\vdots\\
a_{dd}-x_d^2\\
0\\
\vdots\\
0
\end{array}
\right)+
\bordermatrix{
~& n-d & 2d-n & n-d\cr
2d-n & A& B & 0\cr
n-d & C & D & E\cr
n-d & G & 0&  0\cr
}
\left(
\begin{array}{c}
x_1\\
\vdots\\
x_d\\
x_{d+1}\\
\vdots\\
x_n
\end{array}
\right),
\end{align}
where $B$, $E$, and $G$ are invertible matrices.

The inverse map $H:\mathbb{R}^n\to\mathbb{R}^n$ is as follows:
\beqq
\left(
\begin{array}{c}
h_1\\
\vdots\\
h_{n-d}
\end{array}
\right)=G^{-1}
\left(
\begin{array}{c}
x_{d+1}\\
\vdots\\
x_{n}
\end{array}
\right),
\eeqq
\beqq
\left(
\begin{array}{c}
h_{n-d+1}\\
\vdots\\
h_{d}
\end{array}
\right)=
B^{-1}\left(
\left(
\begin{array}{c}
x_{1}\\
\vdots\\
x_{2d-n}
\end{array}
\right)
-
\left(
\begin{array}{c}
a_{11}-h_1^2\\
\vdots\\
a_{2d-n,2d-n}-h_{2d-n}^2
\end{array}
\right)
-AG^{-1}
\left(
\begin{array}{c}
x_{d+1}\\
\vdots\\
x_n
\end{array}
\right)\right),
\eeqq
\begin{align*}
\left(
\begin{array}{c}
h_{d+1}\\
\vdots\\
h_{n}
\end{array}
\right)
=&E^{-1}\left(
\left(
\begin{array}{c}
x_{2d-n+1}\\
\vdots\\
x_d
\end{array}
\right)
-
\left(
\begin{array}{c}
a_{2d-n+1,2d-n+1}-h_{2d-n+1}^2\\
\vdots\\
a_{dd}-h_d^2
\end{array}
\right)
-C
\left(
\begin{array}{c}
h_{1}\\
\vdots\\
h_{2d-n}
\end{array}
\right)\right)\\
&-E^{-1}
D
\left(
\begin{array}{c}
h_{2d-n+1}\\
\vdots\\
h_d
\end{array}
\right).
\end{align*}
It is evident that the degree of the inverse map is four.

Suppose that
\beqq
G^{-1}=
\bordermatrix{
~ & n-d\cr
2d-n & K  \cr
2n-3d & L   \cr},\ \
E^{-1}=\bordermatrix{
~ & 2n-3d & 2d-n \cr
n-d & M & N \cr
},
\eeqq
and
\beqq
E^{-1}D=\bordermatrix{
~ & 2n-3d & 2d-n \cr
n-d & P & Q \cr
}.
\eeqq

The derivative of the map is
\beq\label{inverse-high-ii}
\bordermatrix{
~ & 2d-n& n-d&n-d\cr
n-d & 0 & 0 & G^{-1}\cr
2d-n & B^{-1} & 0 & S\cr
n-d & V & E^{-1} & W \cr
},
\eeq
where
\beqq
S=B^{-1}(\mbox{diag}[2h_1,...,2h_{2d-n}])K-B^{-1}AG^{-1},
\eeqq
\beqq
V=
-Q
B^{-1}+N (\mbox{diag}[2h_{n-d+1},...,2h_d])B^{-1},
\eeqq
\beqq
W=E^{-1}(\mbox{diag}[2h_{2d-n+1},...,2h_d])\left(\begin{array}{c}
L\\
S
\end{array}
\right)-E^{-1}CK-
E^{-1}D
\left(\begin{array}{c}
L\\
S
\end{array}
\right).
\eeqq
\medskip

Case {\rm{(3)}}. It is to investigate the following type of map, where $n<2d$ and the degree of the inverse map is equal to $2^{2d-n}$.
\begin{align} \label{inversemap3}
&\left(
\begin{array}{c}
f_1(x_1,...,x_n)\\
\vdots\\
f_d(x_1,...,x_n)\\
f_{d+1}(x_1,...,x_n)\\
\vdots\\
f_n(x_1,...,x_n)
\end{array}
\right)
=
\left(
\begin{array}{c}
a_{11}-x_1^2\\
\vdots\\
a_{dd}-x_d^2\\
0\\
\vdots\\
0
\end{array}
\right)+
\bordermatrix{
~ & n-d & 2d-n+1 & n-d-1 \cr
n-d-1 & A &   B     &   C \cr
2d-n+1 & D &  E &     0   \cr
n-d  & G  &  0 &   0  \cr
}
\left(
\begin{array}{c}
x_1\\
\vdots\\
x_d\\
x_{d+1}\\
\vdots\\
x_n
\end{array}
\right),
\end{align}
where $C$ and $G$ are invertible.
In the matrix
$\left(
\begin{array}{c}
B\\
E
\end{array}
\right)$, $\prod\limits^{2d-n}_{i=0} a_{n-d+i,n-d+i+1}\neq0$, all the other coefficients are zero.
Suppose that
\beq\label{matrix2}
D=\left(\begin{array}{c} s_1\\ \vdots\\ s_{2d-n+1}\end{array}\right),
\eeq
where $s_i$ is a $1\times (n-d)$ matrix, $1\leq i\leq 2d-n+1$.

The inverse map $H:\mathbb{R}^n\to\mathbb{R}^n$ is
\beq\label{inverseexpression3}
\left(
\begin{array}{c}
h_1\\
\vdots\\
h_{n-d}
\end{array}
\right)=G^{-1}
\left(
\begin{array}{c}
x_{d+1}\\
\vdots\\
x_{n}
\end{array}
\right),
\eeq
\beqq
\left(
\begin{array}{c}
h_{d+2}\\
\vdots\\
h_{n}
\end{array}
\right)=
C^{-1}\left(
\left(
\begin{array}{c}
x_{1}\\
\vdots\\
x_{n-d-1}
\end{array}
\right)
-
\left(
\begin{array}{c}
a_{11}-h_1^2\\
\vdots\\
a_{n-d-1,n-d-1}-h_{n-d-1}^2
\end{array}
\right)
-A G^{-1}
\left(
\begin{array}{c}
x_{d+1}\\
\vdots\\
x_n
\end{array}
\right)\right);
\eeqq
the map $h_{n-d+i}$ is given recurrently, $1\leq i\leq 2d-n+1$,
\beqq
h_{n-d+1}=\frac{1}{a_{n-d,n-d+1}}\left(x_{n-d}-
(a_{n-d,n-d}-h_{n-d}^2)-s_1\left(  G^{-1}
\left(
\begin{array}{c}
x_{d+1}\\
\vdots\\
x_n
\end{array}
\right)\right)\right),
\eeqq
\begin{align*}
h_{n-d+i}&=\frac{1}{a_{n-d+i-1,n-d+i}}\left(x_{n-d+i-1}-(a_{n-d+i-1,n-d+i-1}- h_{n-d+i-1}^2)
-s_i G^{-1}
\left(
\begin{array}{c}
x_{d+1}\\
\vdots\\
x_n
\end{array}
\right)\right),\\
& 2\leq i\leq 2d-n+1,
\end{align*}
where $s_i$ is specified in \eqref{matrix2}.

The degree of the inverse map is equal to $2^{2d-n}$.

Suppose that
\beqq
G^{-1}=
\bordermatrix{
~ & n-d\cr
n-d-1 & K  \cr
1 & L   \cr}.
\eeqq

The derivative of the inverse map $H$ is
\beq\label{inverse-high-iii}
\bordermatrix{
~ & n-d-1 & 2d-n+1 & n-d\cr
n-d & 0 & 0 &  G^{-1} \cr            2d-n+1 & 0 & S & -DG^{-1} \cr
n-d-1& C^{-1} & 0 & V \cr}
+
\bordermatrix{
~ & n \cr
n-d & 0 \cr
2d-n+1 & SW\left(\begin{array}{ccc}
\frac{\partial h_{n-d}}{\partial x_1} & \cdots &  \frac{\partial h_{n-d}}{\partial x_n}\\
\vdots & \vdots & \vdots\\
\frac{\partial h_{d}}{\partial x_1} & \cdots &  \frac{\partial h_{d}}{\partial x_n}
\end{array}\right)\cr
n-d-1 & 0\cr},
\eeq
where
\beqq
S=\mbox{diag}\bigg[\frac{1}{a_{n-d,n-d+1}},...,\frac{1}{a_{d,d+1}}\bigg],
\eeqq
\beqq
V=C^{-1}( \mbox{diag}[2h_1,...,2h_{n-d-1}]) K-C^{-1}AG^{-1},
\eeqq
\beqq
W=\mbox{diag}[2h_{n-d},...,2h_d],
\eeqq
the derivative of the inverse map is defined recurrently.
\medskip

Set
\beq\label{regionU}
R:=\min_{1\leq i\leq d}\sqrt{a_{ii}},\ U:=[-R,R]^n.
\eeq

In the following discussions, fix a constant $\lambda>1$.
\begin{theorem}\label{hyp-highd-dim}
For the maps in Cases (1)--(3), for fixed $a_{ij}$, $1\leq i\neq j\leq n$, and sufficiently large $a_{11}$,...,$a_{dd}$, there exists a Smale horseshoe and a uniformly hyperbolic invariant set on which the map is topologically conjugate to the two-sided fullshift on $2^d$ symbols. Consequently, the map is chaotic in the sense of both Li-Yorke and Devaney.
\end{theorem}

First, it is to study the point $w_0\in U\cap F(U)$, that is, $w_{-1}=F^{-1}(w_0)\in U$.

\begin{lemma}\label{Case-i-hyp}
In Case (1), for any fixed $a_{ij}$ in \eqref{inversemap1}, $1\leq i\neq j\leq n$, there is a positive constant $N_1$, such that if $\min_{1\leq i\leq d}\{a_{ii}\}\geq N_1$, then for the map \eqref{inversemap1} and any $w_0\in U\cap F(U)$, one has that $m'_{\mathcal{C},w_0}\geq\lambda>1$.
\end{lemma}

\begin{proof}
For the point $w_0\in U\cap F(U)$, set $w_{-1}=F^{-1}(w_0)=H(w_0)\in U$.
By the invertibility of $C$, the expression of $F^{-1}$, and the definition of the region $U$, one has that
\begin{equation}\label{bound1}
C_1 a_{ii}\leq h_i^2\leq C'_1 a_{ii},\ 1\leq i\leq d,
\end{equation}
where $C_1$ and $C'_1$ are positive constants dependent on $a_{ij}$, $1\leq i\neq j\leq n$, and $a_{11},...,a_{dd}$ are sufficiently large.

Consider the following unit cone:
\begin{align*}
K_1(\mathbb{R}^d,\mathbb{R}^{n-d})&=\bigg\{\mathbf{v}=
\left(
\begin{array}{c}
v_0\\
u_0
\end{array}
\right):\
v_0=(v_{0,1},...,v_{0,d})^T\in\mathbb{R}^d,\\
&u_0=(u_{0,1},...,u_{0,n-d})^T\in\mathbb{R}^{n-d},\ \mbox{and}\ |u_0|>|v_0|\bigg\},
\end{align*}
where $|v_0|=\max\{|v_{0,1}|,...,|v_{0,d}|\}$, $|u_0|=\max\{|u_{0,1}|,...,|u_{0,d}|\}$,
and $|\mathbf{v}|=\max\{|v_0|,|u_0|\}$.

Suppose that $DH(\mathbf{v})=
\left(
\begin{array}{c}
v_{-1}\\
u_{-1}
\end{array}
\right)$, where $v_{-1}=(v_{-1,1},...,v_{-1,d})^T\in\mathbb{R}^d$ and $u_{-1}=(u_{-1,1},...,u_{-1,n-d})^T\in\mathbb{R}^{n-d}$.
It follows from \eqref{inverse-high-i} that
\beq\label{inverseexpression1}
v_{-1}=G^{-1}u_{0},\
u_{-1}=C^{-1}v_{0}  -C^{-1}LPu_0+Su_0.
\eeq

To show that $m'_{\mathcal{C},w_0}\geq\lambda$, it is sufficient to show that $|v_{-1}|<|u_{-1}|$ and $|u_{-1}|\geq\lambda|u_0|$. By \eqref{inverseexpression1}, we only need to consider the following part:
\beqq
Su_0=S
\left(
\begin{array}{c}
u_{0,1}\\
u_{0,2}\\
\vdots\\
u_{0,n-d}
\end{array}
\right).
\eeqq
Suppose that $$P=\left(\begin{array}{c}s_1\\
s_2\\ \vdots\\ s_d
\end{array}\right),$$
where $s_i$ is a $1\times (n-d)$ matrix, $1\leq i\leq d$. So, by direct calculations, one has that
\beq\label{bound0}
S
\left(
\begin{array}{c}
u_{0,1}\\
u_{0,2}\\
\vdots\\
u_{0,n-d}
\end{array}
\right)=2C^{-1}
\left(
\begin{array}{c}
s_1
\left(
\begin{array}{c}
x_{d+1}\\
x_{d+2}\\
\vdots\\
x_{n}
\end{array}
\right)
s_1
\left(
\begin{array}{c}
u_{0,1}\\
u_{0,2}\\
\vdots\\
u_{0,n-d}
\end{array}
\right)\\
\vdots\\
s_d
\left(
\begin{array}{c}
x_{d+1}\\
x_{d+2}\\
\vdots\\
x_{n}
\end{array}
\right)
s_d
\left(
\begin{array}{c}
u_{0,1}\\
u_{0,2}\\
\vdots\\
u_{0,n-d}
\end{array}
\right)
\end{array}
\right)=
2C^{-1}
\left(
\begin{array}{c}
h_1s_1
\left(
\begin{array}{c}
u_{0,1}\\
u_{0,2}\\
\vdots\\
u_{0,n-d}
\end{array}
\right)\\
\vdots\\
h_d
s_d
\left(
\begin{array}{c}
u_{0,1}\\
u_{0,2}\\
\vdots\\
u_{0,n-d}
\end{array}
\right)
\end{array}
\right)
.
\eeq
This, together with \eqref{matrix1} and \eqref{bound1}, implies that if $a_{11},...,a_{dd}$ are sufficiently large, then $|u_{-1}|>|v_{-1}|$ and $|u_{-1}|\geq\lambda|u_0|$, yielding that $m'_{\mathcal{C},w_0}\geq\lambda>1$. The proof is completed.
\end{proof}

\begin{lemma}\label{Case-ii-hyp}
In Case (2), for any fixed $a_{ij}$ in \eqref{inversemap2}, $1\leq i\neq j\leq n$, there is a positive constant $N_1$, such that if $\min_{1\leq i\leq d}\{a_{ii}\}\geq N_1$, then for the map \eqref{inversemap2} and any $w_0\in U\cap F(U)$, one has that $m'_{\mathcal{C},w_0}\geq\lambda>1$
\end{lemma}

\begin{proof}
It follows from the invertibility of $B$ and $E$, the expression of the inverse of \eqref{inversemap2}, and the definition of the region $U$, that
\begin{equation}\label{bound2}
C_2 a_{ii}\leq h_i^2\leq C'_2 a_{ii},\ 1\leq i\leq d,
\end{equation}
where $C_2$ and $C'_2$ are positive constants dependent on $a_{ij}$, $1\leq j\leq n$, and $j\neq i$, and $a_{11},...,a_{dd}$ are sufficiently large.

Consider the following unit cone:
\begin{align*}
K_1(\mathbb{R}^d,\mathbb{R}^{n-d})&=\bigg\{\mathbf{v}=
\left(
\begin{array}{c}
v_0\\
u_0
\end{array}
\right):\
v_0=(v_{0,1},...,v_{0,d})^T\in\mathbb{R}^d,\\
&u_0=(u_{0,1},...,u_{0,n-d})^T\in\mathbb{R}^{n-d},\ \mbox{and}\ |u_0|>|v_0|\bigg\},
\end{align*}
where $|v_0|=\max\{|v_{0,1}|,...,|v_{0,d}|\}$, $|u_0|=\max\{|u_{0,1}|,...,|u_{0,d}|\}$,
and $|\mathbf{v}|=\max\{|v_0|,|u_0|\}$.
Assume that $DH(\mathbf{v})=
\left(
\begin{array}{c}
v_{-1}\\
u_{-1}
\end{array}
\right)$, where $v_{-1}=(v_{-1,1},...,v_{-1,d})^T\in\mathbb{R}^d$ and $u_{-1}=(u_{-1,1},...,u_{-1,n-d})^T\in\mathbb{R}^{n-d}$. Suppose that $v_{0}=
\left(
\begin{array}{c}
v^1_{0}\\
v^2_{0}
\end{array}
\right)$ and $v_{-1}=
\left(
\begin{array}{c}
v^1_{-1}\\
v^2_{-1}
\end{array}
\right)$, where the dimension of $v^{1}_0$ and $v^{1}_{-1}$ is $2d-n$.
By \eqref{inverse-high-ii}, one has that
\beq\label{inverseexpression2}
\left(
\begin{array}{c}
v^{1}_{-1}\\
v^2_{-1}\\
u_{-1}
\end{array}
\right)=
\left(
\begin{array}{c}
 G^{-1}u_0\\
B^{-1}v^1_0+Su_0\\
Vv^1_0+E^{-1}v^2_0+Wu_0
\end{array}
\right).
\eeq

To show that $m'_{\mathcal{C},w_0}\geq\lambda$, it is sufficient to show that $|v_{-1}|<|u_{-1}|$ and $|u_{-1}|\geq\lambda|u_0|$.
It follows from \eqref{inverse-high-ii} and \eqref{inverseexpression2}
that we only need to study $Su_0$, $Vv^1_0$, and $Wu_0$. So, we have to consider the following three terms:
\beqq
B^{-1}(\mbox{diag}[2h_1,...,2h_{2d-n}])Ku_0,\
N (\mbox{diag}[2h_{n-d+1},...,2h_d])B^{-1}v^1_0,
\eeqq
\beqq
E^{-1}(\mbox{diag}[2h_{2d-n+1},...,2h_d])\left(\begin{array}{c}
L\\
B^{-1}(\mbox{diag}[2h_1,...,2h_{2d-n}])K
\end{array}
\right)u_0.
\eeqq
Since $G^{-1}=\left(\begin{array}{c} K \\ L \end{array}\right)$, where $K$ is $(2d-n)\times (n-d)$, and $L$ is $(2n-3d)\times (n-d)$, one has that
\begin{align*}
&\bigg\|E^{-1}(\mbox{diag}[2h_{2d-n+1},...,2h_d])\left(\begin{array}{c}
L\\
B^{-1}(\mbox{diag}[2h_1,...,2h_{2d-n}])K
\end{array}
\right)\bigg\|\\
\gtrsim&\bigg\| B^{-1}(\mbox{diag}[2h_1,...,2h_{2d-n}])K\bigg\|,
\end{align*}
if $a_{11},...,a_{dd}$ are sufficiently large. By direct calculation,
\beqq
\left(\begin{array}{c}
L\\
B^{-1}(\mbox{diag}[2h_1,...,2h_{2d-n}])K
\end{array}
\right)
=
\left(\begin{array}{cc}
E_{2n-3d}& 0\\
0& B^{-1}
\end{array}
\right)(\mbox{diag}[\underbrace{1,...,1}_{2n-3d},2h_1,...,2h_{2d-n}])
\left(\begin{array}{c}
L\\
K
\end{array}
\right),
\eeqq
where $E_{2n-3d}=\mbox{diag}[\underbrace{1,...,1}_{2n-3d}]$.
It follows from the assumption $\frac{3}{2}d\leq n$, the fact that $E$ and $G$ are invertible, and \eqref{bound2}, that
\beqq
\|Wu_0\|\geq\|Vv^1_0\|,
\eeqq
if $a_{11},...,a_{dd}$ are sufficiently large.

Hence, $m'_{\mathcal{C},w_0}\geq\lambda>1$ for sufficiently large $a_{11},...,a_{dd}$. This completes the proof.
\end{proof}

\begin{lemma}\label{Case-iii-hyp}
In Case (3), for any fixed $a_{ij}$ in \eqref{inversemap3}, $1\leq i\neq j\leq n$, there is a positive constant $N_1$, such that if $\min_{1\leq i\leq d}\{a_{ii}\}\geq N_1$, then for the map \eqref{inversemap3} and any $w_0\in U\cap F(U)$, one has that $m'_{\mathcal{C},w_0}\geq\lambda>1$.
\end{lemma}

\begin{proof}
By the invertibility of $C$ and $G$, the expression of the inverse of \eqref{inversemap3}, and the definition of the region $U$, one has that
\begin{equation}\label{bound3}
C_3 a_{ii}\leq h_i^2\leq C'_3 a_{ii},\ 1\leq i\leq d,
\end{equation}
where $C_3$ and $C'_3$ are positive constants dependent on $a_{ij}$, $1\leq i\neq j\leq n$, and $a_{11},...,a_{dd}$ are sufficiently large.

Consider the following unit cone:
\begin{align*}
K_1(\mathbb{R}^d,\mathbb{R}^{n-d})&=\bigg\{\mathbf{v}=
\left(
\begin{array}{c}
v_0\\
u_0
\end{array}
\right):\
v_0=(v_{0,1},...,v_{0,d})^T\in\mathbb{R}^d,\\
&u_0=(u_{0,1},...,u_{0,n-d})^T\in\mathbb{R}^{n-d},\ \mbox{and}\ |u_0|>|v_0|\bigg\},
\end{align*}
where $|v_0|=\max\{|v_{0,1}|,...,|v_{0,d}|\}$, $|u_0|=\max\{|u_{0,1}|,...,|u_{0,d}|\}$, and $|\mathbf{v}|=\max\{|v_0|,|u_0|\}$.
Denote $DH(\mathbf{v})=
\left(
\begin{array}{c}
v_{-1}\\
u_{-1}
\end{array}
\right)$, where $v_{-1}=(v_{-1,1},...,v_{-1,d})^T$ and $u_{-1}=(u_{-1,1},...,u_{-1,n-d})^T$. Suppose that $v_{0}=
\left(
\begin{array}{c}
v^1_{0}\\
v^2_{0}
\end{array}
\right)$ with the dimension of $v^1_{0}$ equals to $n-d-1$, $u_{-1}=
\left(
\begin{array}{c}
u^1_{-1}\\
u^2_{-1}
\end{array}
\right)$ with the dimension of $u^{1}_{-1}$ equals to $1$, $v_{-1}=
\left(
\begin{array}{c}
v^1_{-1}\\
v^2_{-1}
\end{array}
\right)$ with the dimension of $v^{1}_{-1}$ equals to $n-d$.

By \eqref{inverse-high-iii}, one has that
\beq\label{matrix3}
\left(
\begin{array}{c}
v^{1}_{-1}\\
\hline
v^2_{-1}\\
u^1_{-1}\\
\hline
u^2_{-1}
\end{array}
\right)=
\left(
\begin{array}{c}
G^{-1}u_0\\
\hline
S
v^2_{0}
-DG^{-1}u_0+SW
\left(
\begin{array}{ccc}
\frac{\partial h_{n-d}}{\partial x_1} & \cdots &  \frac{\partial h_{n-d}}{\partial x_n}\\
\vdots & \vdots & \vdots\\
\frac{\partial h_{d}}{\partial x_1} & \cdots &  \frac{\partial h_{d}}{\partial x_n}
\end{array}
\right)\mathbf{v}
\\
\hline
C^{-1}v^1_0+Vu_0
\end{array}
\right).
\eeq

To show that $m'_{\mathcal{C},w_0}\geq\lambda$, it is sufficient to show that $|v_{-1}|<|u_{-1}|$ and $|u_{-1}|\geq\lambda|u_0|$.
By \eqref{inverse-high-iii} and \eqref{matrix3}, one only need to investigate the expressions for $v^2_{-1}$, $u^1_{-1}$, and $u^2_{-1}$. By \eqref{inverseexpression3} and  \eqref{inverse-high-iii}, if $a_{11}$,...,$a_{dd}$ are sufficiently large, then $|u^1_{-1}|\approx C\prod^{d}_{i=n-d}|h_{i}|$, where $C$ is a constant dependent on $\mathbf{v}$ an the parameter $a_{ij}$, $i\neq j$. This, together with the expression for $u^2_{1}$ and \eqref{bound3}, implies that $|v_{-1}|<|u_{-1}|$ and $|u_{-1}|\geq\lambda|u_0|$. This completes the whole proof.
\end{proof}

Now, it is to study $U\cap F^{-1}(U)$.

\begin{lemma}\label{positive-hyp}
For any fixed $a_{ij}$, $1\leq i\neq j\leq n$, there is a positive constant $N_1$, such that if $\min_{1\leq i\leq d}\{a_{ii}\}\geq N_1$, then for the map \eqref{generalmap} and any $w_0\in U\cap F^{-1}(U)$, one has that $m_{\mathcal{C},w_0}\geq\lambda>1$.
\end{lemma}

\begin{proof}
Suppose that $w_0=(x_{0,1},...,x_{0,n})\in U$, $w_1=F(w_0)=(x_{1,1},...,x_{1,n})$, and $w_{-1}=F^{-1}(w_0)=(x_{-1,1},...,x_{-1,n})$.
Since $w_0\in U\cap F^{-1}(U)$, one has that $w_1=F(w_0)\in U$, that is,
\beqq
-R\leq a_{ii}-x_i^2+\sum_{1\leq j\leq n,\ j\neq i}a_{ij}x_{j}\leq R,\ 1\leq i\leq d,
\eeqq
this, together with the fact that $|x_k|\leq R$, $1\leq k\leq n$, and the definition of $R$ in \eqref{regionU}, implies that
\beq\label{boundfor}
C_4 a_{ii}\leq x^2_{i}\leq C'_4 a_{ii},\ 1\leq i\leq d,
\eeq
where $C_4$ and $C'_4$ are positive constants dependent on $a_{ij}$, $1\leq j\leq n$, and $j\neq i$, and $a_{11},...,a_{dd}$ are sufficiently large.

Consider the unit cone
\begin{align*}
K_1(\mathbb{R}^d,\mathbb{R}^{n-d})&=\bigg\{\mathbf{v}=
\left(
\begin{array}{c}
v_0\\
u_0
\end{array}
\right):\
v_0=(v_{0,1},...,v_{0,d})^T\in\mathbb{R}^d,\\
&u_0=(u_{0,1},...,u_{0,n-d})^T\in\mathbb{R}^{n-d},\ \mbox{and}\ |u_0|\leq|v_0|\bigg\},
\end{align*}
where $|v_0|=\max\{|v_{0,1}|,...,|v_{0,d}|\}$, $|u_0|=\max\{|u_{0,1}|,...,|u_{0,d}|\}$, and $|\mathbf{v}|=\max\{|v_0|,|u_0|\}$.

Without loss of generality, suppose that $|v_0|=|v_{0,1}|$ and $|u_0|=|u_{0,1}|$. Set $DF(\mathbf{v})=
\left(
\begin{array}{c}
v_1\\
u_1
\end{array}
\right)$, where $v_1=(v_{1,1},...,v_{1,d})^T\in\mathbb{R}^d$,
$u_1=(u_{1,1},...,u_{1,n-d})^T\in\mathbb{R}^{n-d}$. So, if $a_{11}$ is sufficiently large, by \eqref{boundfor}, then
\begin{align*}
&|v_1|\geq|v_{1,1}|\geq 2|x_1||v_{0,1}|-\sum_{2\leq j\leq d} |a_{1j}||v_{0,j}|-\sum_{d+1\leq j\leq n}|a_{1j}||u_{0,j-d}|\\
&\geq
2|x_1||v_{0,1}|-\sum_{2\leq j\leq n} |a_{1j}||v_{0,1}|\geq \lambda\max_{d+1\leq i\leq n}\bigg(\sum_{1\leq j\leq n,\ j\neq i}|a_{ij}|\bigg)|v_{0,1}|\geq\lambda|u_1|,
\end{align*}
implying that $m_{\mathcal{C},w_0}\geq\lambda>1$. This completes the proof.
\end{proof}

\begin{lemma}\label{high-dim-hyp-set}
In Cases (1)--(3), for fixed parameters $a_{ij}$, $1\leq i\neq j\leq n$, and sufficiently large $a_{11}$,...,$a_{dd}$, the invariant set $\Lambda=\cap^{\infty}_{i=-\infty}F^{i}(U)$ is uniformly hyperbolic.
\end{lemma}

\begin{proof}
This can be obtained by applying Lemmas \ref{unifhyp}, \ref{Case-i-hyp}--\ref{positive-hyp}. \end{proof}

\begin{lemma} \label{components}
In Cases (1)--(3), for fixed $a_{ij}$, $1\leq i\neq j\leq n$, if $a_{11},...,a_{dd}$ are sufficiently large, then $F(U)\cap U$ and $F^{-1}(U)\cap U$ have $2^{d}$ connected components, respectively.
\end{lemma}

\begin{proof}
It is to show that $F(U)\cap U$ has $2^{d}$ components.

Fix any $i$, $1\leq i\leq d$, and fix any $x_{j}$ with $x_j\in[-R,R]$ and $j\neq i$, where $R$ is specified in \eqref{regionU}. So, the function $f_i(x_1,...,x_n)=a_{ii}-x_i^2+\sum_{1\leq j\leq n,\ j\neq i}a_{ij}x_{j}$ can be rewritten as follows:
\beqq
f_i(x_1,...,x_{i-1},y,x_{i+1},...,x_n)=a_{ii}-y^2+\sum_{1\leq j\leq n,\ j\neq i}a_{ij}x_{j},\ y\in[-R,R],\ 1\leq i\leq d.
\eeqq
By the definition of $R$ in \eqref{regionU}, if $a_{ii}$ is sufficiently large, then $ f_i(...,y,...)\approx a_{ii}-y^2+Ca^{1/2}_{ii}\approx a_{ii}-y^2$,
which implies  $\mbox{Graph}(y,f_i(...,y,...))\cap([-R,R]\times[-R,R])$ has two connected components for sufficiently large $a_{ii}$. Therefore, one has that $F(U)\cap U$ has $2^{d}$ components for sufficiently large $a_{ii}$, $1\leq i\leq d$.

Since the map $F$ is diffeomorphic, one has that $F^{-1}(U)\cap U=F^{-1}(F(U)\cap U)$ has $2^{d}$ components. This completes the proof.
\end{proof}

Therefore, it follows from Lemmas \ref{chaoticshift},  \ref{high-dim-hyp-set}, and \ref{components} that Theorem \ref{hyp-highd-dim} holds.

\begin{remark}
Similar results can be obtained, if we substitute the functions $\Delta_i(x)$ by general polynomials. For example, we could consider the following type of polynomial maps:
\begin{equation}\label{generalmap111}
\left\{
\begin{array}{ll}
f_1(x_1,...,x_n)=a_{11}x_{1}(1-x_1)+\sum_{2\leq j\leq n} a_{1j}x_j\\
\vdots \qquad  \qquad \vdots \qquad  \qquad \vdots \qquad  \qquad \vdots\\
f_d(x_1,...,x_n)=a_{dd}x_{d}(1-x_d)+\sum_{1\leq j\leq n,\ j\neq d}a_{dj}x_{j}\\
f_{d+1}(x_1,...,x_n)=\sum_{1\leq j\leq n,\ j\neq d+1}a_{d+1,j}x_{j}\\
\vdots \qquad  \qquad \vdots \qquad  \qquad \vdots \qquad  \qquad \vdots\\
f_n(x_1,...,x_n)=\sum_{1\leq j\leq n-1}a_{nj}x_j
\end{array} \right.,
\end{equation}
where $a_{ij}$, $1\leq i,j\leq n$, are parameters.
\end{remark}

\begin{remark}
It is an interesting question to find the general inverse maps for the following type of maps:
\begin{equation}\label{generalmap11}
\left\{
\begin{array}{ll}
f_1(x_1,...,x_n)=a_{11}(x_{1}-b_{11})\cdots(x_1-b_{1s_1})+\sum_{2\leq j\leq n} a_{1j}x_j\\
\vdots \qquad  \qquad \vdots \qquad  \qquad \vdots \qquad  \qquad \vdots\\
f_d(x_1,...,x_n)=a_{dd}(x_{d}-b_{d1})\cdots(x_d-b_{ds_d})+\sum_{1\leq j\leq n,\ j\neq d}a_{dj}x_{j}\\
f_{d+1}(x_1,...,x_n)=\sum_{1\leq j\leq n,\ j\neq d+1}a_{d+1,j}x_{j}\\
\vdots \qquad  \qquad \vdots \qquad  \qquad \vdots \qquad  \qquad \vdots\\
f_n(x_1,...,x_n)=\sum_{1\leq j\leq n-1}a_{nj}x_j
\end{array} \right.,
\end{equation}
where $a_{ij}$ and $b_{kl}$ are parameters. We guess that we could apply similar discussions above to show that for fixed $a_{ij}$, $1\leq i\neq j\leq n$ and sufficiently large $a_{11},...,a_{dd}$, there might exist a Smale horseshoe and  a uniformly hyperbolic invariant set on which the map is topologically conjugate with the fullshift on $s_1s_2\cdots s_d$ symbols.
\end{remark}

\bigskip

\section{Existence of strange attractors by simulations}\label{attractor}

In this section, many interesting maps with strange attractors are collected.
The Mathematical software Maple is applied to find the strange attractors and the corresponding maximal Lyapunov exponent is calculated. The maximal Lyapunov exponents of these maps are all positive, yielding that these maps have complicated dynamical behavior. For the calculation of the maximal Lyapunov exponents, please refer to \cite{ChoeErgodic}.
For the simulation of these strange attractors, the initial value is taken as $(0.5,0.5,0.5)$,
the first $80000$ points are omitted and the next $200000$ points are kept. Figures \ref{FigureMap2}-\ref{FigureMap33} give the simulation graphs of the maps (2), (8), (11), (20), (24), and (33) from the following table.

\begin{center}
\begin{longtable}{|c|c|c|c|c|}
\hline
  \multicolumn{4}{|c|}{Strange Attractors} \\
\hline
\ No. &\ Parameters & \ Polynomial &\ Max Lya Exp \\
\hline
 & $a_1=2.6$, $a_2=0.3$, $a_3=0$, & p(x)= x(1-x)&  \\
1. & $b_1=0.4$, $b_2=1.7$, $b_3=0.5,$ & q(y)= y(1-y)& 0.109706\\
& $c_1=1$, $c_2=0$, $c_3=0$ & r(z)=0 &  \\
\hline
& $a_1=2.8$, $a_2=0.2$, $a_3=0.6,$ & p(x)=x(1-x) & \\
2. & $b_1=0.6$, $b_2=2.4$, $b_3=0.1$, & q(y)=y(1-y) & 0.322432\\
& $c_1=1$, $c_2=0$, $c_3=0$ & r(z)=0 &\\
\hline
& $a_1=4.5$, $a_2=0.1$, $a_3=0$, & $p(x)=x(1-x)^2$ & \\
3. & $b_1=0.3$, $b_2=7$, $b_3=0.1$, & $q(y)=y(1-y)^2$ & 0.399491\\
& $c_1=1$, $c_2=0$, $c_3=0$ & $r(z)=0$ &\\
\hline
& $a_1=5.9$, $a_2=0.1$, $a_3=0$, & $p(x)=x^2(1-x)$ & \\
4. & $b_1=0.3$, $b_2=1.3$, $b_3=0.1$, & $q(y)=y^2(1-y)$ & 0.418139\\
& $c_1=1$, $c_2=0$, $c_3=0$ & r(z)=0 &\\
\hline
& $a_1=5.9$, $a_2=0.1$, $a_3=0$, & $p(x)=x^2(1-x)$ & \\
5. & $b_1=0.3$, $b_2=4.7$, $b_3=0.1$, & $q(y)=y^2(1-y)$ & 0.476862\\
& $c_1=1$, $c_2=0$, $c_3=0$ & $r(z)=0$ &\\
\hline
& $a_1=8$, $a_2=0.1$, $a_3=0$, & $p(x)=x(1-x)^3$ &\\
6. & $b_1=0.3$, $b_2=1.3$, $b_3=0.7$, & $q(y)=y(1-y)^3$ & 0.312734\\
& $c_1=1$, $c_2=0$, $c_3=0$ & $r(z)=0$ &\\
\hline
 & $a_1=13$, $a_2=0.1$, $a_3=0,$ & $p(x)=x^2(1-x)^2$ & \\
7. &$b_1=0.3$, $b_2=8$, $b_3=0.1$, & $q(y)=y^2(1-y)^2$ & 0.427413 \\
& $c_1=1$, $c_2=0$, $c_3=0$ & $r(z)=0$ &\\
\hline
& $a_1=10$, $a_2=0.1$, $a_3=0$, & $p(x)=x(1-x)^4$ &\\
8. &$b_1=0.7$, $b_2=10$, $b_3=0.2$ & $q(y)=y(1-y)^4$ & 0.304944 \\
& $c_1=1$, $c_2=0$, $c_3=0$ & $r(z)=0$ &\\
\hline
 & $a_1=28$, $a_2=0.1$, $a_3=0$, & $p(x)=x^2(1-x)^3$ & \\
9. & $b_1=0.3$, $b_2=20$, $b_3=0.1$, & $q(y)=y^2(1-y)^3$ & 0.383889 \\
& $c_1=1$, $c_2=0$, $c_3=0$ & $r(z)=0$ &  \\
\hline
& $a_1=20$, $a_2=0.4$, $a_3=0$, & $p(x)=x^3(1-x)^2$ & \\
10.& $b_1=0.3$, $b_2=20$, $b_3=0.1$ & $q(y)=y^3(1-y)^2$ & 0.434574  \\
& $c_1=1$, $c_2=0$, $c_3=0$ & $r(z)=0$ & \\
\hline
 & $a_1=34.5$, $a_2=0.4$, $a_3=0$, & $p(x)=x^2(1-x)^4$ &\\
11. & $b_1=0.5$, $b_2=34$, $b_3=0.2$, & $q(y)=y^2(1-y)^4$ & 0.280940 \\
& $c_1=1$, $c_2=0$, $c_3=0$ & $r(z)=0$ &\\
\hline
 & $a_1=30$, $a_2=0.1$, $a_3=0$, & $p(x)=x^2(1-x)^4$ & \\
12 & $b_1=0.2$, $b_2=23$, $b_3=0.2$, & $q(y)=y^2(1-y)^4$ & 0.330790\\
& $c_1=1$, $c_2=0$, $c_3=0$ & $r(z)=0$ &\\
\hline
& $a_1=37$, $a_2=0.2$, $a_3=0$ & $p(x)=x^3(1-x)^3$, &\\
13. & $b_1=0.6$, $b_2=36$, $b_3=0.2$, & $q(y)=y^3(1-y)^3$ & 0.231571\\
& $c_1=1$, $c_2=0$, $c_3=0$ & $r(z)=0$ &\\
\hline
& $a_1=29$, $a_2=0.2$, $a_3=0$, & $p(x)=x^4(1-x)^2$ &\\
14. & $b_1=0.6$, $b_2=22$, $b_3=0.2$, & $q(y)=y^4(1-y)^2$ &0.367095\\
& $c_1=1$, $c_2=0$, $c_3=0$ & $r(z)=0$ &\\
\hline
 & $a_1=65$, $a_2=0.3$, $a_3=0$, & $p(x)=x^3(1-x)^4$ &\\
15. & $b_1=0.5$, $b_2=50$, $b_3=0.3$, & $q(y)=y^3(1-y)^4$ & 0.252414\\
& $c_1=1$, $c_2=0$, $c_3=0$ & $r(z)=0$ &\\
\hline
& $a_1=180$, $a_2=0.3$, $a_3=0$, & $p(x)=x^4(1-x)^4$ &\\
16. & $b_1=0.2$, $b_2=160$, $b_3=0.3$, & $q(y)=y^4(1-y)^4$ &0.399516\\
& $c_1=1$, $c_2=0$, $c_3=0$ & $r(z)=0$ &\\
\hline
 &$a_1=390$, $a_2=0.3$, $a_3=0$, & $p(x)=x^4(1-x)^5$ &\\
17. &$b_1=0.5$, $b_2=250$, $b_3=0.1$, & $q(y)=y^4(1-y)^5$ &0.358692\\
&$c_1=1$, $c_2=0$, $c_3=0$ & $r(z)=0$& \\
\hline
& $a_1=19000$, $a_2=0.1$, $a_3=0$, & $p(x)=x^8(1-x)^7$ &\\
18.& $b_1=0.3$, $b_2=800$, $b_3=0.7$, & $q(y)=y^6(1-y)^7$ &0.542847\\
& $c_1=1$, $c_2=0$, $c_3=0$ & $r(z)=0$ &\\
\hline
& $a_1=2.5$, $a_2=0.2$, $a_3=0$, & $p(x)=(x+1)x(1-x)$ &\\
19. &$b_1=0.2$, $b_2=2.7$, $b_3=0.1$, & $q(y)=(y+1)y(1-y)$ & 0.210240\\
&$c_1=1$, $c_2=0$, $c_3=0$ & $r(z)=0$ &\\
\hline
 & $a_1=0$, $a_2=1$, $a_3=0$, & $p(x)=0$ &\\
20.& $b_1=0$, $b_2=0$, $b_3=1$, & $q(y)=0$ &0.175343\\
& $c_1=0.6$, $c_2=0.6$, $c_3=1$ & $r(z)=(z+1)^2(1-z)-1$ &\\
\hline
 & $a_1=0$, $a_2=1$, $a_3=0$, & $p(x)=0$ &\\
21.& $b_1=0$, $b_2=0$, $b_3=1$, & $q(y)=0$ &0.401407\\
& $c_1=0.1$, $c_2=0.2$, $c_3=14.5$ & $r(z)=z^2(1-z)^2$ &\\
\hline
 & $a_1=0$, $a_2=1$, $a_3=0$, & $p(x)=0$ &\\
22.& $b_1=0$, $b_2=0$, $b_3=1$, & $q(y)=0$ &0.250262\\
& $c_1=0.3$, $c_2=0.4$, $c_3=14.1$ & $r(z)=z^3(1-z)^2$ &\\
\hline
 & $a_1=0$, $a_2=0.7$, $a_3=0$, & $p(x)=0$ &\\
23.& $b_1=0$, $b_2=0$, $b_3=0.7$, & $q(y)=0$ &0.492135\\
& $c_1=0.2$, $c_2=0$, $c_3=23$ & $r(z)=z^3(1-z)^2$ &\\
\hline
  & $a_1=0$, $a_2=0.8$, $a_3=0$, & $p(x)=0$ &\\
24.& $b_1=0$, $b_2=0$, $b_3=0.7$, & $q(y)=0$ &0.510269\\
& $c_1=0.2$, $c_2=0$, $c_3=53$ & $r(z)=z^2(1-z)^4$ &\\
\hline
 & $a_1=0$, $a_2=0.8$, $a_3=0$, & $p(x)=0$ &\\
25.& $b_1=0$, $b_2=0$, $b_3=0.7$, & $q(y)=0$ &0.279378\\
& $c_1=0.3$, $c_2=0.3$, $c_3=440$ & $r(z)=z^4(1-z)^5$ &\\
\hline
  & $a_1=2$, $a_2=0.3$, $a_3=0$, & $p(x)=x(1-x)$ &\\
26.& $b_1=0.4$, $b_2=2$, $b_3=0.5$, & $q(y)=y(1-y)$ &0.113261\\
& $c_1=0.5$, $c_2=0$, $c_3=2$ & $r(z)=z(1-z)$ &\\
\hline
 & $a_1=3.5$, $a_2=0.1$, $a_3=0$, & $p(x)=x(1-x)$ &\\
27.& $b_1=0.3$, $b_2=2.3$, $b_3=0.2$, & $q(y)=y(1-y)$ &0.336029\\
& $c_1=0.2$, $c_2=0$, $c_3=3$ & $r(z)=z(1-z)$ &\\
\hline
 & $a_1=3.2$, $a_2=0.1$, $a_3=0$, & $p(x)=x(1-x)$ &\\
28.& $b_1=0.3$, $b_2=2.5$, $b_3=0.2$, & $q(y)=y(1-y)$ &0.0399768\\
& $c_1=0.2$, $c_2=0$, $c_3=2.6$ & $r(z)=z(1-z)$ &\\
\hline
 & $a_1=3.2$, $a_2=0.1$, $a_3=0$, & $p(x)=x(1-x)$ &\\
29.& $b_1=0.3$, $b_2=2.5$, $b_3=0.2$, & $q(y)=y(1-y)$ &0.0971938\\
& $c_1=0.2$, $c_2=0$, $c_3=3$ & $r(z)=z(1-z)$ &\\
\hline
  & $a_1=3.2$, $a_2=0.1$, $a_3=0$, & $p(x)=x(1-x)$ &\\
30.& $b_1=0.3$, $b_2=2.5$, $b_3=0.2$, & $q(y)=y(1-y)$ &0.198463\\
& $c_1=0.2$, $c_2=0$, $c_3=3.2$ & $r(z)=z(1-z)$ &\\
\hline
& $a_1=3.2$, $a_2=0.1$, $a_3=0$, & $p(x)=x(1-x)$ &\\
31.& $b_1=0.3$, $b_2=2.5$, $b_3=0.2$, & $q(y)=y(1-y)$ &0.226691\\
& $c_1=0.2$, $c_2=0$, $c_3=3.3$ & $r(z)=z(1-z)$ &\\
\hline
 & $a_1=3.2$, $a_2=0.1$, $a_3=0$, & $p(x)=x(1-x)$ &\\
32.& $b_1=0.3$, $b_2=2.5$, $b_3=0.2$, & $q(y)=y(1-y)$ &0.293952\\
& $c_1=0.2$, $c_2=0$, $c_3=3.4$ & $r(z)=z(1-z)$ &\\
\hline
  & $a_1=15$, $a_2=0.2$, $a_3=0.1$, & $p(x)=x^2(1-x)^4$ &\\
33.& $b_1=0.3$, $b_2=2.6$, $b_3=0.2$, & $q(y)=y(1-y)$ &0.149243\\
& $c_1=0.6$, $c_2=0.2$, $c_3=20$ & $r(z)=z^2(1-z)^3$ &\\
\hline
  & $a_1=20$, $a_2=0.3$, $a_3=0.2$, & $p(x)=x^2(1-x)^4$ &\\
34.& $b_1=0.2$, $b_2=2.5$, $b_3=0.2$, & $q(y)=y(1-y)$ &0.183814\\
& $c_1=0.2$, $c_2=0.2$, $c_3=25$ & $r(z)=z^2(1-z)^3$ &\\
\hline
  & $a_1=4$, $a_2=0.4$, $a_3=0.3$, & $p(x)=x^2(1-x)^3$ &\\
35.& $b_1=0.2$, $b_2=90$, $b_3=0.2$, & $q(y)=y^3(1-y)^4$ &0.333486\\
& $c_1=0.2$, $c_2=0.2$, $c_3=1760$ & $r(z)=z^5(1-z)^6$ &\\
\hline
  & $a_1=4$, $a_2=0.4$, $a_3=0.3$, & $p(x)=x^2(1-x)^3$ &\\
36.& $b_1=0.2$, $b_2=60$, $b_3=0.2$, & $q(y)=y^3(1-y)^4$ &0.206663\\
& $c_1=0.2$, $c_2=0.2$, $c_3=1700$ & $r(z)=z^5(1-z)^6$ &\\
\hline
\end{longtable}
\end{center}
\bigskip

\section{Example}

In this section, two examples are provided to illustrate the theoretical results obtained in Theorems \ref{onepositive} and \ref{twopositive}. All the illustration graphs in this section are drawn by using the software Mathematica. The interested readers can make simple programs to run some softwares and obtain more interesting graphs.

\begin{example} \label{illustr2}
Consider the following example:
\begin{equation}
\left\{
  \begin{array}{ll}
    f_1(x,y,z)=7(x-1)(2-x)+0.2y+z \\
    f_2(x,y,z)=x \\
    f_3(x,y,z)=0.2y,
  \end{array}
\right.
\end{equation}
It is evident that this map satisfies the conditions in Case (i) of Theorem \ref{onepositive}. Set $U=[1,2]\times [1,2]\times[0.2,0.4]$. Figures \ref{Case-ii-positive} and \ref{Case-ii-inverse} provide the graphs of $F(U)\cap U$ and $F^{-1}(U)\cap U$, respectively.  By Theorem \ref{onepositive}, there exists a Smale horseshoe, the invariant set is uniformly hyperbolic and is topologically conjugate to fullshift on two symbols.
\end{example}

\begin{example}\label{illustr1}
Consider the following example:
\begin{equation}
\left\{
  \begin{array}{ll}
    f_1(x,y,z)=7(x-1)(2-x)+0.4y \\
    f_2(x,y,z)=7(y-1)(2-y)+0.4z \\
    f_3(x,y,z)=x,
  \end{array}
\right.
\end{equation}
It is evident that this map satisfies the conditions in Case (i) of Theorem \ref{twopositive}. Set $U=[1,2]\times [1,2]\times[1,2]$. Figures \ref{Caseipositive} and \ref{Caseiinverse} give the graphs of $F(U)\cap U$ and $F^{-1}(U)\cap U$, respectively. From the computer simulations, we find that the graph of $F^{-1}(U)$ is very big compared with $U$. So, we only give a part of the graph of $F^{-1}(U)\cap U$.  It follows from Theorem \ref{twopositive} that there are a Smale horseshoe and the uniformly hyperbolic invariant set on which the map is topologically conjugate to fullshift on four symbols.
\end{example}

\bigskip

\section*{Acknowledgments}

I would like to thank Professor Sheldon Newhouse for his encouragement, comments, and providing many useful references.

I devote this work to my parents, I cannot thank my parents enough for all the support and love they have given me.
\bigskip

\baselineskip=10pt
\bibliography{referhenon}
\bibliographystyle{plain}

\newpage

\centerline{\Large \bf List of Figure Captions}

\baselineskip=16pt
\vspace{0.4 in} \noindent
\begin{enumerate}

\item[Figure 1.] The chaotic attractor of map (2) in Section 5, where the initial value is taken as $(0.5,0.5,0.5)$.

\item[Figure 2.] The chaotic attractor of system map (8) in Section 5, where the initial value is taken as $(0.5,0.5,0.5)$.

\item[Figure 3.] The chaotic attractor of map (11) in Section 5, where the initial value is taken as $(0.5,0.5,0.5)$.

\item[Figure 4.] The chaotic attractor of map (20) in Section 5, where the initial value is taken as $(0.5,0.5,0.5)$.

\item[Figure 5.] The chaotic attractor of map (24) in Section 5, where the initial value is taken as $(0.5,0.5,0.5)$.

\item[Figure 6.] The chaotic attractor of map (33) in Section 5, where the initial value is taken as $(0.5,0.5,0.5)$.

\item[Figure 7.] The illustration graph of $F(U)\cap U$ in Example \ref{illustr2}, where $U=[1,2]\times[1,2]\times[0.2,0.4]$.

\item[Figure 8.]  The illustration graph of $F^{-1}(U)\cap U$ in Example \ref{illustr2}, where $U=[1,2]\times[1,2]\times[0.2,0.4]$.

\item[Figure 9.] The illustration graph of $F(U)\cap U$ in Example \ref{illustr1}, where $U=[1,2]\times[1,2]\times[1,2]$.

\item[Figure 10.] The illustration graph of $F^{-1}(U)\cap U$ in Example \ref{illustr1}, where $U=[1,2]\times[1,2]\times[1,2]$.

\end{enumerate}

\newpage

\begin{figure}[H]
\begin{center}
\scalebox{0.3 }{ \includegraphics{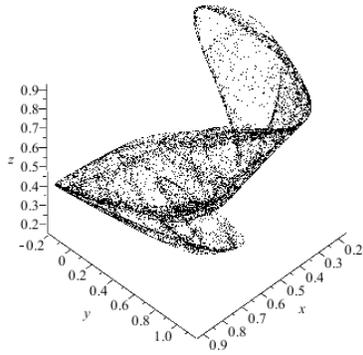}}
\renewcommand{\figure}{Fig.}
\caption{The chaotic attractor of map (2) in Section 5, where the initial value is taken as $(0.5,0.5,0.5)$.
}\label{FigureMap2}
\end{center}
\end{figure}

\begin{figure}[H]
\begin{center}
\scalebox{0.3 }{ \includegraphics{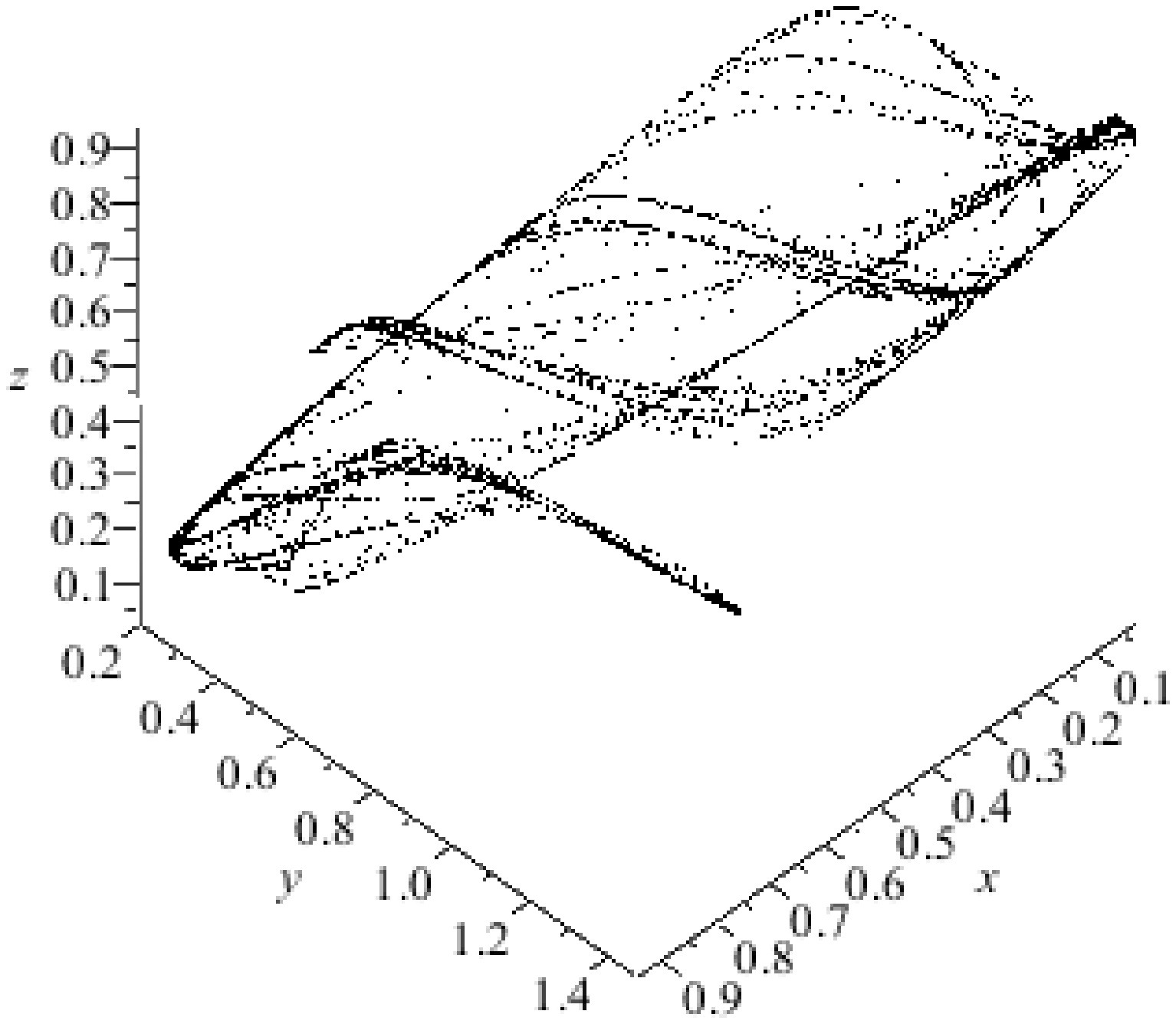}}
\renewcommand{\figure}{Fig.}
\caption{The chaotic attractor of system map (8) in Section 5, where the initial value is taken as $(0.5,0.5,0.5)$.}\label{FigureMap8}
\end{center}
\end{figure}

\begin{figure}[H]
\begin{center}
\scalebox{0.3 }{ \includegraphics{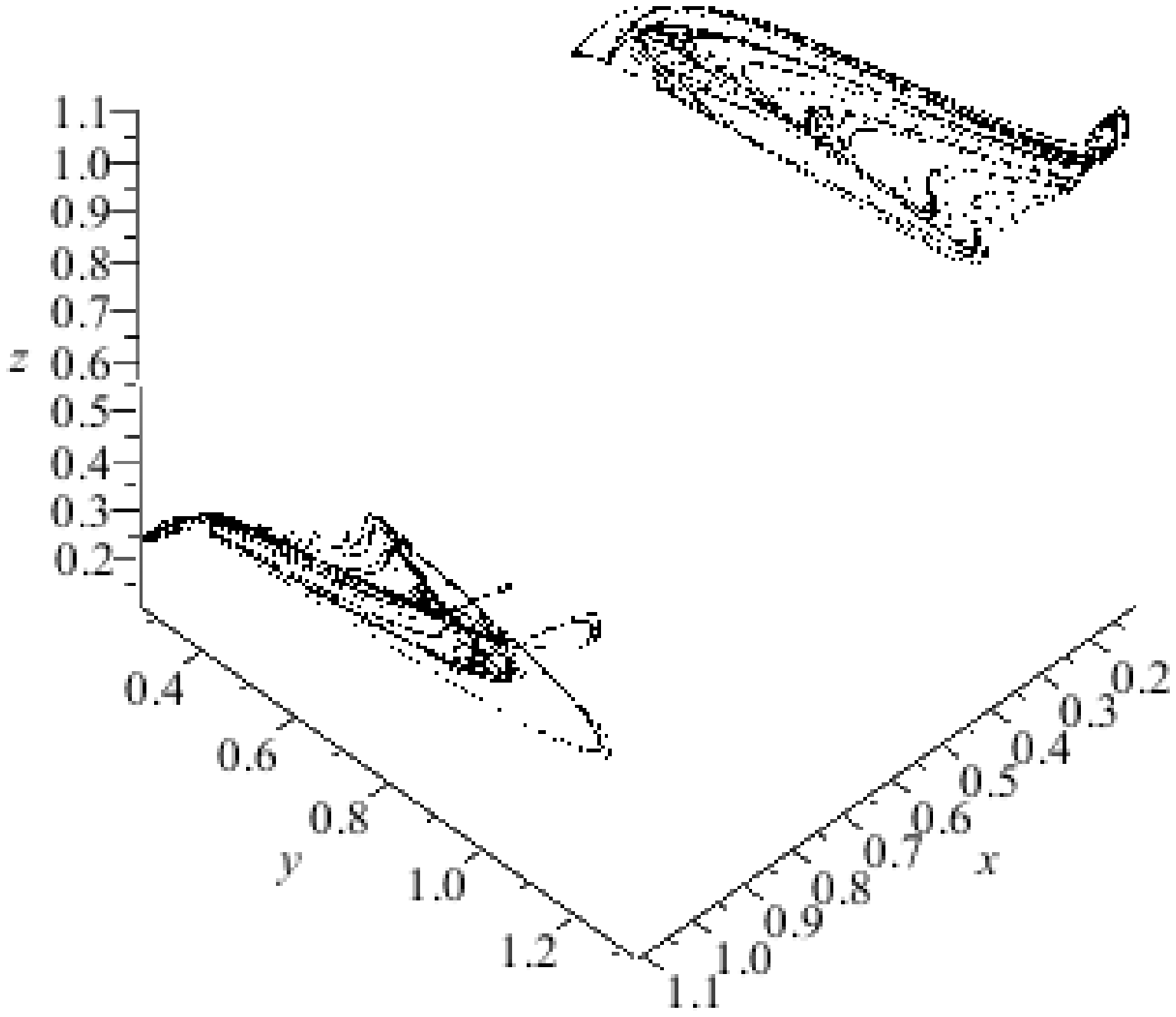}}
\renewcommand{\figure}{Fig.}
\caption{The chaotic attractor of map (11) in Section 5, where the initial value is taken as $(0.5,0.5,0.5)$.
}\label{FigureMap11}
\end{center}
\end{figure}

\begin{figure}[H]
\begin{center}
\scalebox{0.3 }{ \includegraphics{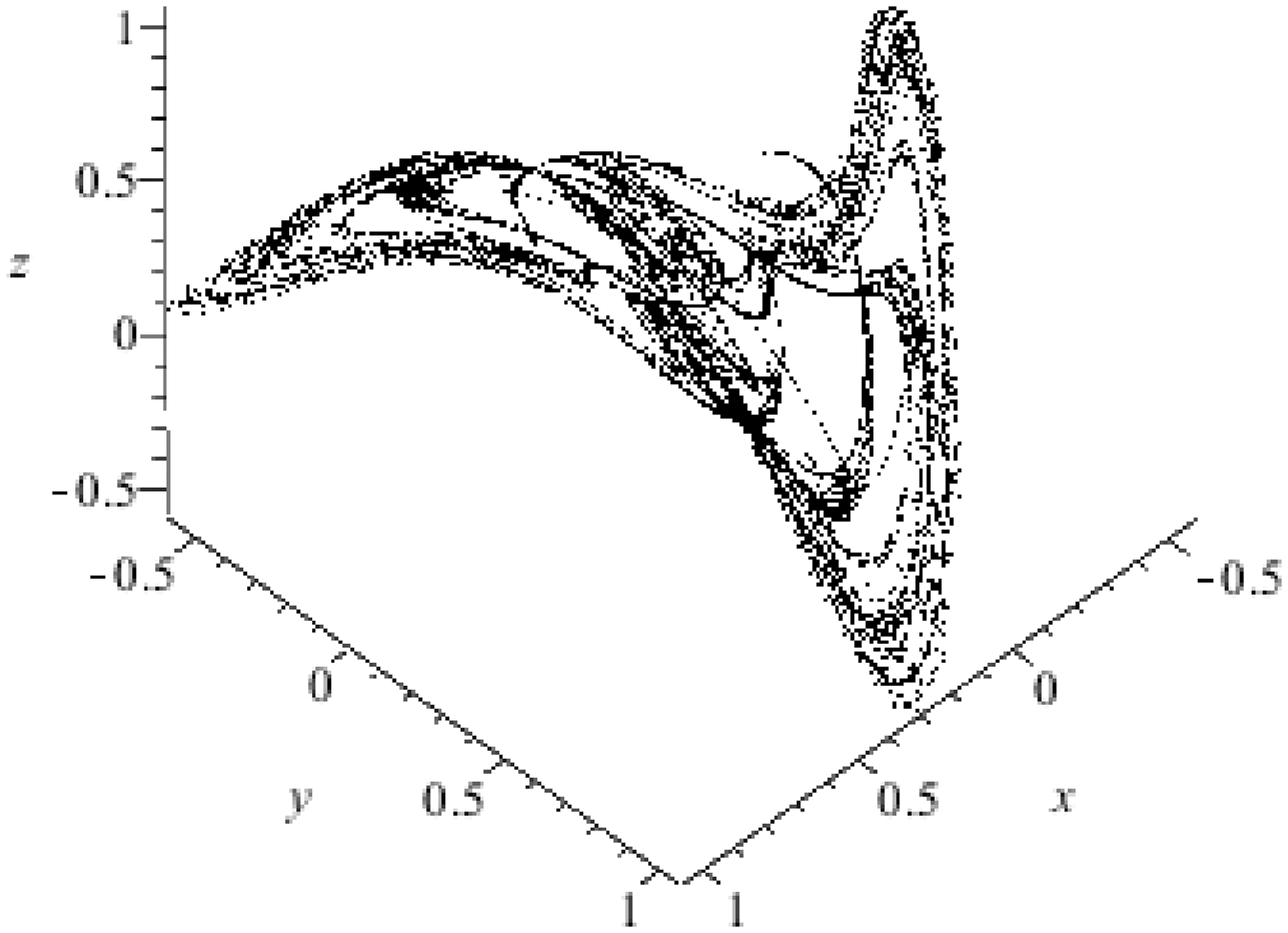}}
\renewcommand{\figure}{Fig.}
\caption{The chaotic attractor of map (20) in Section 5, where the initial value is taken as $(0.5,0.5,0.5)$.}\label{FigureMap20}
\end{center}
\end{figure}

\begin{figure}[H]
\begin{center}
\scalebox{0.3 }{ \includegraphics{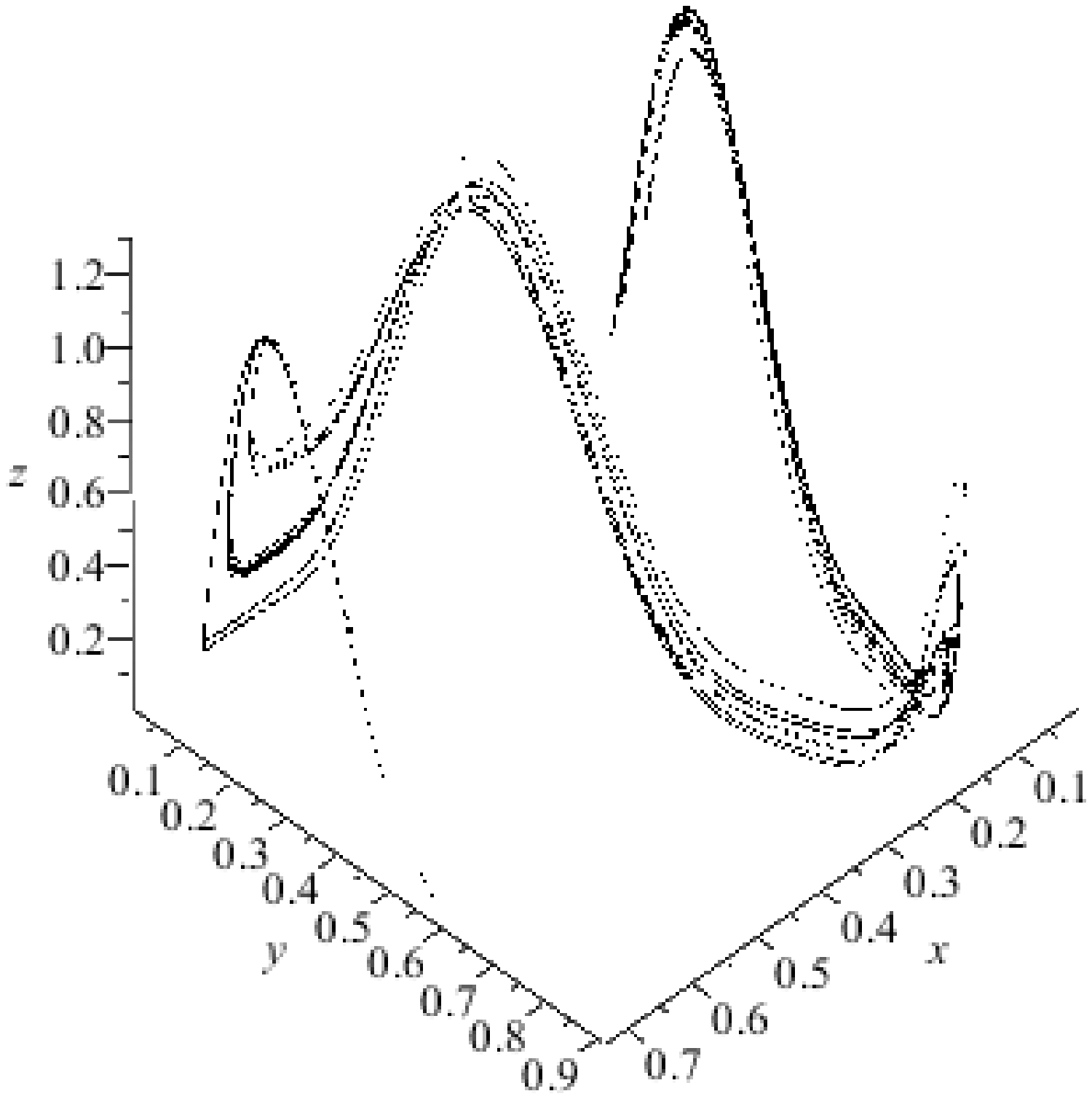}}
\renewcommand{\figure}{Fig.}
\caption{The chaotic attractor of map (24) in Section 5, where the initial value is taken as $(0.5,0.5,0.5)$.}\label{FigureMap24}
\end{center}
\end{figure}

\begin{figure}[H]
\begin{center}
\scalebox{0.3 }{ \includegraphics{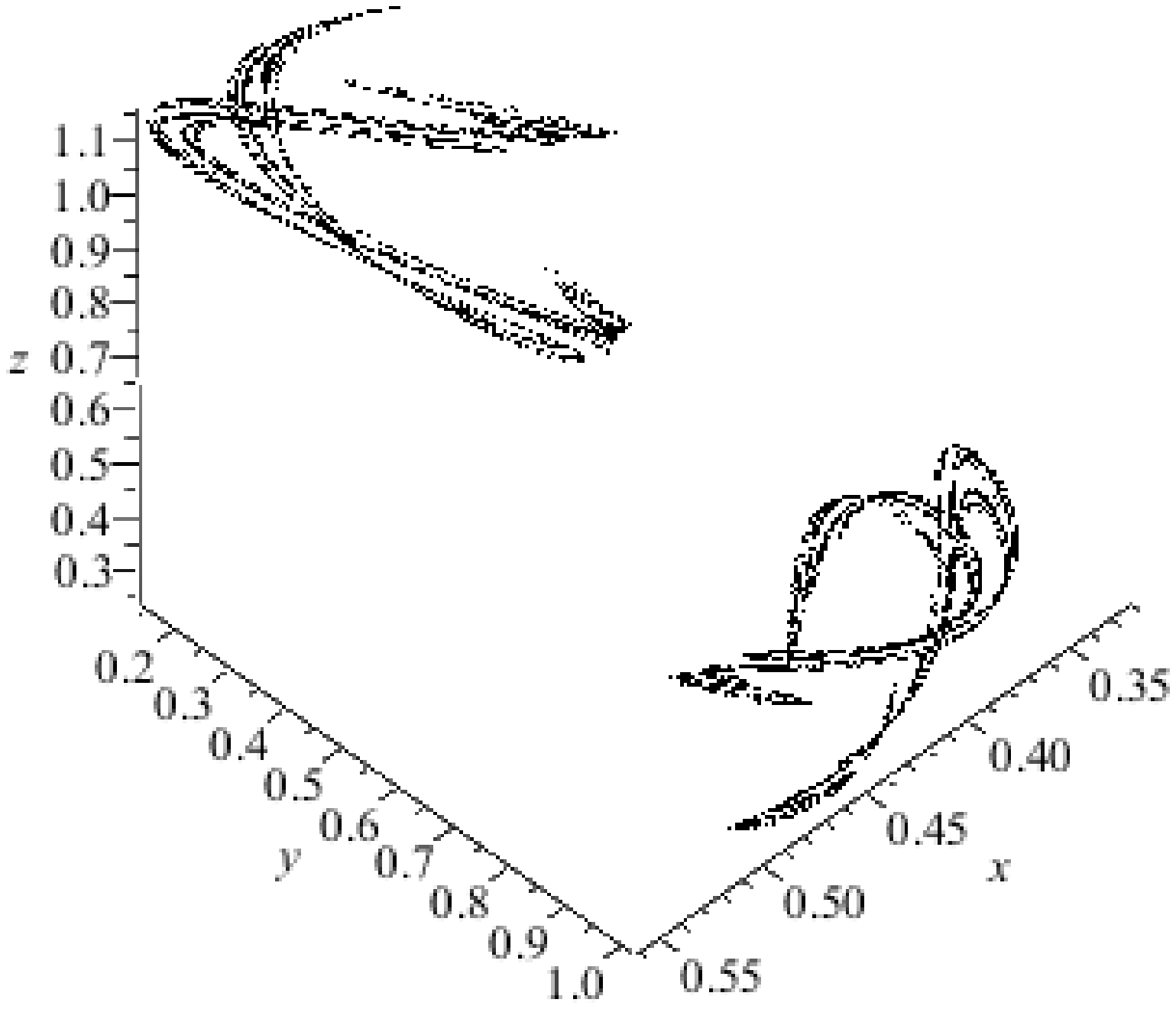}}
\renewcommand{\figure}{Fig.}
\caption{The chaotic attractor of map (33) in Section 5, where the initial value is taken as $(0.5,0.5,0.5)$.}\label{FigureMap33}
\end{center}
\end{figure}

\begin{figure}
\begin{center}

\centerline{\epsfig{figure=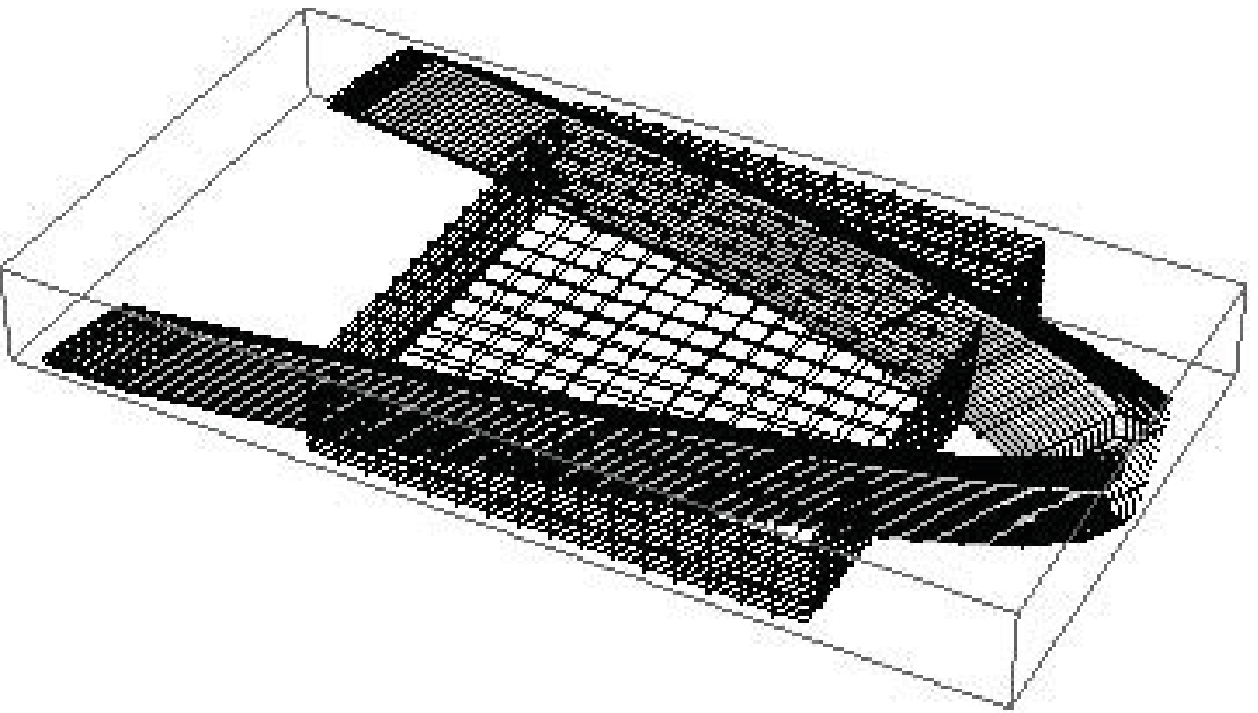,height=3in,width=4.5in}}

\caption{The illustration graph of $F(U)\cap U$ in Example \ref{illustr2}, where $U=[1,2]\times[1,2]\times[0.2,0.4]$.}\label{Case-ii-positive}
\end{center}
\end{figure}

\begin{figure}
\begin{center}

\centerline{\epsfig{figure=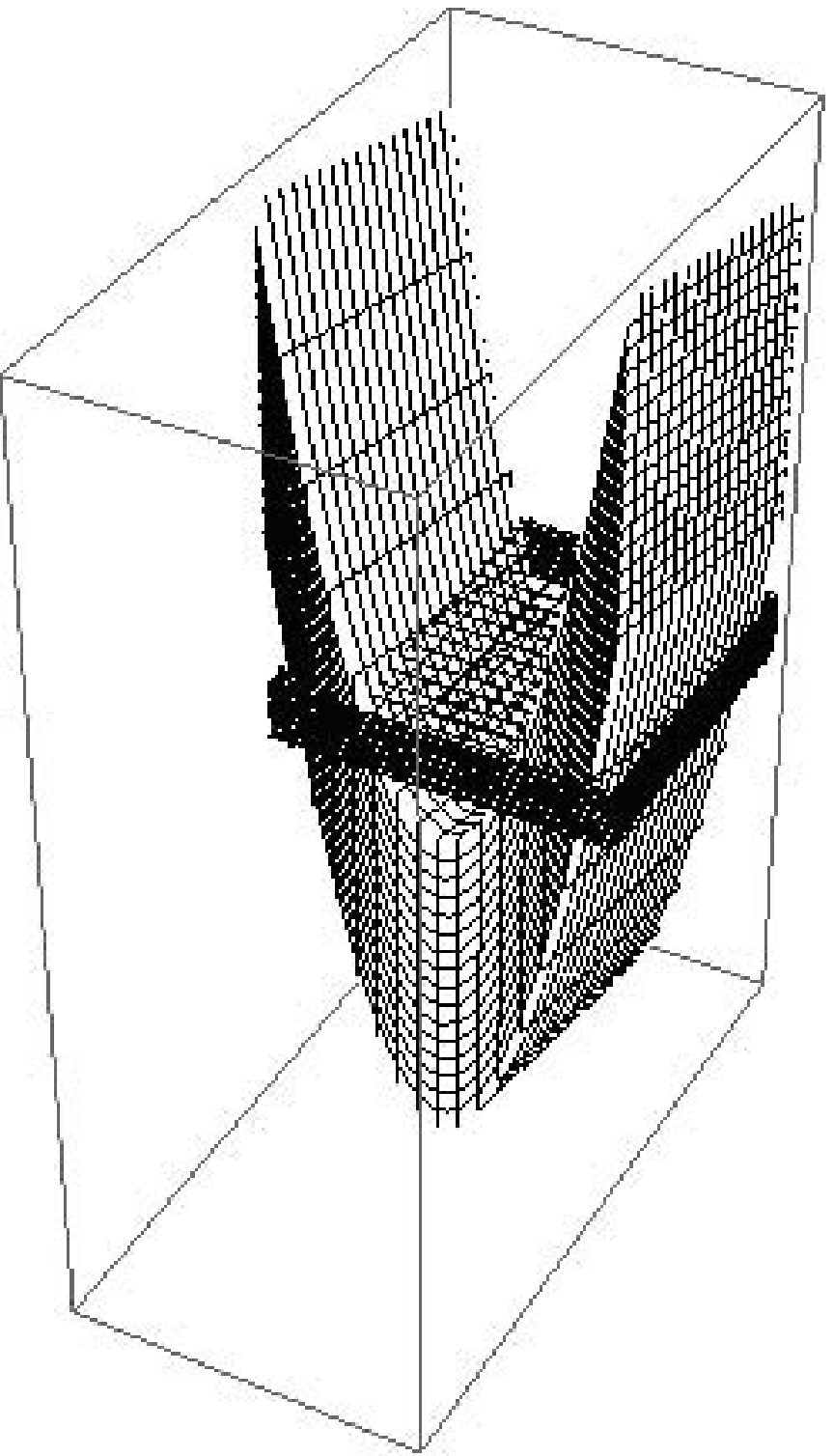,height=3in,width=4.5in}}

\caption{The illustration graph of $F^{-1}(U)\cap U$ in Example \ref{illustr2}, where $U=[1,2]\times[1,2]\times[0.2,0.4]$.}\label{Case-ii-inverse}
\end{center}
\end{figure}

\begin{figure}
\begin{center}
\centerline{\epsfig{figure=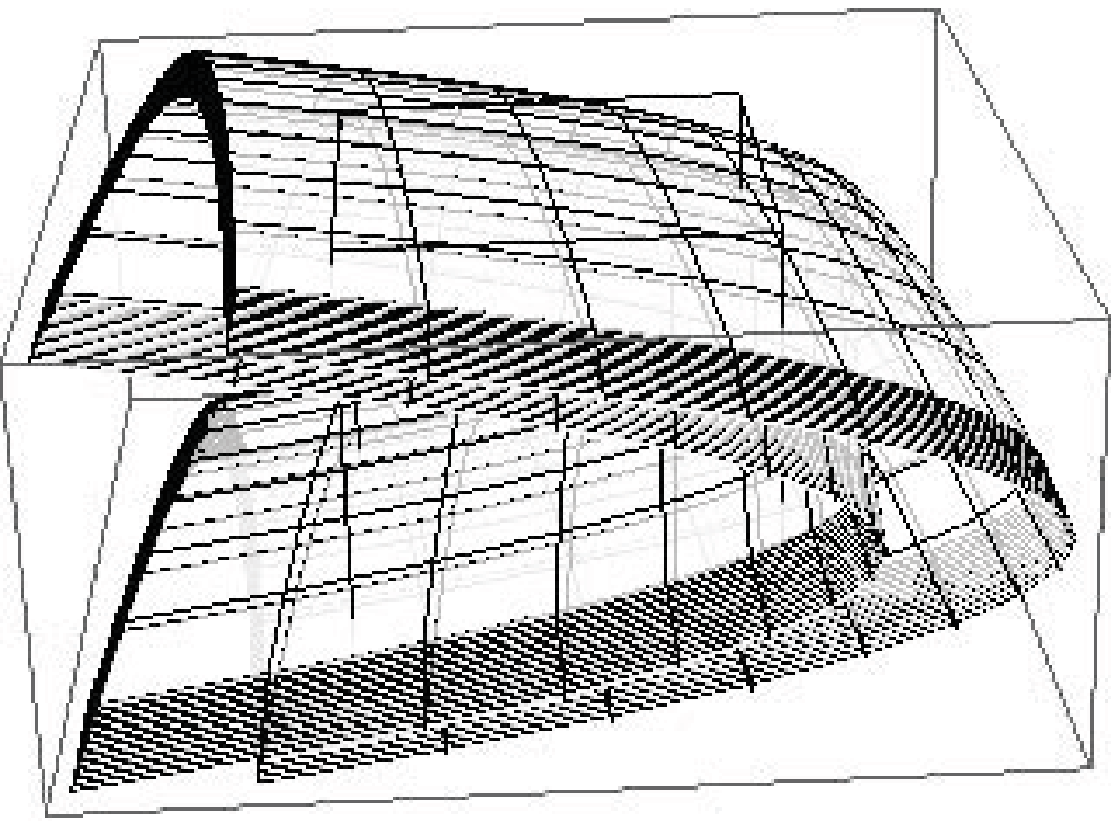,height=3in,width=4.5in}}
\caption{The illustration graph of $F(U)\cap U$ in Example \ref{illustr1}, where $U=[1,2]\times[1,2]\times[1,2]$.}\label{Caseipositive}
\end{center}
\end{figure}

\begin{figure}
\begin{center}

\centerline{\epsfig{figure=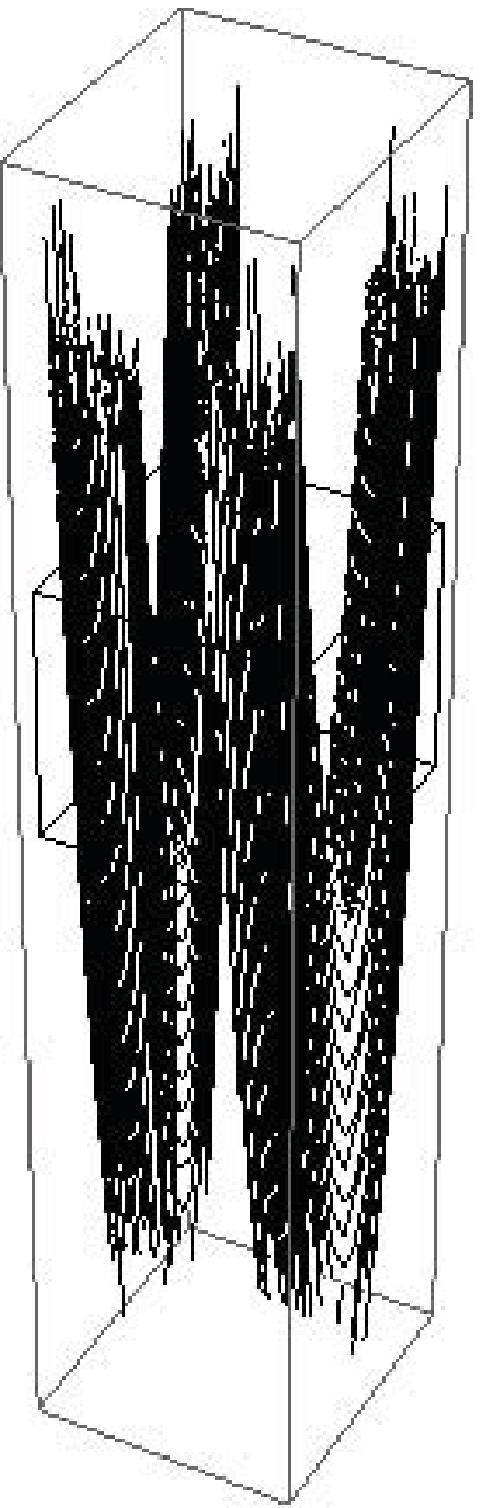,height=3in,width=3in}}

\caption{The illustration graph of $F^{-1}(U)\cap U$ in Example \ref{illustr1}, where $U=[1,2]\times[1,2]\times[1,2]$.}\label{Caseiinverse}
\end{center}
\end{figure}

\end{document}